\documentclass[preprint,12pt]{elsarticle}

\journal{Physics Letters A}

\usepackage[utf8]{inputenc} 
\usepackage[T1]{fontenc}
\usepackage{lmodern}

\usepackage{amsthm,amsmath,amsfonts,amssymb,stmaryrd,bbm}
\usepackage{mathtools}
\usepackage{enumitem}
\usepackage{mathrsfs} 

\usepackage{subcaption} 
\usepackage{placeins}
\usepackage{graphicx}
\usepackage[dvipsnames]{xcolor}
\usepackage{caption,tikz}

\usepackage{anyfontsize}

\author[bris]{Nicholas P. Baskerville}

\affiliation[bris]{organization={School of Mathematics, University of Bristol},
            addressline={Fry Building, Woodland Road}, 
            city={Bristol},
            postcode={BS8 1UG}, 
            country={United Kingdom}}

\author[oxmi]{Jonathan P. Keating}
\author[bris]{Francesco Mezzadri}
\author[bris]{Joseph Najnudel}
\author[huai]{Diego Granziol}

\affiliation[oxmi]{organization={Mathematical Institute, University of Oxford},
            addressline={Andrew Wiles Building, Woodstock Road}, 
            city={Oxford},
            postcode={OX2 6GG}, ,
            country={United Kingdom}}

\affiliation[huai]{organization={Huawei AI Research},
            addressline={1 St Pancras Square}, 
            city={London},
            postcode={ABC DEF}, ,
            country={United Kingdom}}

\title{Universal characteristics of deep neural network loss surfaces from random matrix theory}

\usepackage[english]{babel}
\usepackage[colorlinks=true, linkcolor=blue, urlcolor=black, citecolor=blue,pdfstartview=FitH]{hyperref}
\usepackage{hyperref}

\numberwithin{equation}{section}

\let\originalleft\left
\let\originalright\right
\renewcommand{\left}{\mathopen{}\mathclose\bgroup\originalleft}
\renewcommand{\right}{\aftergroup\egroup\originalright}

\newlength{\bibitemsep}
\newlength{\bibparskip}
\let\oldthebibliography\thebibliography
\renewcommand\thebibliography[1]{\oldthebibliography{#1}
  \setlength{\parskip}{\bibitemsep}
  \setlength{\itemsep}{\bibparskip}}

\DeclareMathOperator{\Tr}{Tr}

\DeclareMathOperator{\supp}{supp}

\newcommand{\muinf}{\mu_{\infty}}

\renewcommand{\epsilon}{\varepsilon}

\renewcommand{\leq}{\leqslant}
\renewcommand{\geq}{\geqslant}

\renewcommand{\P}{\mathbb{P}}
\newcommand{\E}{\mathbb{E}}
\newcommand{\R}{\mathbb{R}}
\newcommand{\C}{\mathbb{C}}
\newcommand{\N}{\mathbb{N}}

\newcommand{\1}{\mathbbm{1}}

\renewcommand{\supp}{\text{supp}}

\def\Xint#1{\mathchoice
{\XXint\displaystyle\textstyle{#1}}%
{\XXint\textstyle\scriptstyle{#1}}%
{\XXint\scriptstyle\scriptscriptstyle{#1}}%
{\XXint\scriptscriptstyle\scriptscriptstyle{#1}}%
\!\int}
\def\XXint#1#2#3{{\setbox0=\hbox{$#1{#2#3}{\int}$}
\vcenter{\hbox{$#2#3$}}\kern-.5\wd0}}

\def\dashint{\Xint-}

\newcommand{\vertiii}[1]{{\left\vert\kern-0.25ex\left\vert\kern-0.25ex\left\vert #1 
    \right\vert\kern-0.25ex\right\vert\kern-0.25ex\right\vert}}

\usepackage{bm}
\renewcommand{\vec}[1]{\bm{#1}}
\theoremstyle{plain} 
\newtheorem{thm}{Theorem}[section]

\newtheorem{lem}[thm]{Lemma}
\newtheorem{cor}[thm]{Corollary}

\newtheorem{prop}[thm]{Proposition}
\newtheorem{assump}[thm]{Assumption}

\newtheorem{rem}[thm]{Remark}

\def\author#1{\par
    {\centering{\authorfont#1}\par\vspace*{0.05in}}
}

\def\titlefont{\fontsize{13}{15}\bfseries\boldmath\selectfont\centering{}}
\def\authorfont{\fontsize{13}{15}}

\def\title#1{
    \thispagestyle{plain}
    \vspace*{-14pt}
    \vskip 79pt
    {\centering{\titlefont #1\par}}%
    \vskip 1em
}

\renewcommand{\j}{\mathfrak{j}}

\newcommand{\B}{\mathscr{B}}

\newcommand{\bX}{\bar{X}}
\newcommand{\TN}{\mathbb{T}_N}
\newcommand{\Vs}{\mathbb{V}_s}
\newcommand{\hQUE}{\widehat{\text{QUE}}}
\newcommand{\sbplain}{\mathsf{s}}
\newcommand{\sbf}{\sbplain(b)}
\newcommand{\rmes}{\texttt{r}}
\newcommand{\lmes}{\texttt{l}}

\newcommand{\barlam}{\bar{\lambda}^{(i, e)}_b}
\newcommand{\barlamnoe}{\bar{\lambda}^{(i)}_b}
\newcommand{\lam}{\lambda^{(i, j, e)}_b}

\begin{document}
\begin{abstract}
    This paper considers several aspects of random matrix universality in deep neural networks.
    Motivated by recent experimental work, we use universal properties of random matrices related to local statistics to derive practical implications for deep neural networks based on a realistic model of their Hessians.
    In particular we derive universal aspects of outliers in the spectra of deep neural networks and demonstrate the important role of random matrix local laws in popular pre-conditioning gradient descent algorithms.
    We also present insights into deep neural network loss surfaces from quite general arguments based on tools from statistical physics and random matrix theory.
\end{abstract}
\maketitle

\noindent
 
 
 


\section{Introduction}
    The success of large deep neural networks (DNNs) optimised in surprisingly na\"{i}ve ways, such as stochastic gradient descent, has spawned a considerable amount of interest in characterising their loss surfaces.
    The \emph{loss} of a neural network is simply a real valued function that measures the network's performance on a task with respect to some reference data set.
    Neural network are defined in terms of parameters, or \emph{weights}, of which there are typically a very large number.
    Fixing the data set, or even the data generating distribution, a neural network's loss can be seen as a surface in high dimensions parametrised by the network's weights.
    
    \medskip 
    It was first observed\footnote{Note that early connections between neural networks and spin systems were presented in \cite{gardner1988space}.} in a sequence of experimental and theoretical papers \cite{choromanska2015loss, sagun2014explorations} that the loss surfaces of DNNs can be connected with spherical multi-spin glasses from statistical physics \cite{mezard1987spin}.
    Mathematically, multi-spin glasses are just certain random multivariate polynomials, or equivalently Gaussian processes with specific covariance functions.
    These works were very influential, providing the first steps towards explaining why DNNs can be trained at all, given the apparent intractability of their optimisation from the viewpoint of classical optimisation theory.
    The key insight from the spin glass connection was that DNN loss surfaces are indeed very complicated and filled with many local optima, but that, in the high-dimensional limit, the local optima are arranged favourably, so that simple gradient based optimisation methods can be expected to converge to local optima close to the global minimum.
    This work has been extended in several directions.
    \cite{baskerville2021loss} extended the scope of the results to more general neural networks, and in so doing uncovered some of the limits of the spin glass model's explanatory power.
    Departing from the direct connection between glassy systems and neural networks, several authors have used similar high-dimensional random models as toy models for high-dimensional optimisation.
    Among these works are several that focus on the existence of spurious minima and the problem of recovering signals from high dimensional signal-in-noise models \cite{ros2019complex, maillard2019landscape, mannelli2019afraid}.
    Similarly, using spin glasses as surrogates for the complex, high-dimensional, random loss surfaces of DNNs, \cite{baskerville2022spin} extended the spin glass analysis to study the non-standard loss surfaces of generative adversarial networks \cite{goodfellow2014generative}.
    
    \medskip
    A significant feature common amongst all the work mentioned so far is random matrix theory \cite{mehta2004random,li2018visualizing,anderson2010introduction}.
    When studying the number and configuration of local optima of high dimensional random functions, Kac-Rice formulae are an essential mathematical tool, and random matrices arise as Hessians in these formulae.
    For spin glasses, this path was first trodden in the theoretical physics literature by Fyodorov \cite{fyodorov2004complexity, fyodorov2005counting} and made rigorous in later work \cite{auffinger2013random}.
    Specifically in the case of \cite{choromanska2015loss} and \cite{auffinger2013random}, which lays its mathematical foundation, very detailed calculations can be completed.
    This is possible because the random matrices that appear belong to the Gaussian Orthogonal Ensemble (GOE), one of the canonical ensembles of random matrix theory.
    Almost anything one could want to know of the GOE is known and many powerful tools of random matrix theory can be applied to it.
    In particular, the full joint distribution of its eigenvalues and large deviations results for its eigenvalues and its spectral density, all of which are required by \cite{auffinger2013random, choromanska2015loss}.
    Works such as \cite{baskerville2021loss, baskerville2022spin} show how detailed calculations can be completed beyond the standard spin glass case, however these results still depend on important properties of the GOE, as the Hessians in those cases are closely related to the GOE.
    In a recent work, \cite{granziol2020learning} showed how valuable practical insights about DNN optimisation can be obtained by considering the outliers in the spectrum on the loss surface Hessian.
    Once again, this work relies on special properties from random matrix theory, namely the Wigner semicircle law for Wigner matrices and outlier results for rotationally invariant matrices \cite{benaych2011eigenvalues}.
    
    \medskip
    Challenging the above-mentioned works, an experimental line of work has demonstrated convincingly that special RMT ensembles like the GOE do not appear to be present in DNNs \cite{papyan2018full, granziol2020beyond, baskerville2022appearance}, for example as their Hessians.
    In addition, there have been challenges in the literature to the practical relevance of spin glass loss surface results for DNNs \cite{baity2019comparing}.
    In this context, it is natural to question the relevance to DNNs of many of the results discussed above.
    Indeed, does random matrix theory actually provide insight into real DNNs, or are its powerful tools merely applicable to toy models too divorced from real DNNs to be of any value?
    Random matrix theory itself may hold the answer to this question, in particular its notion of \emph{universality}.
    Broadly speaking, universality refers to the phenomenon that certain properties of special random matrix ensembles (such as the GOE) remain true for more general random matrices that share some key feature with the special ensembles.
    For example, the Wigner semicircle is the limiting spectral density of the Gaussian Wigner ensembles, i.e. matrices with Gaussian entries, independent up two symmetry (symmetric real matrices, Hermitian complex matrices) \cite{mehta2004random}.
    The Gaussian case is the simplest to prove, and there are various powerful tools not available in the non-Gaussian case, however the Wigner semicircle has been established as the limiting spectral density for Wigner matrices with quite general distributions on their entries \cite{anderson2010introduction,tao2012topics}.
    While surprisingly general is some sense, the Wigner semicircle relies on independence (up to symmetry) of matrix entries, a condition which is not typically satisfied in real systems.
    The limiting form of the spectral density of a random matrix ensemble is a \emph{macroscopic} property, i.e. the matrix is normalised such that the average distance between adjacent eigenvalues is on the order of $1/\sqrt{N}$, where $N$ is the matrix size.
    At the opposite end of the scale is the \emph{microscopic}, where the normalisation is such that eigenvalues are spaced on a scale of order $1$; at this scale, random matrices display a remarkable universality.
    For example, any real symmetric matrix has a set of orthonormal eigenvectors and so the set of all real symmetric matrices is closed under conjugation by orthogonal matrices.
    Wigner conjectured that certain properties of GOE matrices hold for very general random matrices that share the same (orthogonal) symmetry class, namely symmetric random matrices (the same is true of Hermitian random matrices and the unitary symmetry class).
    The spacings between adjacent eigenvalues should follow a certain explicit distribution, the Wigner surmise, and the eigenvectors should be \emph{delocalised}, i.e. the entries should all be of the same order as the matrix size grows.
    Both of these properties are true for the GOE and can be proved straightforwardly with quite elementary techniques.
    These properties and more have been proved in a series of works over the last decade or-so, of which a good review is \cite{erdos2017dynamical}.
    Crucial in these results is the notion of a \emph{local law} for random matrices.
    The technical statement of local laws is given later in the paper, but roughly they assert that the spectrum of a random matrix is, with very high probability, close to the deterministic spectrum defined by its limiting spectral density (e.g. the semicircle law for Wigner matrices).
    Techniques vary by ensemble, but generally a local law for a random matrix ensemble provides the control required to demonstrate that certain matrix statistics are essentially invariant under the evolution of the Dyson Brownian motion.
    In the case of real symmetric matrices, the Dyson Brownian motion converges in finite time to the GOE, hence the statistics preserved under the Dyson Brownian motion must match the GOE.
    The $n$-point correlation functions of eigenvalues are one such preserved quantity, from which follows, amongst other properties, the Wigner surmise on adjacent spacings.
    
    \medskip
    At the macroscopic scale, there are results relevant to neural networks, for example \cite{pennington2018emergence,pastur2020random} consider random neural networks with Gaussian weights and establish results that are generalised to arbitrary distributions with optimal conditions, so demonstrating universality.
    On the microscopic scale, \cite{baskerville2022appearance} provided the first experimental demonstration of the presence of universal local random matrix statistics deep neural networks, specifically in the Hessians and Gauss-Newton matrices of their loss surfaces.
    This work has recently been extended to the weight matrices of neural networks \cite{thamm2022random}.
    
    \medskip
    This paper explores the consequences of random matrix universality in deep neural networks.
    Our main mathematical result is a significant generalisation of the Hessian spectral outlier result recently presented by \cite{granziol2020learning}.
    This generalisation removes any need for GOE or Wigner forms of the Hessian and instead leverages much more universal properties of the eigenvectors and eigenvalues of random matrices which we argue are quite likely to hold for real networks.
    Our results make concrete predictions about the outliers of DNN Hessians which we compare with experiments on several real-world DNNs.
    These experiments provide indirect evidence of the presence of universal random matrix statistics in the Hessians of large DNNs, which is noteworthy as certainly these DNNs are far too large to permit exact eigendecomposition of their Hessians as in \cite{baskerville2022appearance}.
    Along a similar line, we show how local random matrix laws in DNNs can dramatically simplify the dynamics of certain gradient descent optimisation algorithms and may be in part responsible for their success.
    Finally, we highlight another aspect of random matrix universality relevant to DNN loss surfaces.
    Recent work \cite{arous2021exponential} has shown that the so-called `self averaging' property of random matrix determinants is very much more universal than previously thought.
    The self-averaging of random matrix determinants has been used in the spin glass literature both rigorously and non-rigorously (e.g. \cite{fyodorov2004complexity,fyodorov2005counting,auffinger2013random,baskerville2021loss,baskerville2022spin} inter alia) and is the key property that produces the exponentially large/small number of local optima repeatedly observed.
    We argue that insights into the geometry of DNN loss surfaces can be conjectured from quite general assumptions about the Hessian and gradient noise and from the general self-averaging effect of random matrix determinants.
    
    \medskip
    The paper is structured as follows. Section \ref{sec:hess_model} introduces our random matrix theory model for DNN Hessians, derives results for their outliers and compares with experimental results. Section \ref{sec:que} presents the proof of our main result for addition of random matrices, combining quantum unique ergodicity and the supersymmetric method. Section \ref{sec:determinants} proves an extension of the random matrix determinant results of \cite{arous2021exponential} and presents insights into DNN loss surfaces from complexity calculations and random matrix determinants. Section \ref{sec:precond} describes the role of random matrix local laws in certain popular DNN optimisation algorithms. Section \ref{sec:concl} concludes the paper.

    \paragraph{Notation} We adopt the following conventions throughout
    \begin{itemize}
        \item For a probability measure $\mu$, $g_{\mu}$ is its Stieljtes transform.
        \item $\mu\boxplus\nu$ denotes the free additive convolution between probability measures $\mu$ and  $\nu$.
        \item Hats denote empirical quantities unless stated otherwise. For example $\hat{\mu}_X$ is the empirical spectral measure of a matrix $X$.
        \item $\rmes(\mu), \lmes(\mu)$ denote the right and left edges of the support of a probability measure $\mu$ respectively.
    \end{itemize}
\section{General random matrix model for loss surface Hessians}\label{sec:hess_model}

\subsection{The model}\label{subsec:hess_model}
Given a loss function $\mathcal{L}: \mathscr{Y}\times \mathscr{Y}\rightarrow\R$, a data generating distribution $\P_{\text{data}}$ supported on $\mathscr{X}\times \mathscr{Y}$ and a neural network $f_{\vec{w}}: \mathscr{X}\rightarrow\mathscr{Y}$ parametrised by $\vec{w}\in\mathbb{R}^N$, its batch Hessian is given by
\begin{align}\label{eq:batch_hessian_def}
    H_{\text{batch}} = \frac{1}{b}\sum_{i=1}^b \frac{\partial^2}{\partial \vec{w}^2} \mathcal{L}(f_{\vec{w}}(\vec{x_i}), y_i), ~~ (\vec{x}_i, y_i)\overset{\text{i.i.d.}}{\sim}\mathbb{P}_{\text{data}}
\end{align}
and its true Hessian is given by
\begin{align}
    H_{\text{true}} = \E_{(\vec{x}, y)\sim\mathbb{P}_{\text{data}}} \frac{\partial^2}{\partial \vec{w}^2}  \mathcal{L}(f_{\vec{w}}(\vec{x}), y).
\end{align}
Both  $H_{\text{batch}}$ and $H_{\text{true}}$ are $N\times N$ matrix functions of $\vec{w}$; $H_{\text{batch}}$ is random but $H_{\text{true}}$ is deterministic. Only in very specific cases and under strong simplifying assumptions can one hope to obtain the distribution of $H_{\text{batch}}$ or the value of $H_{\text{true}}$ from $\mathcal{L}, \P_{\text{data}}$ and $f_{\vec{w}}$. Inspired by the success of many random matrix theory applications, e.g. in Physics, we will instead seek to capture the essential features of deep neural network Hessians in a sufficiently general random matrix model.

\medskip
We introduce the following objects:
\begin{itemize}
    \item A sequence (in $N$) of random real symmetric $N\times N$ matrices $X$. $X$ possesses a limiting spectral probability measure $\mu$, i.e. if $\lambda_1,\ldots,\lambda_N$ are the eigenvalues of $X$ then
    \begin{align}
        \frac{1}{N}\sum_{i=1}^N \delta_{\lambda_i} \rightarrow \mu
    \end{align}
    weakly almost surely. We further assume that $\mu$ has compact support and admits and smooth density with respect to Lebesgue measure.
    
    \item A sequence (in $N$) of deterministic real symmetric $N\times N$ matrices $A$ with eigenvalues
    \begin{align}
        \theta_1,\ldots,\theta_p,\xi_1,\ldots \xi_{N-p-q},\theta_1',\ldots,\theta_q'
    \end{align}
    for fixed integers $p, q$. We assume the existence of limiting measure $\nu$ such that, weakly, \begin{align}
        \frac{1}{N-p-q}\sum_{i=1}^{N-p-q} \delta_{\xi_i} \rightarrow \nu
    \end{align}
    where $\nu$ is a compactly supported probability measure. The remaining eigenvalues satisfy 
    \begin{align}
        \theta_1 >\ldots > \theta_p > \rmes(\nu), ~ \theta_1' < \ldots < \theta_q' < \lmes(\nu).
    \end{align}
    $\nu$ is also assumed to be of the form $\nu = \epsilon\eta + (1-\eta)\delta_0$ where $\eta$ is a compactly supported probability measure which admits a density with respect to Lebesgue measure.
    \item A decreasing function $\sbplain: \N \rightarrow (0, 1)$.
\end{itemize}
With these definitions, we construct the following model for the Hessian:
\begin{align}\label{eq:hessian_def}
    H_{\text{batch}}\equiv H = \sbf X + A
\end{align}
where $b$ is the batch size. We have dropped the subscript on $H_{\text{batch}}$ for brevity. Note that $H$ takes the place of the batch Hessian and $A$ taken the place of the true Hessian. $\sbf X$ takes the place of the random noise introduced by sampling a finite batch at which to evaluate the Hessian. $\sbf$ is an overall scaling induced in $X$ by the batch-wise averaging.

\medskip
This model is almost completely general. Note that we allow the distribution of $X$ and the value of $A$ to depend on the position in weight space $\vec{w}$. The only restrictions imposed by the model are
\begin{enumerate}
    \item the existence of $\nu$;
    \item the position of $\theta_i, \theta_j'$ relative to the support of $\nu$;
    \item $\nu$ may only possess an atom at $0$;
    \item the fixed number of  $\theta_i, \theta_j'$;
    \item the existence of $\mu$;
    \item the existence of the scaling $\sbf$ in batch size.
\end{enumerate}
All of the above restrictions are discussed later in the section. Finally, we must introduce some properties of the noise model $X$ in order to make any progress. We introduce the assumption that the eigenvectors of $X$ obey \emph{quantum unique ergodicity} (QUE). The precise meaning of this assumption and a thorough justification and motivation is given later in this section. For now it suffices to say that QUE roughly means that the eigenvectors of $X$ are \emph{delocalised} or that they behave roughly like the rows (or columns) of a uniform random $N\times N$ orthogonal matrix (i.e. a matrix with Haar measure). QUE is known to hold for standard ensembles in random matrix theory, such as quite general Wigner matrices, Wishart matrices, adjacency matrices of certain random graphs etc. Moreover, as discussed further section \ref{subsec:que_just} below, QUE can be thought of as a property of quite general random matrix models.

\subsection{Quantum unique ergodicity}
It is well known that the eigenvectors of quite general random matrices display a universal property of \emph{delocalisation}, namely
\begin{align}
    |u_k|^2 \sim \frac{1}{N}
\end{align}
for any component $u_k$ of an eigenvector $\vec{u}$.
Universal delocalisation was conjectured by Wigner along with the Wigner surmise for adjacent eigenvalue spacing.
Both of these properties, and the more phenomenon of universal correlation functions on the microscopic scale have since been rigorously established for quite a variety of matrix models e.g. \cite{erdos2017dynamical, erdHos2012universality,erdHos2019random}.
\cite{bourgade2017eigenvector} show that the eigenvectors of generalised Wigner matrices obey \emph{Quantum unique ergodicity}, a particular form of delocalisiation, stronger than the above statement. Specifically, they are shown to be approximately Gaussian in the following sense. (\cite{bourgade2017eigenvector} Theorem 1.2)
\begin{align}\label{eq:que_def}
    \sup_{||\vec{q}|| = 1}\sup_{\substack{I\subset [N],\\ |I| = n}} \left|\E P\left(\left(N|\vec{q}^T\vec{u}_k|^2\right)_{k\in I}\right) - \E P\left(\left(|\mathcal{N}_j|^2\right)_{j=1}^m
    \right)\right| \leq N^{-\epsilon}
\end{align}
for large enough $N$, where $\mathcal{N}_j$ are i.i.d. standard normal random variables, $(\vec{u}_k)_{k=1}^N$ are the normalised eigenvectors, $P$ is any polynomial in $n$ variables and $\epsilon > 0$.

\subsection{Batch Hessian outliers}
Let $\{\lambda_i\}$ be the eigenvalue of $H$. To set the context of our results, let us first simplify and suppose momentarily that $\sbplain = 1$ and, instead of mere QUE, $X$ has eigenvectors distributed with Haar measure, and $A$ is fixed rank, i.e. $\xi_i=0 ~\forall i$, then the results of \cite{benaych2011eigenvalues} would apply and give \begin{align}
    \lambda_j \overset{a.s.}{\rightarrow} \begin{cases}
                    g_{\mu}^{-1}(1/\theta_j) ~~&\text{if } \theta_j > 1/g_{\mu}(\rmes(\mu)),\\
                    \rmes(\mu) &\text{otherwise},
                \end{cases}\label{eq:bgn_outlier_1}\\
    \lambda_{N-j+1} \overset{a.s.}{\rightarrow} \begin{cases}
                    g_{\mu}^{-1}(1/\theta_j') ~~&\text{if } \theta_j' < 1/g_{\mu}(\lmes(\mu)),\\
                    \lmes(\mu) &\text{otherwise},
                \end{cases}\label{eq:bgn_outlier_2}                
\end{align}

\paragraph{\textbf{An interlude on prior outlier results}} {\em It was conjectured in \cite{benaych2011eigenvalues} that (\ref{eq:bgn_outlier_1})-(\ref{eq:bgn_outlier_2}) still hold when $X$ has delocalised eigenvectors in some sense, rather than strictly Haar. Indeed, a careful consideration of the proof in that work does reveal that something weaker than Haar would suffice, for example QUE. See in particular the proof of the critical Lemma 9.2 therein which can clearly be repeated using QUE. There is a considerable subtlety here, however, which is revealed best by considering more recent results on deformations of general Wigner matrices. \cite{knowles2017anisotropic} shows that very general deterministic deformations of general Wigner matrices possess an optimal anisotropic local law, i.e. $Y + B$ for Wigner $Y$ and deterministic symmetric $B$. It is expected therefore that $Y+B$ has delocalised eigenvectors in the bulk. Consider the case where $B$ is diagonal, and say that $B$ has a fixed number of ``spike'' eigenvalues $\varphi_1>\ldots>\varphi_r$ and remaining eigenvalues $\zeta_1,\ldots, \zeta_{N-r}$ where the empirical measure of the $\zeta_i$ converges to some measure $\tau$ and $\varphi_r > \rmes(\tau)$. We can then split $B = B_{i} + B_o$ where $B_i$ contains only the $\zeta_j$ and $B_o$ only the $\varphi_j$. The previously mentioned results applies to $Y + B_i$ and then we might expect the generalised result of \cite{benaych2011eigenvalues} to apply to give outliers $g_{\mu_{SC}\boxplus \tau}^{-1}(1/\varphi_i)$ of $Y + B$. This contradicts, however, another result concerning precisely the the outliers of such generally deformed Wigner matrices. It was shown in \cite{capitaine2016spectrum} that the outliers of $Y+B$ are $\omega^{-1}(\varphi_j)$ where $\omega$ is the subordination function such that $g_{\mu_{SC} \boxplus \tau}(z) = g_{\tau}(\omega(z))$. These two expressions coincide when}
\begin{align}
    &\omega^{-1}(z) = g_{\mu_{SC}\boxplus\tau}^{-1}(z^{-1})\notag\\
    \iff &\omega^{-1}(z) = \omega^{-1}(g_{\tau}^{-1}(z^{-1}))\notag\\
    \iff &g_{\tau}^{-1}(z^{-1}) = z\notag\\
    \iff &g_{\tau}(z) = z^{-1}\notag\\
    \iff &\tau = \delta_0,
\end{align}
{\em i.e. only when $B$ is in fact of negligible rank as $N\rightarrow\infty$. This apparent contradiction is resolved by the observation that the proof in \cite{benaych2011eigenvalues} in fact relies implicitly on an \emph{isotropic local law}. Note in particular section 4.1, which translated to our context, would require $\vec{v}^TG_{Y + B_i}(z)\vec{v}\approx g_{\mu_{SC}\boxplus\tau}(z)$ with high probability for general unit vectors $\vec{v}$. Such a result holds if and only if $Y + B_i$ obeys an isotropic local law and is violated if its local law is instead anistropic, as indeed it is, thanks to the deformation.}

\medskip
Returning to our main thread, we have satisfied the conditions on $X$ required to invoke Theorem \ref{thm:nearly_free_addn} and so conclude that
\begin{align}
    \hat{g}_{H}(z) = g_{\mu_b \boxplus \nu}(z) + o(1) = g_{\nu}(\omega(z)) + o(1) = \hat{g}_{A}(\omega(z)) + o(1)
\end{align}
where $\omega$ is the subordination function such that $g_{\mu_b\boxplus\nu}(z) = g_{\nu}(\omega(z))$ and $\mu_b$ is the limiting spectral measure of $\sbf X$. The reasoning found in \cite{capitaine2016spectrum} then applies regarding the outliers of $H$. Indeed, suppose that $\lambda$ is an outlier of $H$, i.e. $\lambda$ is an eigenvalue of $H$ contained in $\R \backslash \text{supp}(\mu\boxplus\nu)$. Necessarily $\hat{g}_H$ possesses a singularity at $\lambda$, and so $\hat{g}_A$ must have a singularity at $\omega(\lambda)$. For this singularity to persist for all $N$,ppp $\omega(\lambda)$ must coincide with one of the outliers of $A$ which, unlike the bulk eigenvalues $\xi_j$, remain fixed for all $N$. Therefore we have the following expressions for the outliers of $H$:
\begin{align}\label{eq:outlier_sets}
    \{\omega^{-1}(\theta_j) \mid \omega^{-1}(\theta_j)\in \R \backslash\text{supp}(\mu_b\boxplus\nu)\}\cup \{\omega^{-1}(\theta_j') \mid \omega^{-1}(\theta_j')\in \R \backslash\text{supp}(\mu_b\boxplus\nu)\}.
\end{align}

We now consider $\epsilon$ to be small and analyse these outlier locations as a perturbation in $\epsilon$. Firstly note that 
\begin{align}
    g_{\mu_b}(z) = \int \frac{d\mu_b(x)}{z - x} = \int \frac{d\mu(x/\sbf)}{z - x} = \sbf\int \frac{d\mu(x)}{z - \sbf x} = g_{\mu}(z/\sbf).
\end{align}
Also
\begin{align}
    \omega^{-1}(z) &= g_{\mu_b \boxplus \nu}^{-1}(g_{\nu}(z))\label{eq:omega_inv_def}\\
    &= R_{\mu_b}(g_{\nu}(z)) + g_{\nu}^{-1}(g_{\nu}(z))\notag\\
    &= R_{\mu_b}(g_{\nu}(z)) + z.\label{eq:omega_prog1}
\end{align}
We now must take care in computing $R_{\mu_b}$ from $g_{\mu_b}$.
Recall that the $R$-transform of a measure is defined as a formal power series \cite{anderson2010introduction}
\begin{align}
   R(z) = \sum_{n=0}^{\infty} k_{n+1} z^n
\end{align}
where $k_n$ is the $n$-th cumulant of the measure.
It is known \cite{anderson2010introduction} that $k_n=C_n$ where the functional inverse of the Stieljtes transform of the measure is given by the formal power series
\begin{align}
    K(z) = \frac{1}{z} + \sum_{n=1} C_n z^{n-1}.
\end{align}
Now let $m_n$ be the $n$-th moment of $\mu$ and similarly let $m_n^{(b)}$ be the $n$-th moment of $\mu_b$, so formally
\begin{align}
    g_{\mu}(z) = \sum_{n\geq 0} m_n z^{-(n+1)}, ~~ g_{\mu_b}(z) = \sum_{n\geq 0} m_n^{(b)} z^{-(n+1)}.
\end{align}
Also let $k_n$ be the $n$-th cumulant of $\mu$ and $k_n^{(b)}$ be the $n$-th cumulant of $\mu_b$.
Referring to the proof of Lemma 5.3.24 in \cite{anderson2010introduction} we find \begin{align}
    m_n &= \sum_{r=1}^n \sum_{\substack{0\leq i_1,\ldots, i_r\leq n-r \\ i_1+\ldots + i_r = n-r}} k_r m_{i_1}\ldots m_{i_r},\\
    m_n^{(b)} &= \sum_{r=1}^n \sum_{\substack{0\leq i_1,\ldots, i_r\leq n-r \\ i_1+\ldots + i_r = n-r}} k_r^{(b)} m_{i_1}^{(b)}\ldots m_{i_r}^{(b)}.
\end{align}
But clearly the moments of $\mu_b$ have a simple scaling in $\sbf$, namely $m_n^{(b)} = \sbf^{n} m_n$, hence
\begin{align}
    m_n = \sbf^{-n}\sum_{r=1}^n \sum_{\substack{0\leq i_1,\ldots, i_r\leq n-r \\ i_1+\ldots + i_r = n-r}} k_r^{(b)} m_{i_1}\ldots m_{i_r} \sbf^{n-r}
\end{align}
from which we deduce $k_n^{(b)} = \sbf^n k_n$, which establishes that $R_{\mu_b}(z) = \sbf R_{\mu}(\sbf z)$.
Recalling (\ref{eq:omega_prog1}) we find
\begin{align}
    \omega^{-1}(z) = \sbf R_{\mu} (\sbf g_{\nu}(z)) + z.
\end{align}
The form of $\nu$ gives \begin{align}
    g_{\nu}(z) = \frac{1-\epsilon}{z} + \epsilon g_{\eta}(z) = \frac{1}{z} + \epsilon\left(g_{\eta}(z) - \frac{1}{z}\right)
\end{align}
and so we can expand to give 
\begin{align}
    \omega^{-1}(\theta_j) &= \theta_j + \sbf R_{\mu}(\sbf\theta_j^{-1}) + \epsilon \sbf^2 \left(g_{\eta}(\theta_j) - \theta_j^{-1}\right)R_{\mu}'(\sbf\theta_j^{-1}) + \mathcal{O}(\epsilon^2)\notag\\
    &= \theta_j + \sbf R_{\mu}(\sbf\theta_j^{-1}) + \epsilon \sbf^2 d_{\eta}(\theta_j)R_{\mu}'(\sbf\theta_j^{-1}) + \mathcal{O}(\epsilon^2)\label{eq:omega_inv_used_form}
\end{align}
where we have defined $d_{\eta}(z) = g_{\eta}(\theta_j) - \theta_j^{-1}$.
Note that (\ref{eq:omega_inv_used_form}) reduces to outliers of the form $\theta_j + \sbf^2 R_{\mu}(\theta_j^{-1})$ if $\epsilon=0$ or $d_{\eta} = 0$, as expected from \cite{benaych2011eigenvalues}\footnote{Note that $d_{\eta}=0 \iff \eta = \delta_0$ which is clearly equivalent (in terms of $\nu$) to $\epsilon=0$.}.


\medskip
The problem of determining the support of $\mu_b\boxplus\nu$ is difficult and almost certainly analytically intractable, with \cite{bao2020support} containing the most advanced results in that direction.
However overall, we have a model for deep neural network Hessians with a spectrum consisting, with high-probability, of a compactly supported bulk $\mu_b\boxplus\nu$ and a set of outliers given by (\ref{eq:omega_inv_used_form}) (and similarly for $\theta_j'$) subject to (\ref{eq:outlier_sets}).

\medskip
(\ref{eq:omega_inv_used_form}) is a generalised form of the result used in \cite{granziol2020learning}. We have the power series
\begin{align}
    R_{\mu}(\sbf \theta_j^{-1}) &= k_1^{(\mu)} + \frac{k_2^{(\mu)}\sbf}{\theta_j} + \frac{k_3^{(\mu)}\sbf^2}{\theta_j^2} + \ldots,\\
    d_{\eta}(\theta_j) &= \frac{m_1^{(\eta)}}{\theta_j^2} + \frac{m_2^{(\eta)}}{\theta_j^3} + \ldots
\end{align}
where $m_n^{(\eta)}$ are the moments of $\eta$ and $k_n^{(\mu)}$ are the cumulants of $\mu$. In the case that the spikes $\theta_j$ are large enough, we approximate by truncating these power series to give 
\begin{align}
    \omega^{-1}(\theta_j) &\approx \theta_j + \sbf m_1^{(\mu)} + \sbf^2k_2^{(\mu)}\left(\frac{1}{\theta_j} + \frac{\epsilon m_1^{(\eta)}}{\theta_j^2}\right)
\end{align}
where the approximation is more precise for larger $b$ and smaller $\epsilon$ and we have used the fact that the first cumulant of any measure matches the first moment.
One could consider for instance a power law for $\sbf$, i.e.
\begin{align}
    \omega^{-1}(\theta_j) &\approx \theta_j +  \frac{k_1^{(\mu)}}{b^{\upsilon}} + \frac{k_2^{(\mu)}}{b^{2\upsilon}}\left(\frac{1}{\theta_j} + \frac{\epsilon m_1^{(\eta)}}{\theta_j^2}\right) =\theta_j +  \frac{m_1^{(\mu)}}{b^{\upsilon}} + \frac{k_2^{(\mu)}}{b^{2\upsilon}}\left(\frac{1}{\theta_j} + \frac{\epsilon m_1^{(\eta)}}{\theta_j^2}\right) \label{eq:omega_inv_approx_form}
\end{align}
for some $\upsilon > 0$.
In the case that $\mu$ is a semicircle, then all cumulants apart from the second vanish, so setting $\epsilon = 0$ recovers \emph{exactly} \begin{align}
    \omega^{-1}(\theta_j) = \theta_j + \frac{\sigma^2}{4b^{2\upsilon}\theta_j}
\end{align}
where $\sigma$ is the radius of the semicircle.
To make the link with \cite{granziol2020learning} obvious, we can take $\upsilon = 1/2$ and $\mu$ to be the semicircle, so giving
\begin{align}
    \omega^{-1}(\theta_j) \approx \theta_j + \frac{\sigma^2}{4b\theta_j}
\end{align}
where we have truncated $\mathcal{O}(\epsilon)$ term.
We present an argument in favour of the $\upsilon=1/2$ power law below, but we allow for general $\upsilon$ when comparing to experimental data.

\begin{rem}
It is quite possible for $\mu$'s density to have a sharp spike at the origin, or even for $\mu$ to contain a $\delta$ atom at $0$, as observed empirically in the spectra of deep neural network Hessians.
\end{rem}

\subsection{Experimental results} 
The random matrix Hessian model introduced above is quite general and abstract.
Necessarily the measures $\mu$ and $\eta$ must be allowed to be quite general as it is well established experimentally \cite{papyan2018full, granziol2020beyond, baskerville2022appearance} that real-world deep neural network Hessians have spectral bulks that are not familiar as being any standard canonical examples from random matrix theory.
That being said, the approximate form in (\ref{eq:omega_inv_approx_form}) gives quite a specific form for the Hessian outliers.
In particular, the constants $m_1^{(\mu)}, m_1^{(\eta)}$ and $m_2^{(\mu)}, \epsilon > 0$ are shared between all outliers at all batch sizes.
If the form of the Hessian outliers seen in (\ref{eq:omega_inv_approx_form}) is not observed experimentally, it would suggest at least one of the following does not hold:
\begin{enumerate}
    \item batch sampling induces a simple multiplicative scaling on the Hessian noise (\ref{eq:hessian_def});
    \item the true Hessian is approximately low-rank (as measured by $\epsilon$) and has a finite number of outliers;
    \item the Hessian noise model $X$ has QUE.
\end{enumerate}
In view of this third point, agreement with (\ref{eq:omega_inv_approx_form}) provides an indirect test for the presence of universal random matrix statistics in deep neural network Hessians. 

\medskip
We can use Lanczos power methods \cite{meurant2006lanczos} to compute good approximations to the top few outliers in the batch Hessian spectra of deep neural networks \cite{granziol2020learning}.
Indeed the so-called Pearlmutter trick \cite{pearlmutter1994fast} enable efficient numerical computation of Hessian-vector products, which is all that one requires for power methods.
Over a range of batch sizes, we compute the top 5 outliers of the batch Hessian for 10 different batch seeds.
We repeat this procedure at every 25 epochs throughout the training of VGG16 and WideResNet$28\times10$ image classifiers on the CIFAR100 dataset and at every epoch during the training of a simple multi-layer perceptron network on the MNIST dataset.
By the end of training each of the models have high test accuracy, specifically the VGG$16$ architecture which does not use batch normalisation, has a test accuracy of $\approx 75\%$, whereas the WideResNet$28\times10$ has a test accuracy of $\approx 80 \%$. The MLP has a test set accuracy of $\approx 95\%$.

\begin{rem}
There is a subtlety with regard to obtaining the top outliers using the Lanczos power method.
Indeed, since Lanczos provides, in some sense, an approximation to the whole spectrum of a matrix, truncating at $m$ iterations for a $N\times N$ matrix cannot produce good approximations to all of the $m$ top eigenvalues.
In reality, experimental results \cite{papyan2018full,granziol2019deep} show that, for deep neural networks, and using sufficiently many iterations ($m$), the top $r$ eigenvalues may be recovered, for $r\ll m$.
We display some spectral plots of the full Lanczos results in the Appendix which demonstrate clearly a large number of outliers, and clearly more than $5$.
As a result, we can have confidence that our numerical procedure is indeed recovering approximations to the top few eigenvalues required for our experiments.
\end{rem}

\medskip
Let $\lambda^{(i, j, e)}_b$ be the top $i$-th empirical outlier (so $i=1$ is the top outlier) for the $j$-th batch seed and a batch size of $b$ for the model at epoch $e$.
To compare the experimental results to our theoretical model, we propose the following form:
\begin{align}\label{eq:omega_fit_form}
    \lambda^{(i, j, e)}_b \approx \theta^{(i, e)} + \frac{\alpha^{(e)}}{b^{\upsilon}} + \frac{\beta^{(e)}}{b^{2\upsilon}}\left(\frac{1}{\theta^{(i, e)}} + \frac{\gamma^{(e)}}{(\theta^{(i, e)})^2}\right)
\end{align}
where $\beta^{(e)} > 0$ (as the second cumulant of a any measure of non-negative) and $\theta^{(i, e)} > \theta^{(i+1, e)} > 0$ for all $i,e$.
The parameters $\alpha^{(e)}, \beta^{(e)}, \gamma^{(e)}$ and $\theta^{(i, e)}$ need to be fit to the data, which could be done with standard black-box optimisation to minimise squared error in (\ref{eq:omega_fit_form}), however we propose an alternative approach which reduces the number of free parameters and hence should regularise the optimisation problem.
Observe that (\ref{eq:omega_fit_form}) is linear in the parameters $\alpha^{(e)}, \beta^{(e)}, \gamma^{(e)}$ so, neglecting the positivity  constraint on $\beta^{(e)}$, we can in fact solve exactly for optimal values.
Firstly let us define $\barlam$ to be the empirical mean of $\lam$ over the batch seed index $j$.
Each epoch will be treated entirely separately, so let us drop the $e$ superscripts to streamline the notation.
We are then seeking to optimise $\alpha, \beta, \gamma, \theta^{(i)}$ to minimise \begin{align}\label{eq:opt_straight_form}
    E = \sum_{i,b}\left(\barlamnoe - \theta^{(i)} - \frac{\alpha}{b^{\upsilon}} - \frac{\beta}{\theta^{(i)}b^{2\upsilon}} - \frac{\beta\gamma}{b^{2\upsilon} (\theta^{(i)})^2}\right)^2.
\end{align}
Now make the following definitions
\begin{align}
    y_{ib} = \barlamnoe - \theta^{(i)}, ~ \vec{x}_{ib} = \left(\begin{array}{c} b^{-\upsilon} \\ (\theta^{(i)} b)^{-2\upsilon} \\ (b^{2\upsilon} (\theta^{(i)})^2)^{-1}\end{array}\right), ~ \vec{w} = \left(\begin{array}{c} \alpha \\ \beta \\ \beta\gamma\end{array}\right),
\end{align}
so that 
\begin{align}\label{eq:lin_reg_form}
    E = \sum_{i,b} (y_{ib} - \vec{w}^T \vec{x}_{ib})^2.
\end{align}
Finally we can define the vector $n$-dimensional $\vec{Y}$ by flattening the matrix $(y_{ib})_{ib}$, and the $3\times n$ matrix $\vec{X}$ by stacking the vectors $\vec{x}_{ib}$ and then flattening of the $i,b$ indices.
That done, we have have a standard linear regression problem with design matrix $\vec{X}$ and parameters $\vec{w}$.
For fixed $\theta$, the global minimum of $E$ is then attained at parameters
\begin{align}
    \vec{w}^*(\vec{\theta}) = (\vec{X}\vec{X}^T)^{-1}\vec{X}\vec{Y}
\end{align}
where the dependence on the parameters $\vec{\theta}$ is through $\vec{Y}$ and $\vec{X}$ as above.
We thus have $$\alpha = w^*_1, \beta = w^*_2, \gamma = w^*_3/w^*_2$$
and can plug these values back in to (\ref{eq:opt_straight_form}) to obtain an optimisation problem only over the $\theta^{(i)}$.
There is no closed form solution for the optimal $\theta^{(i)}$ for this problem, so we fit them using gradient descent.
The various settings and hyperparameters of this optimisation were tuned by hand to give convergence and are detailed in the Appendix Section \ref{sec:impl_fitting}.
To address the real constraint $\beta > 0$, we add a penalty term to the loss (\ref{eq:opt_straight_form}) which penalises values of $\theta^{(i)}$ leading to negative values of $\beta$.
The constraint $\theta^{(i)} > \theta^{(i+1)} > 0$ is implemented using a simple differentiable transformation detailed in the Appendix Section \ref{sec:impl_constraints}..
Finally, the exponent $\upsilon$ is selected by fitting the parameters for each $\upsilon$ in $\{-0.1, -0.2, \ldots, -0.9\}$ and taking the value with the minimum mean squared error $E$.

\medskip
The above process results in 12 fits for VGG and Resnet and 10 for MLP (one per epoch).
For each of these, we have a theoretical fit for each of the $5$ top outliers as a function of batch size which can be compared graphically to the data, resulting in $(2\times 12 + 10)\times 5 = 170$ plots.
Rather than try to display them all, we will select a small subset that illustrates the key features.
Figure \ref{fig:outlier_fit_resnet} shows results for the Resnet at epochs 0 (initialisation), 25, 250 and 300 (end of training) and outliers 1, 3 and 5.
Between the three models, the Resnet shows consistently the best agreement between the data and the parametric form (\ref{eq:omega_fit_form}).
The agreement is excellent at epoch 0 but quickly degrades to that seen in the second row of Figure \ref{fig:outlier_fit_resnet}, which is representative of the early and middle epochs for the Resnet.
Towards the end of training the Resnet returns to good agreement between theory and data, as demonstrated in the third and fourth rows of Figure \ref{fig:outlier_fit_resnet} at epochs 250 and 300 respectively.
\begin{figure}[h!]
    \centering
    \begin{subfigure}{0.3\linewidth}
     \centering
     \includegraphics[width=\linewidth]{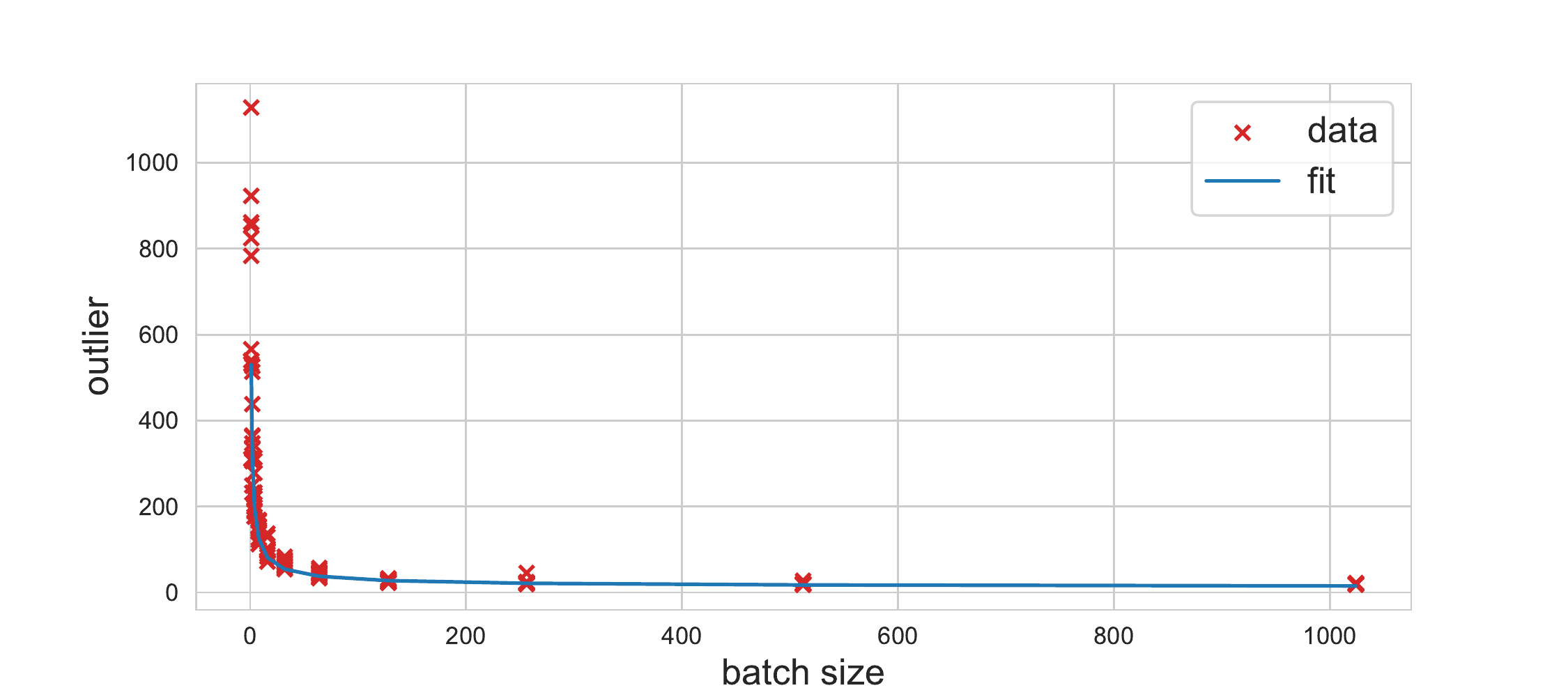}
     \subcaption{Outlier 1, epoch 0}
    \end{subfigure}
    \begin{subfigure}{0.3\linewidth}
     \centering
     \includegraphics[width=\linewidth]{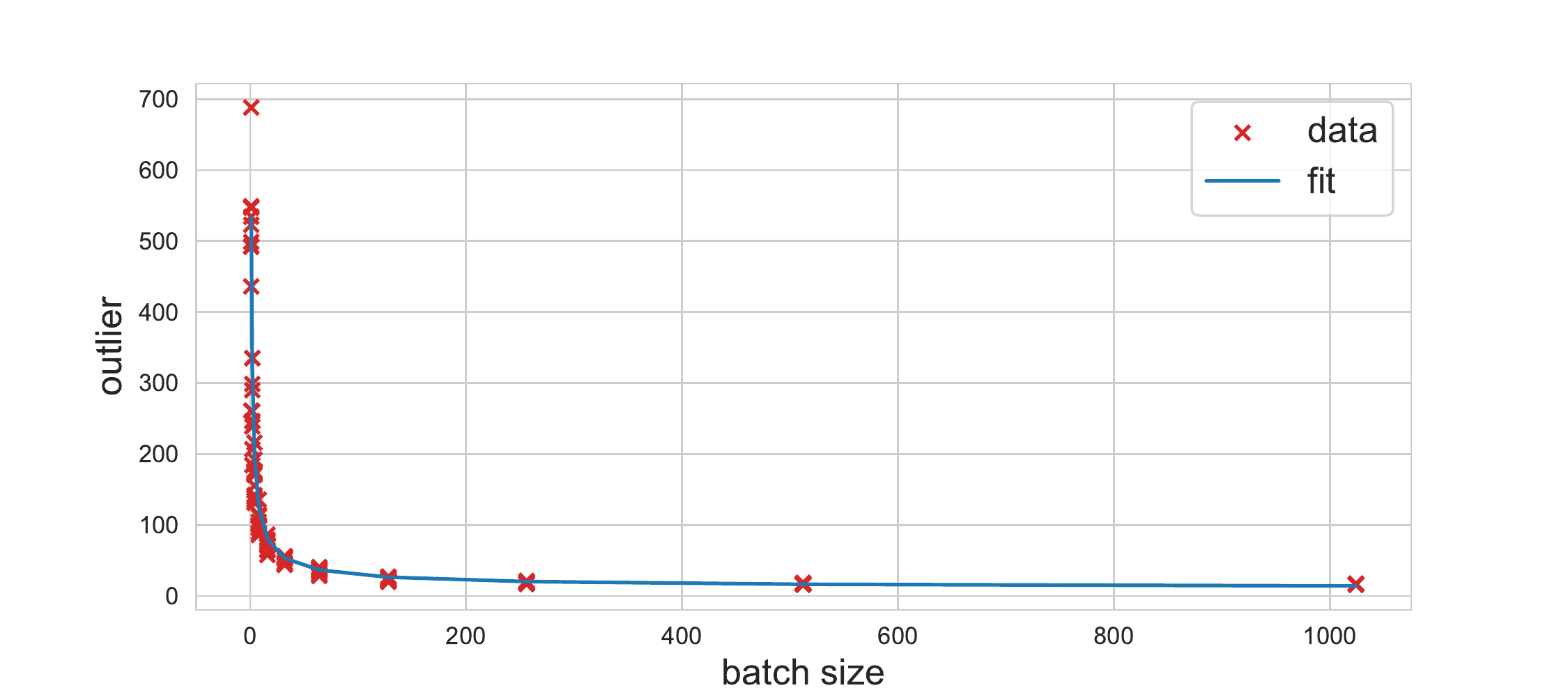}
     \subcaption{Outlier 3, epoch 0}
    \end{subfigure}
    \begin{subfigure}{0.3\linewidth}
     \centering
     \includegraphics[width=\linewidth]{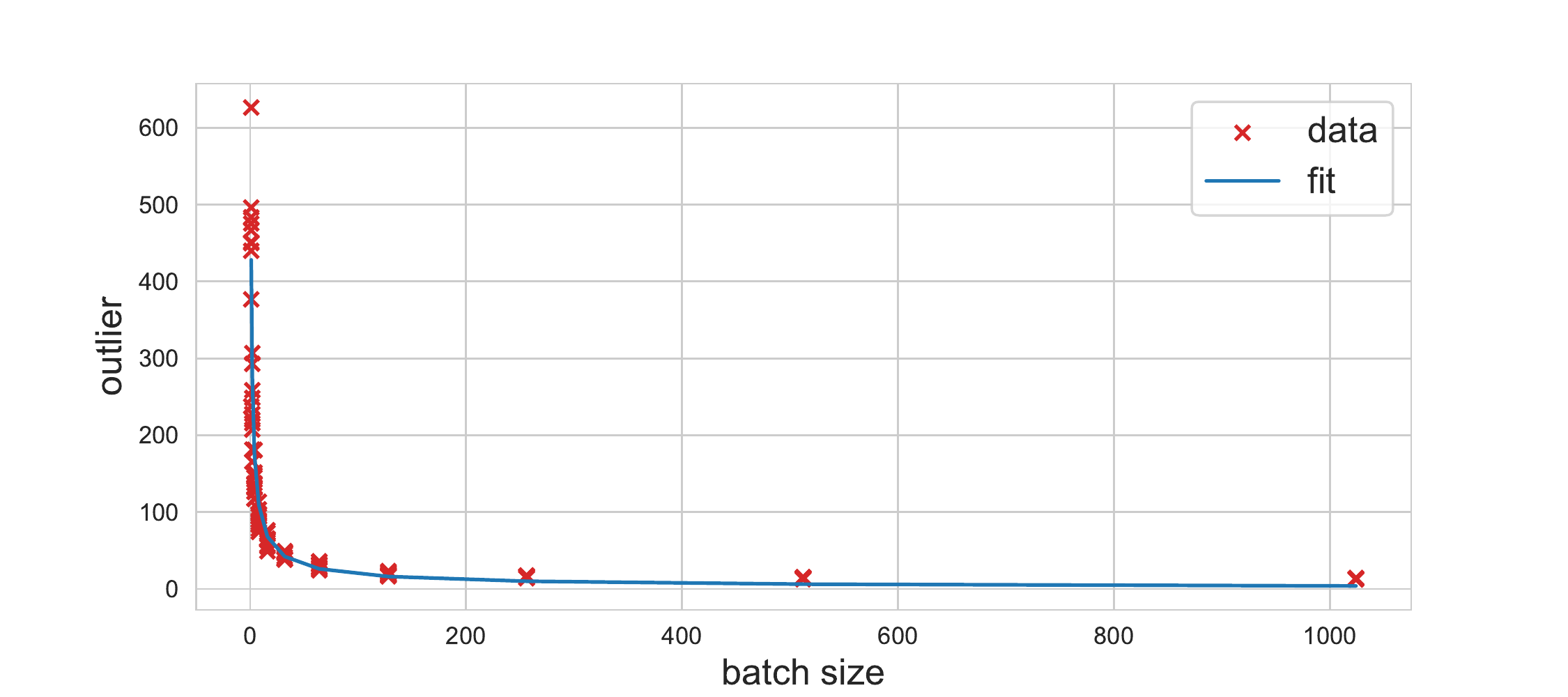}
     \subcaption{Outlier 5, epoch 0}
    \end{subfigure}
    
    \begin{subfigure}{0.3\linewidth}
     \centering
     \includegraphics[width=\linewidth]{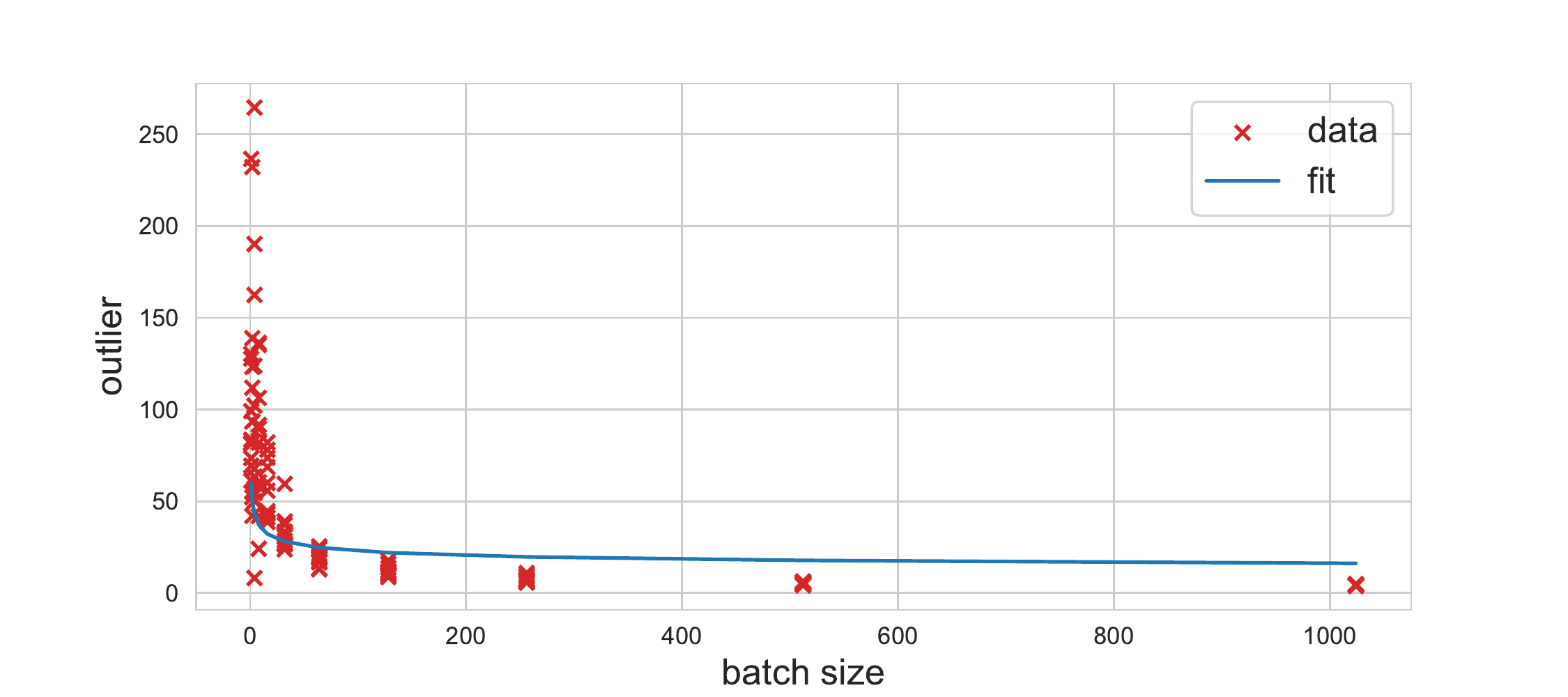}
     \subcaption{Outlier 1, epoch 25}
    \end{subfigure}
    \begin{subfigure}{0.3\linewidth}
     \centering
     \includegraphics[width=\linewidth]{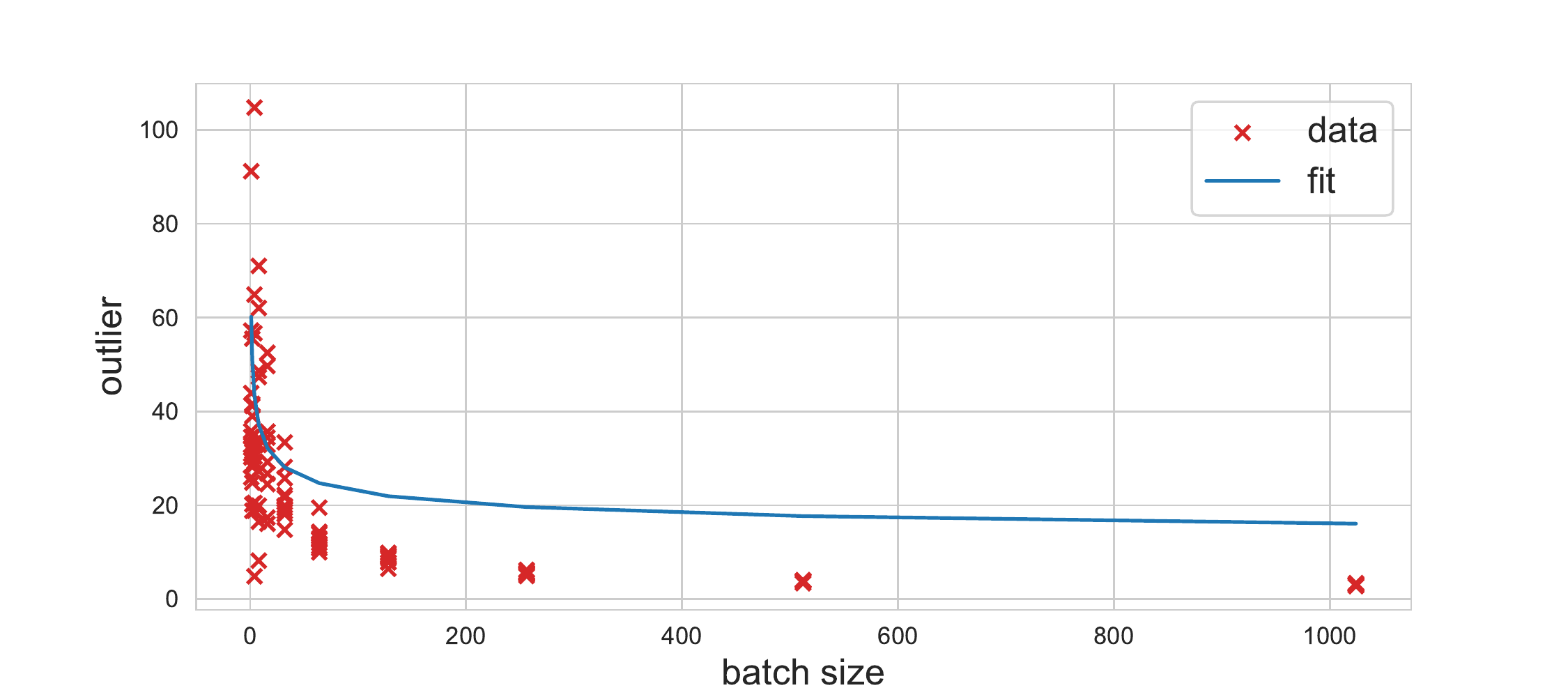}
     \subcaption{Outlier 3, epoch 25}
    \end{subfigure}
    \begin{subfigure}{0.3\linewidth}
     \centering
     \includegraphics[width=\linewidth]{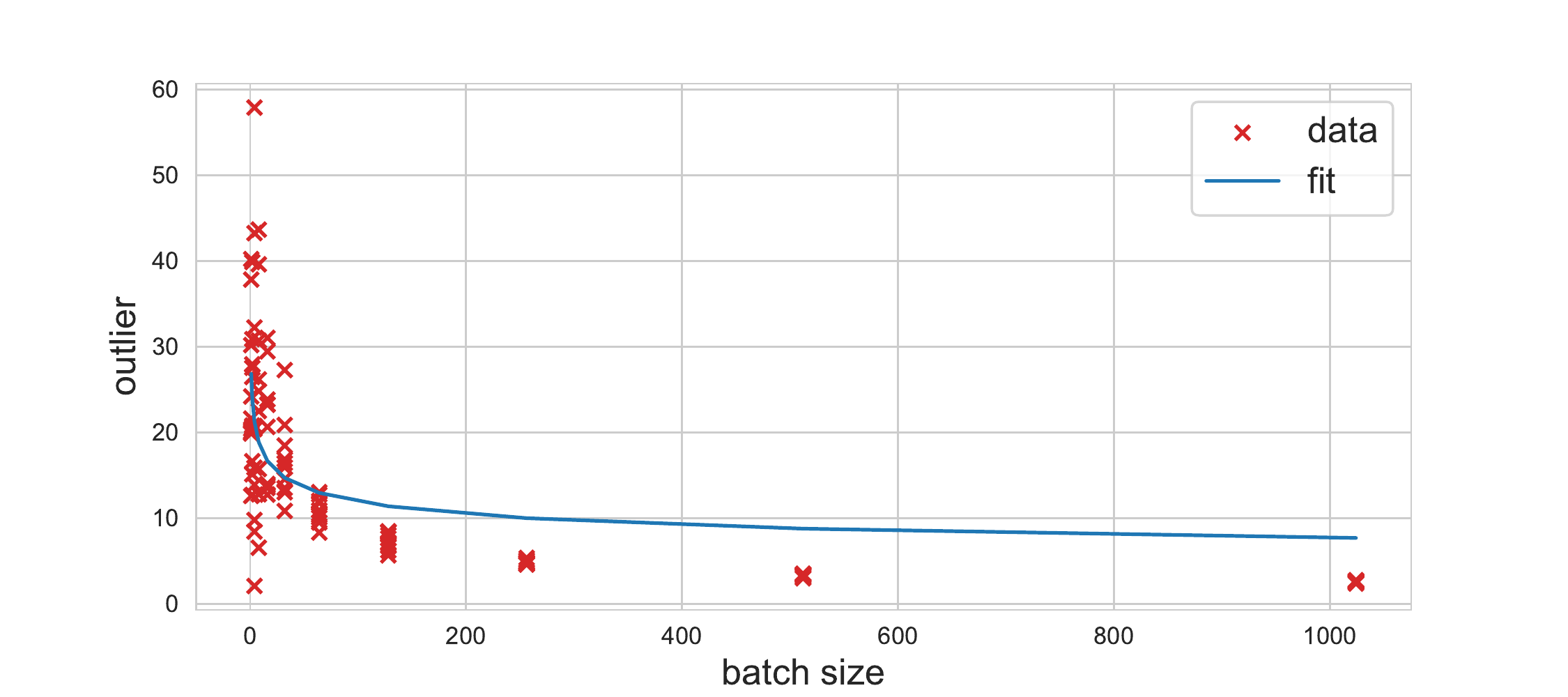}
     \subcaption{Outlier 5, epoch 25}
    \end{subfigure}
    
    \begin{subfigure}{0.3\linewidth}
     \centering
     \includegraphics[width=\linewidth]{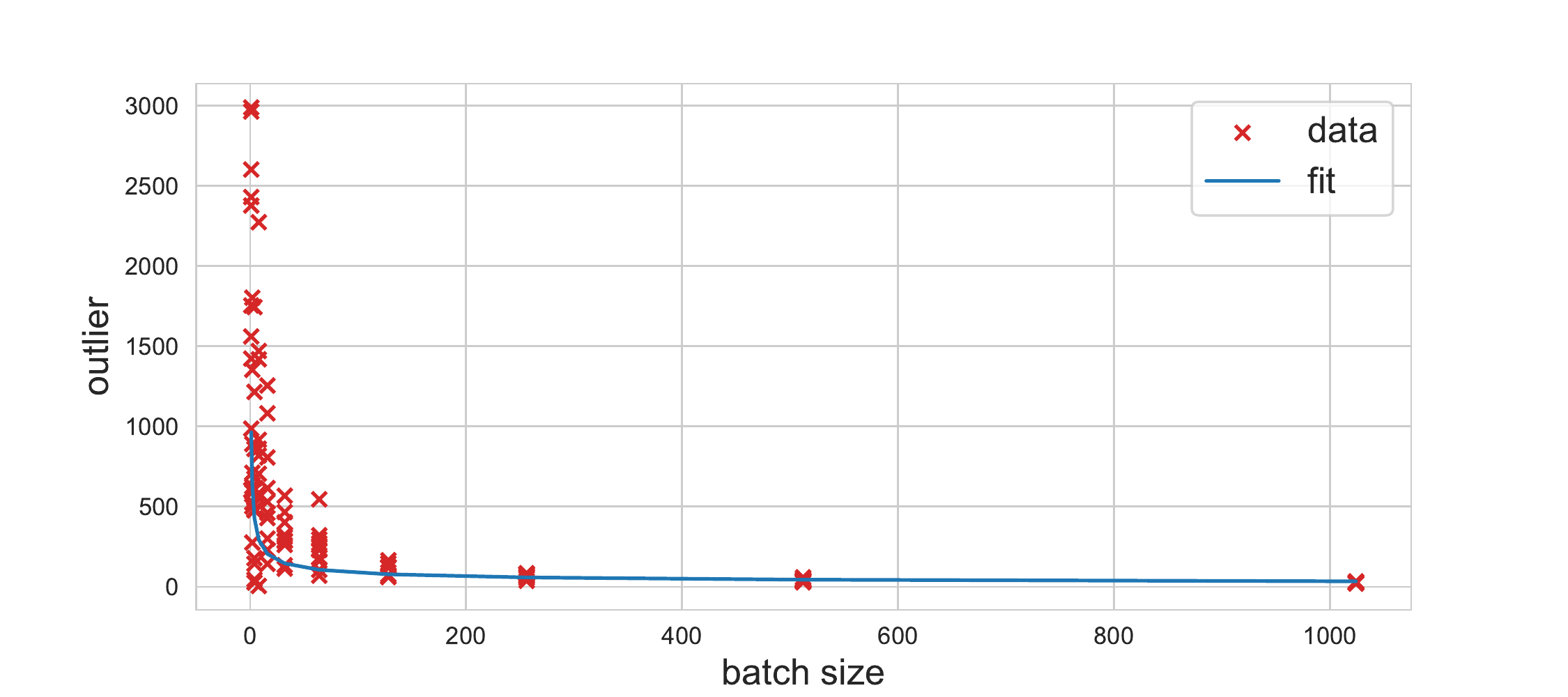}
     \subcaption{Outlier 1, epoch 250}
    \end{subfigure}
    \begin{subfigure}{0.3\linewidth}
     \centering
     \includegraphics[width=\linewidth]{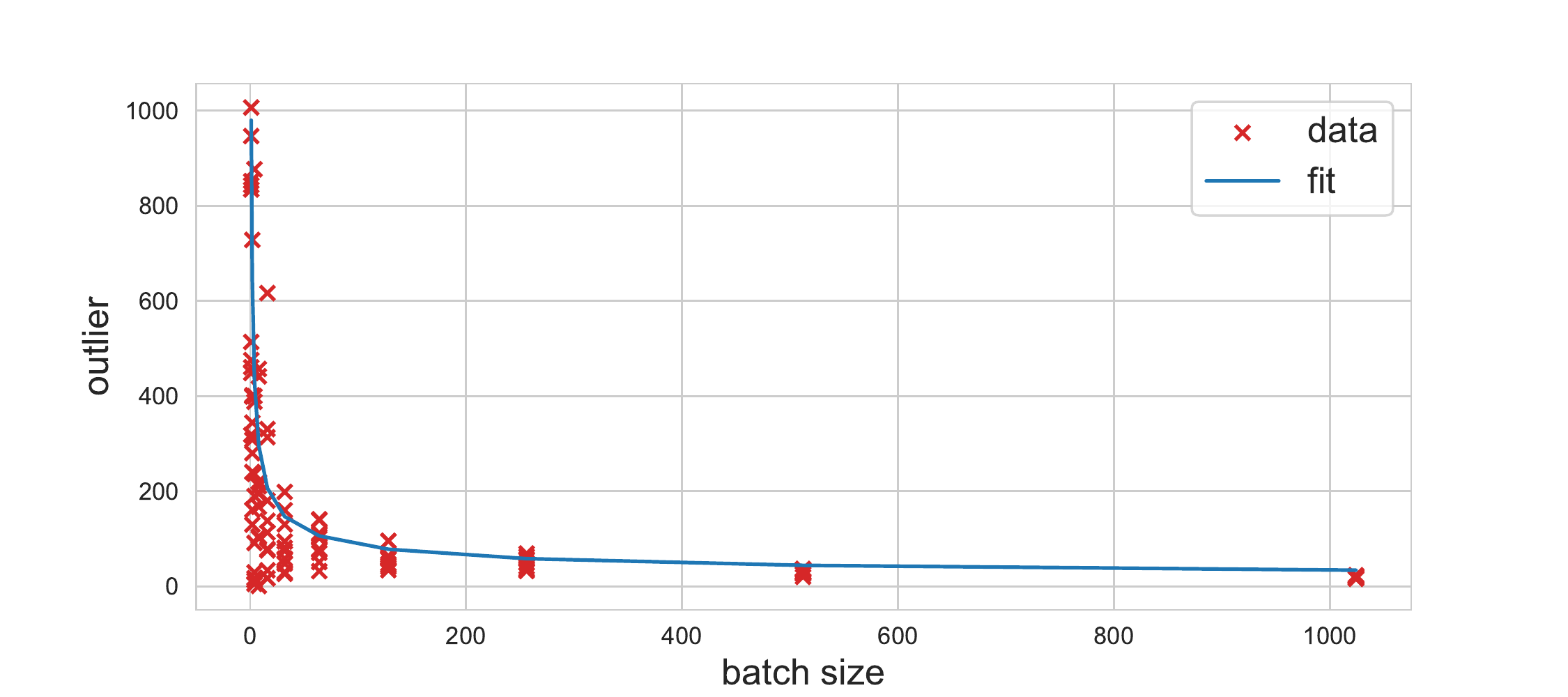}
     \subcaption{Outlier 3, epoch 250}
    \end{subfigure}
    \begin{subfigure}{0.3\linewidth}
     \centering
     \includegraphics[width=\linewidth]{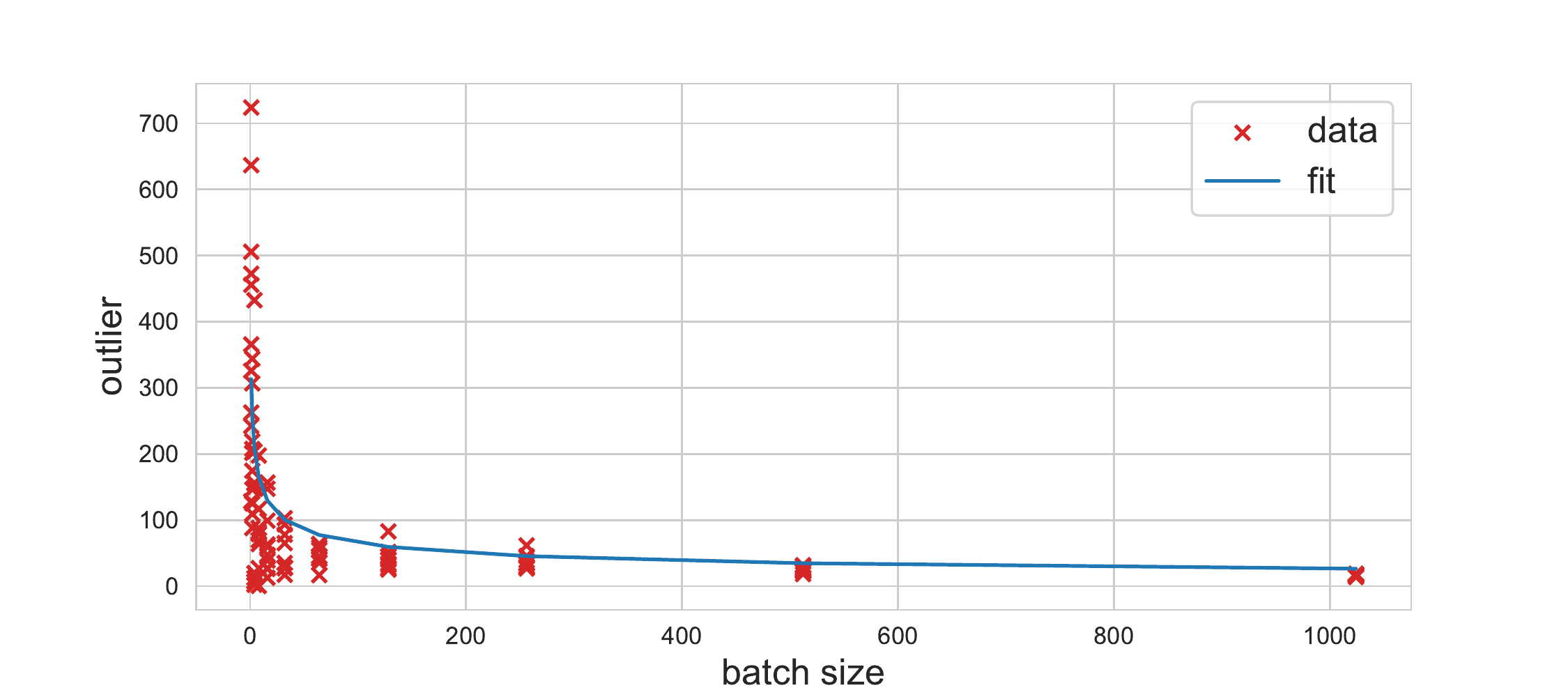}
     \subcaption{Outlier 5, epoch 250}
    \end{subfigure}
    
    \begin{subfigure}{0.3\linewidth}
     \centering
     \includegraphics[width=\linewidth]{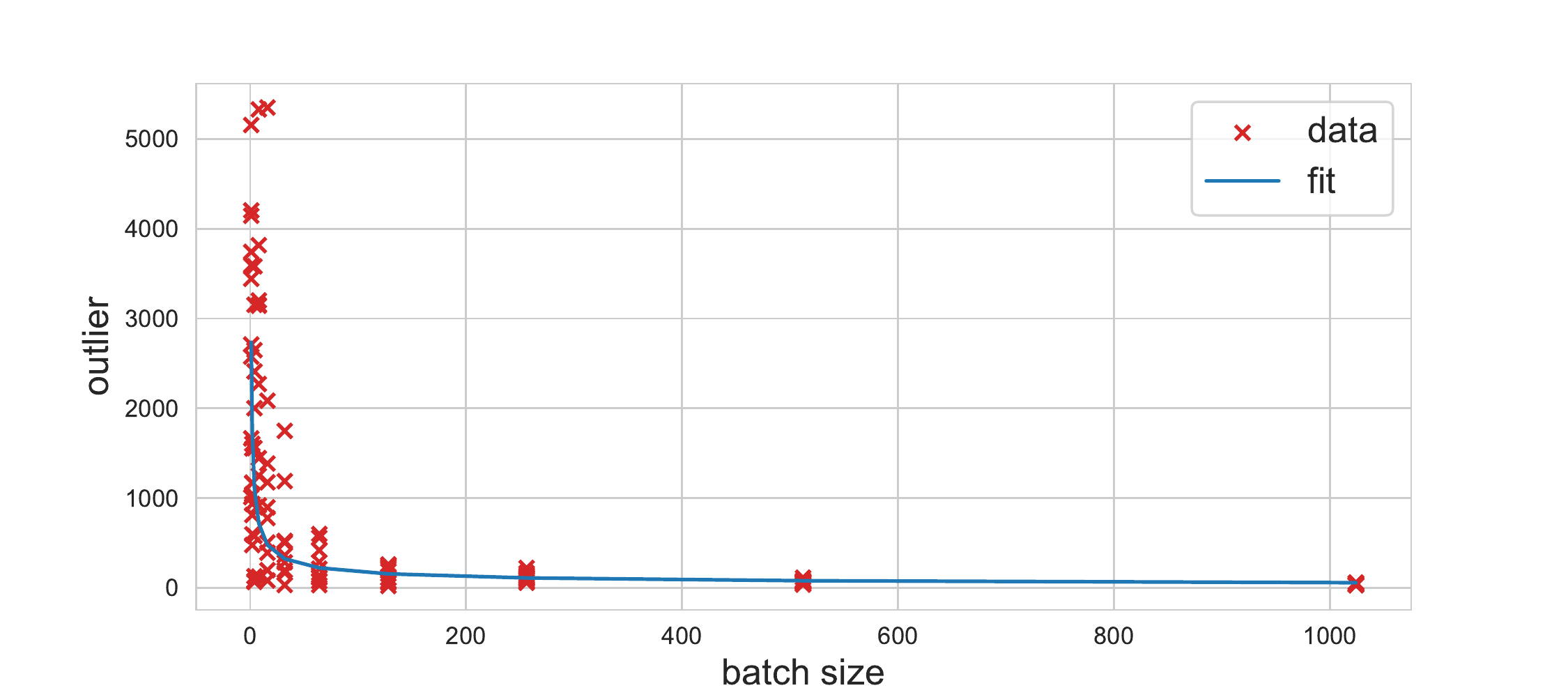}
     \subcaption{Outlier 1, epoch 300}
    \end{subfigure}
    \begin{subfigure}{0.3\linewidth}
     \centering
     \includegraphics[width=\linewidth]{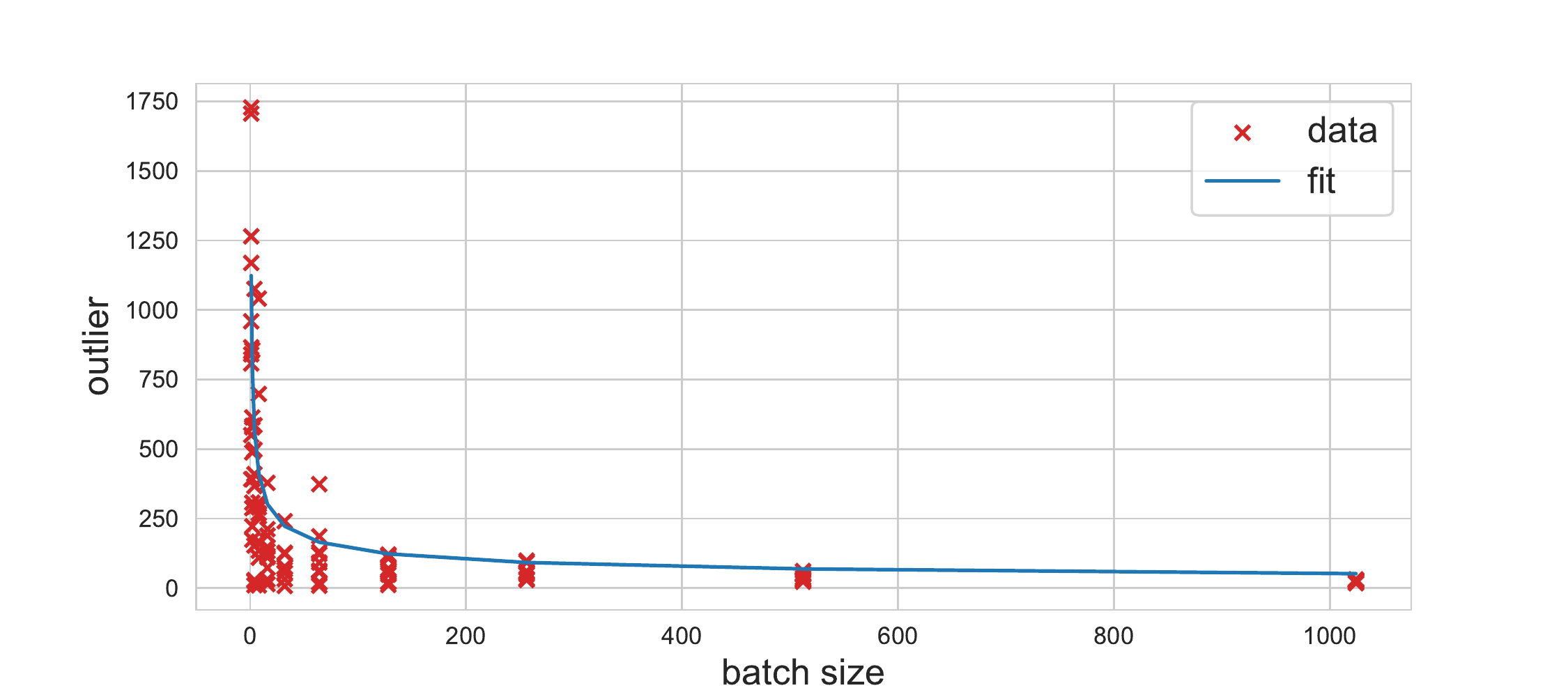}
     \subcaption{Outlier 3, epoch 300}
    \end{subfigure}
    \begin{subfigure}{0.3\linewidth}
     \centering
     \includegraphics[width=\linewidth]{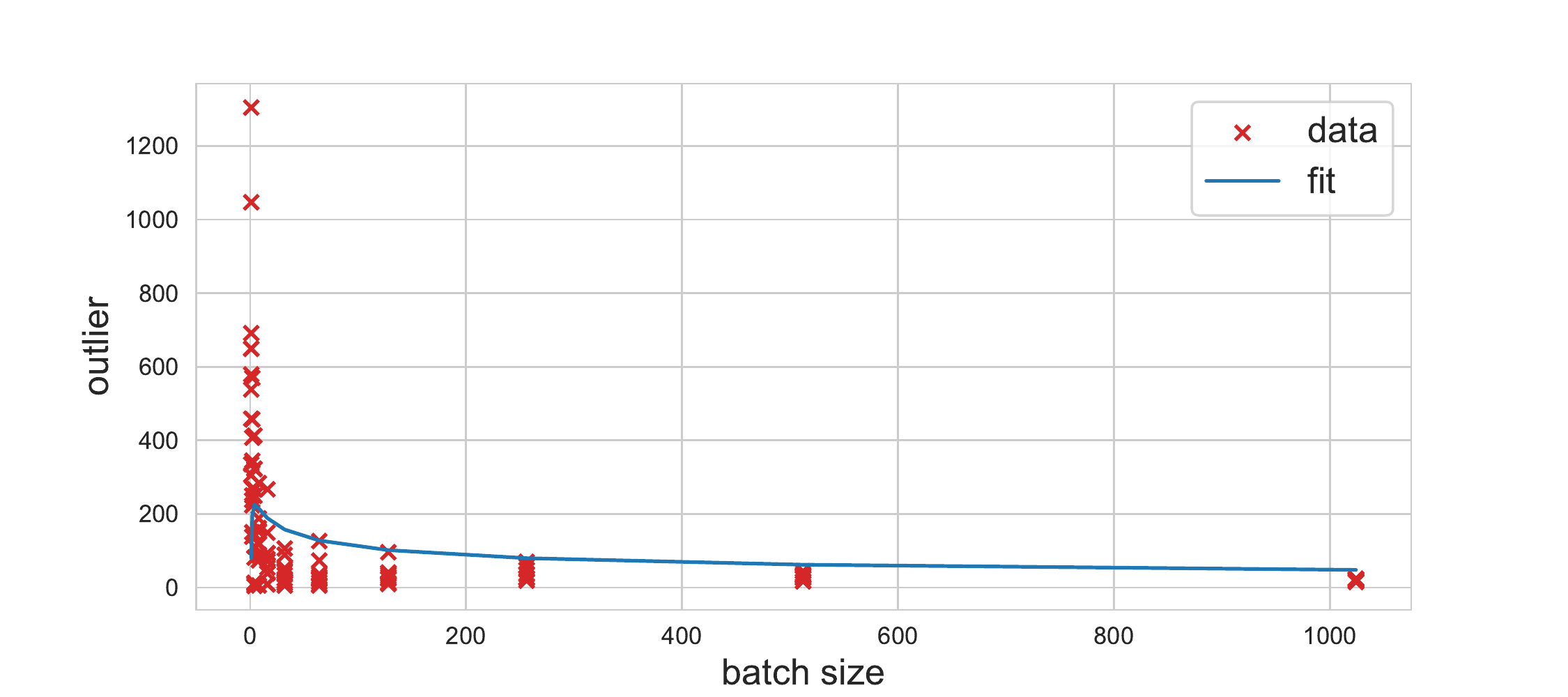}
     \subcaption{Outlier 5, epoch 300}
    \end{subfigure}
    \centering
    \caption{The batch-size scaling of the outliers in the spectra of the Hessians of the Resnet loss on CIFAR100. Training epochs increase top-to-bottom from initialisation to final trained model. Left-to-right the outlier index varies (outlier 1 being the largest). Red cross show results from Lanczos approximations over 10 samples (different batches) for each batch size. The blue lines are parametric power law fits of the form (\ref{eq:omega_fit_form}).}
    \label{fig:outlier_fit_resnet}
\end{figure}

The VGG16 also has excellent agreement between theory and data at epoch 0, and thereafter is similar to the early epochs of the Resnet, i.e. reasonable, but not excellent, until around epoch 225 where the agreement starts to degrade significantly until the almost complete failure at epoch 300 shown in the first row of Figure \ref{fig:outlier_fit_ugly_ducklings}.
The MLP has the worst agreement between theory and data, having again excellent agreement at epoch 0, but really quite poor agreement even by epoch 1, as shown in the second row of Figure \ref{fig:outlier_fit_ugly_ducklings}.
\begin{figure}[h!]
    \centering
    \begin{subfigure}{0.3\linewidth}
     \centering
     \includegraphics[width=\linewidth]{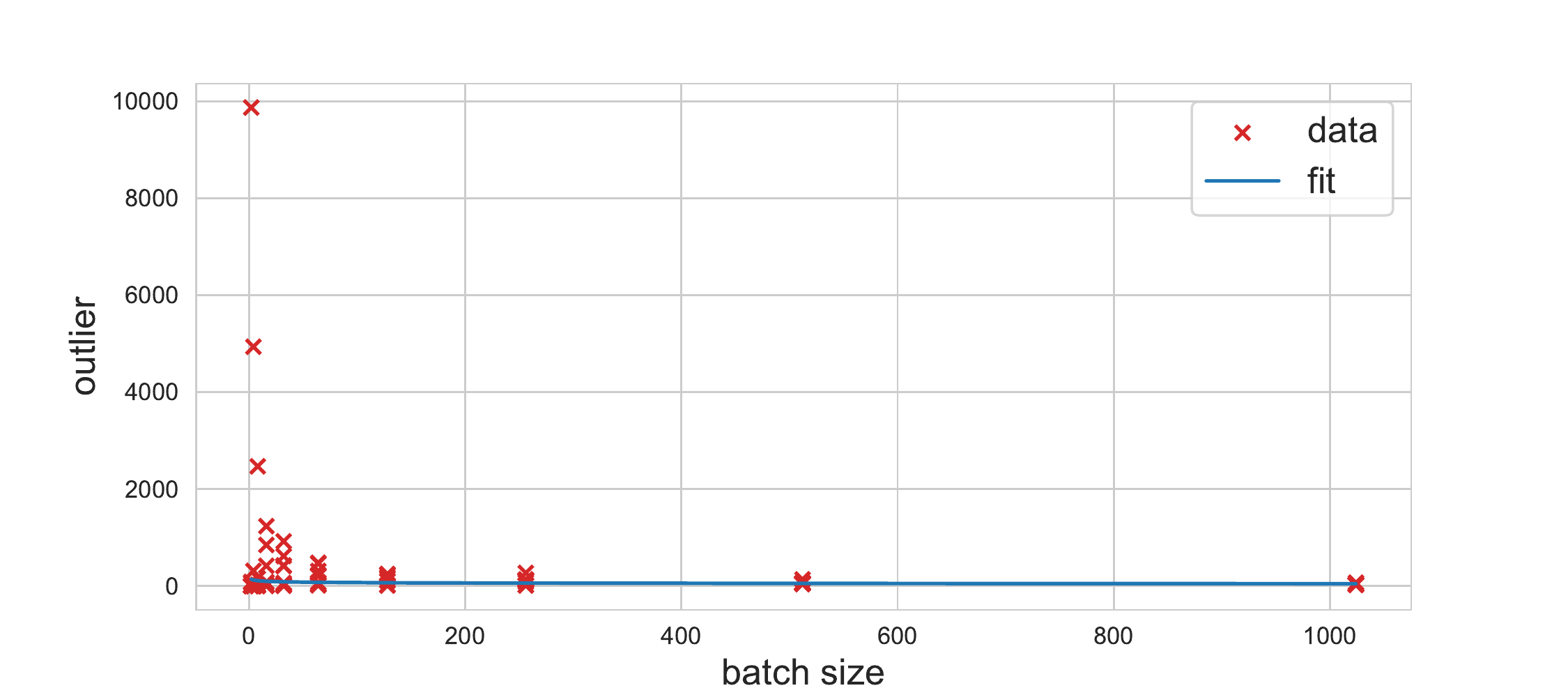}
     \subcaption{VGG16, epoch 300, outlier 1}
    \end{subfigure}
    \begin{subfigure}{0.3\linewidth}
     \centering
     \includegraphics[width=\linewidth]{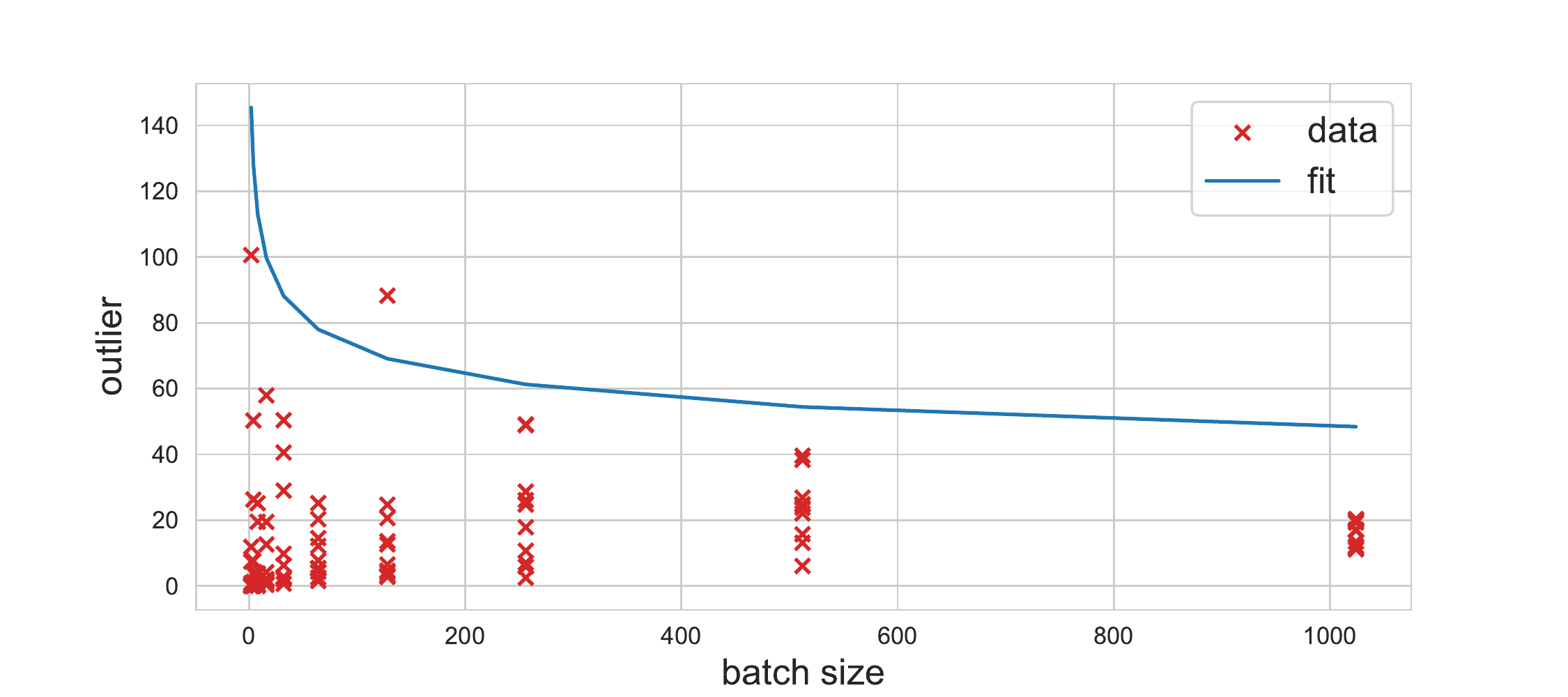}
     \subcaption{VGG16, epoch 300, outlier 3}
    \end{subfigure}
    \begin{subfigure}{0.3\linewidth}
     \centering
     \includegraphics[width=\linewidth]{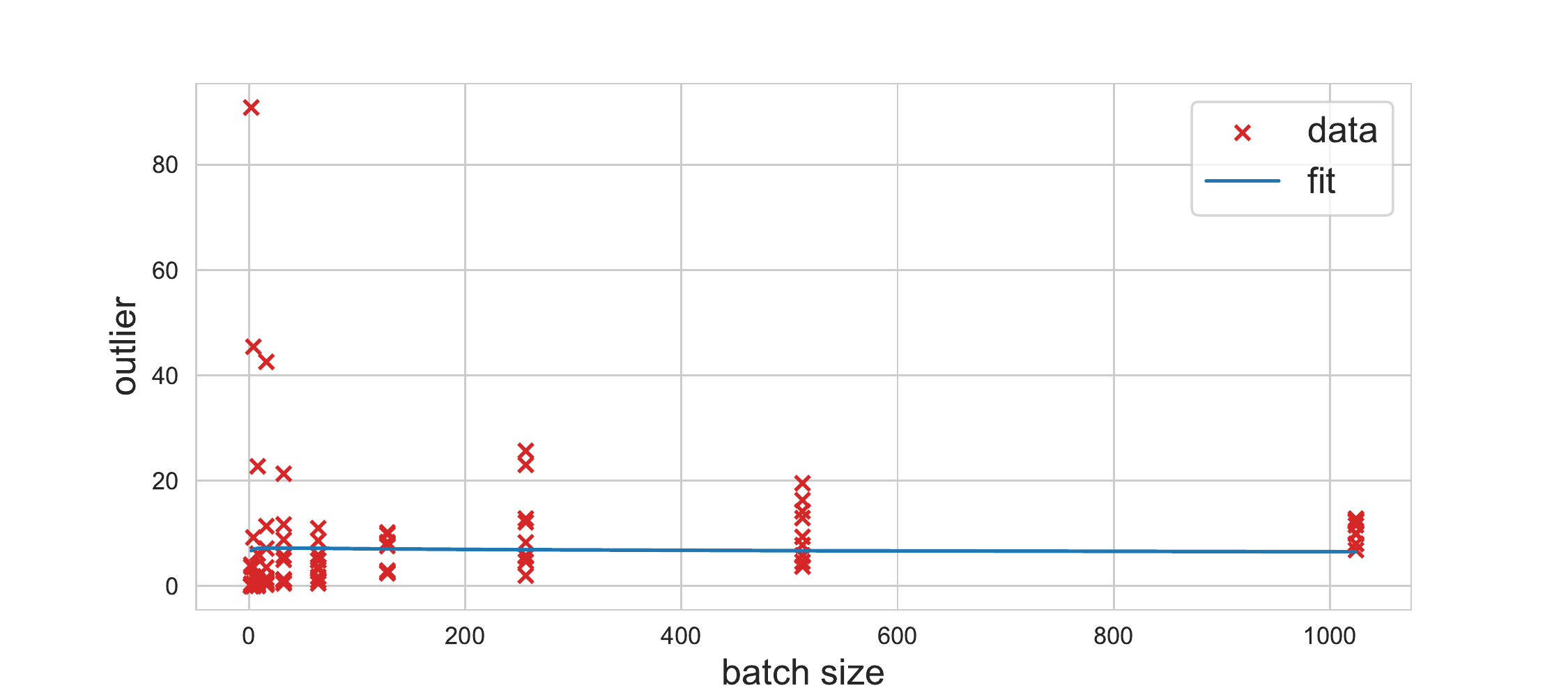}
     \subcaption{VGG16, epoch 300, outlier 5}
    \end{subfigure}
    
    \begin{subfigure}{0.3\linewidth}
     \centering
     \includegraphics[width=\linewidth]{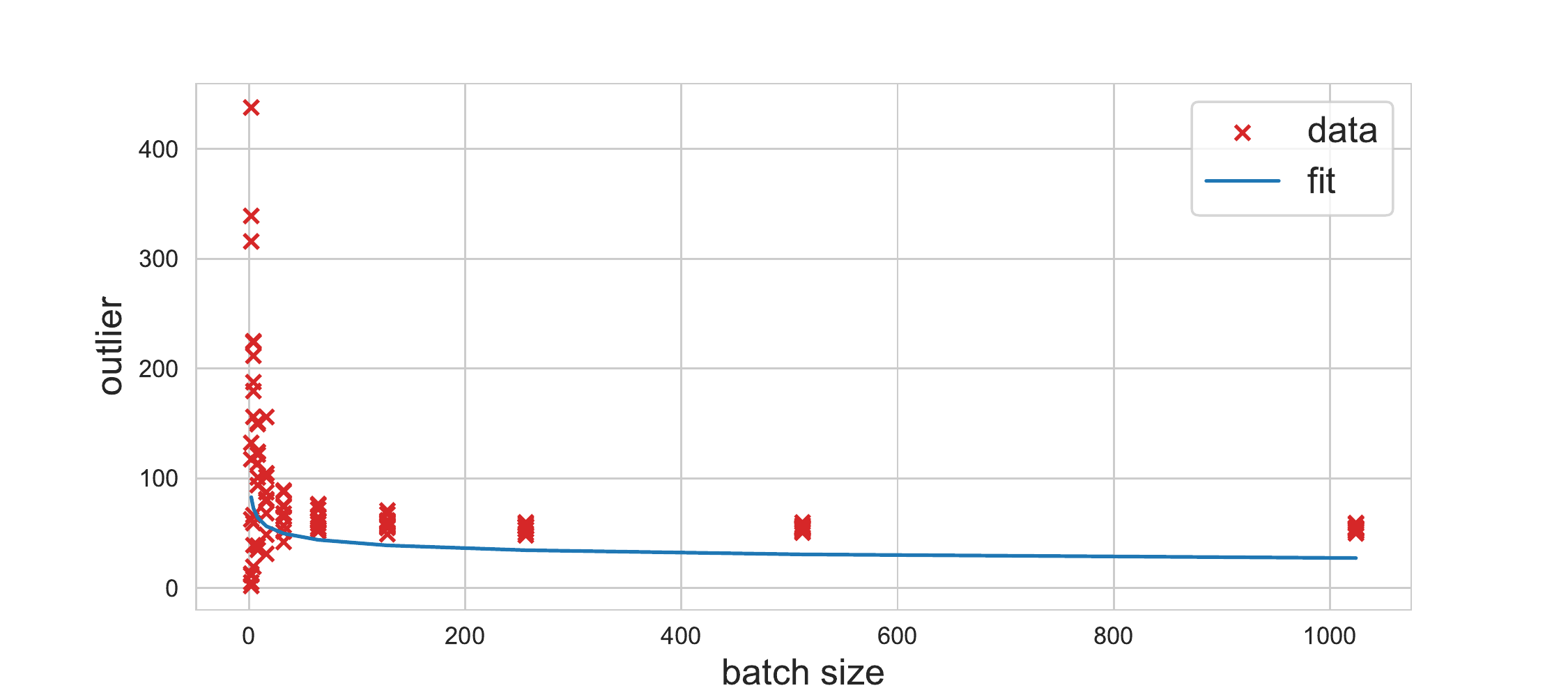}
     \subcaption{MLP, epoch 1, outlier 1}
    \end{subfigure}
    \begin{subfigure}{0.3\linewidth}
     \centering
    \includegraphics[width=\linewidth]{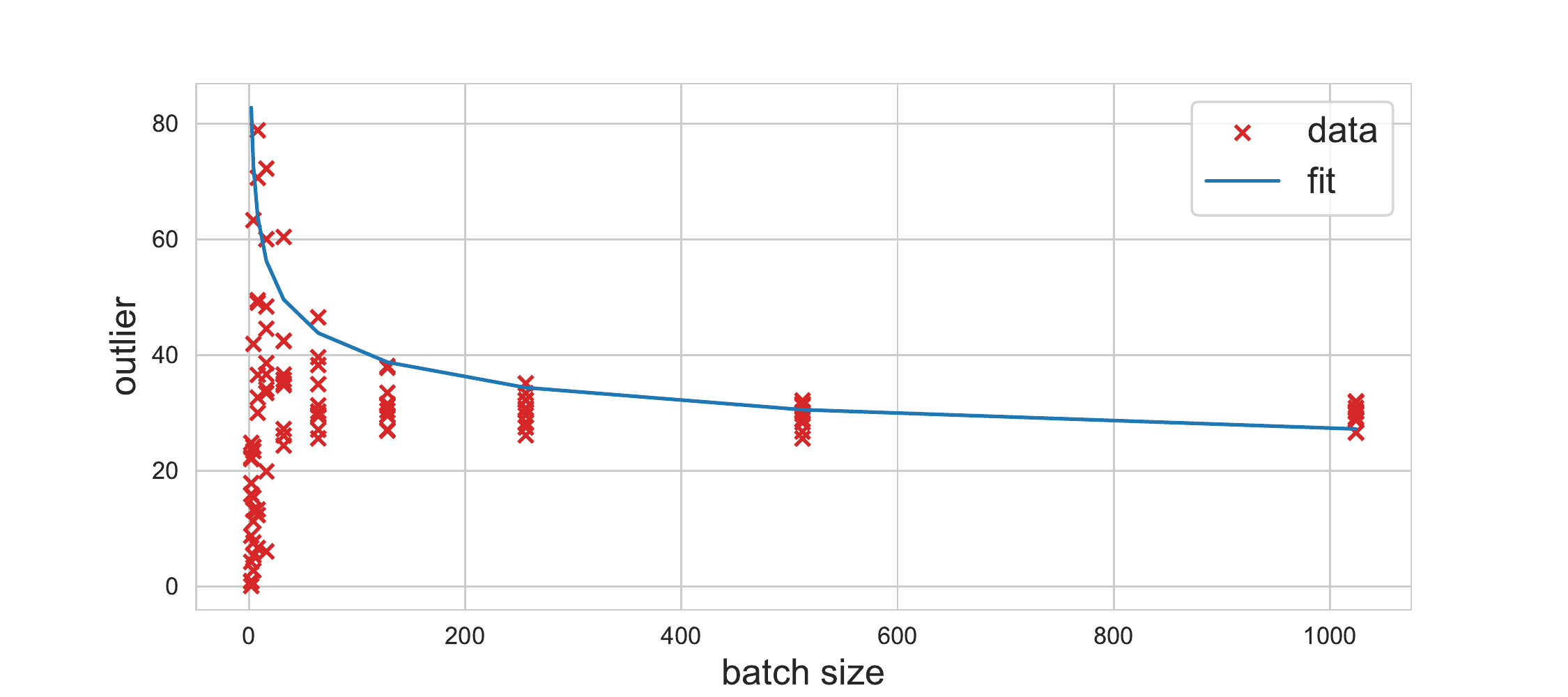}
     \subcaption{MLP, epoch 1, outlier 3}
    \end{subfigure}
    \begin{subfigure}{0.3\linewidth}
     \centering
    \includegraphics[width=\linewidth]{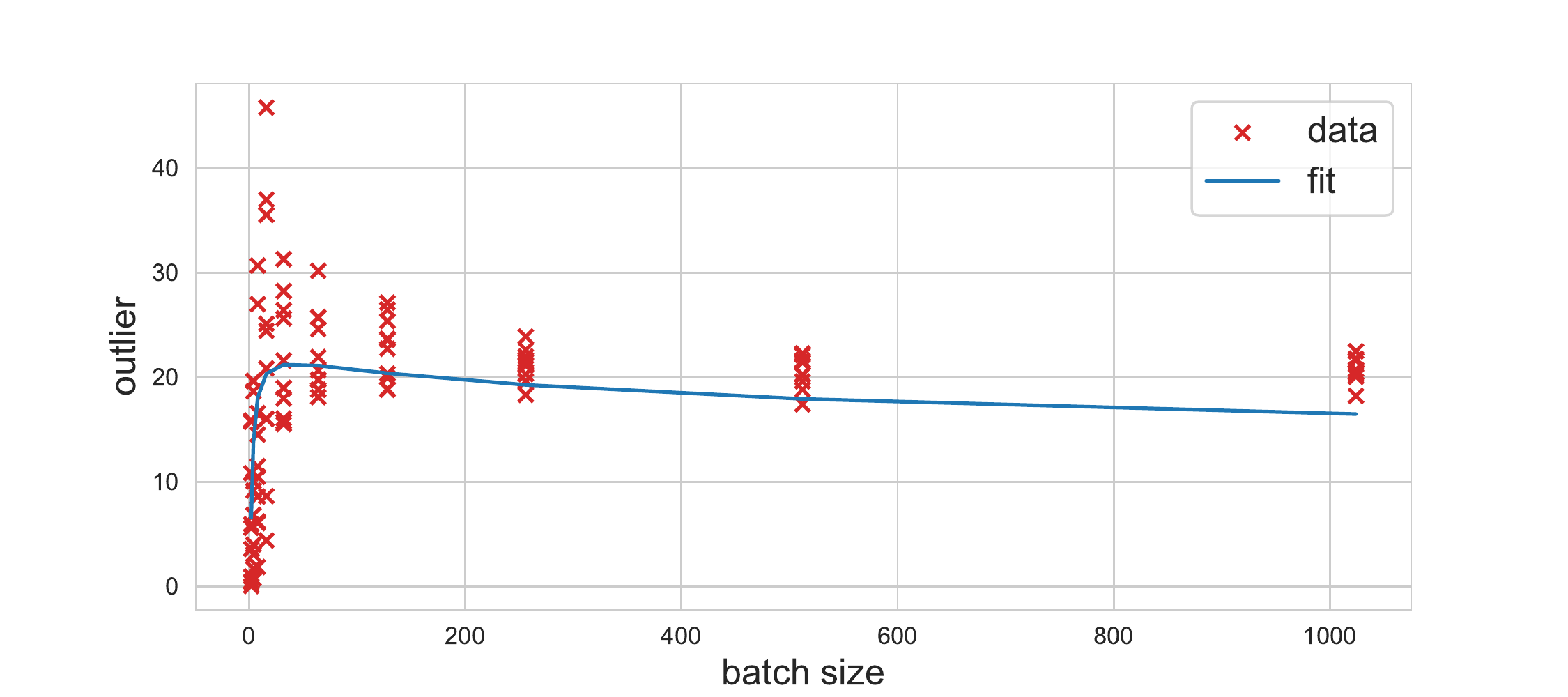}
     \subcaption{MLP, epoch 1, outlier 5}
    \end{subfigure}
    \centering
    \caption{Left-to-right the outlier index varies (outlier 1 being the largest). Red cross show results from Lanczos approximations over 10 samples (different batches) for each batch size. The blue lines are parametric power law fits of the form (\ref{eq:omega_fit_form}). This plot show the final epoch (300) for the VGG16 on CIFAR100 and the first epoch for the MLP on MNIST, both being examples of the parametric fit failing to match the data.}
    \label{fig:outlier_fit_ugly_ducklings}
\end{figure}

The experimental results show an ordering Resnet $>$ VGG $>$ MLP, in terms of how well the random matrix theory loss surface predictions explain the Hessian outliers.
We conjecture that this relates to the difficulty of the loss surfaces.
Resnets are generally believed to have smoother, simpler loss surfaces \cite{li2018visualizing} and be easier to train than other architectures, indeed the residual connections were originally introduced for precisely this reason.
The VGG is generally more sensitive to training set-up, requiring well-tuned hyperparameters to avoid unstable or unsuccessful training \cite{https://doi.org/10.48550/arxiv.2011.08181}.
The MLP is perhaps too small to benefit from high-dimensional highly over-parametrised effects.

\medskip
The parameter values obtained for all models over all epochs are shown in Figure \ref{fig:outlier_fit_summary_params}, with a column for each model.
There are several interesting features to draw out of these plots, however note that we cannot meaningfully interpret the parameters for the MLP beyond epoch 0, as the agreement with (\ref{eq:omega_fit_form}) is so poor. 
Firstly consider the parameter $m_1^{(\mu)}$, which is interpreted as the first moment (i.e. mean) of the spectral density of the noise matrix $X$.
Note that for all models this parameter starts close to 0 and generally grows with training epochs, noting that the right hand side of Figure \ref{fig:vgg_m1} at the higher epochs should be ignored owing to the bad fit discussed above.
$m_1^{(\mu)}=0$ is significant, as it is seen in the case of the a symmetric measure $\mu$, such as the Wigner semicircle used by \cite{granziol2020learning}.
It is interesting also to observe that $\epsilon m_1^{(\eta)}$ remains small for all epochs particularly compared to $m_1^{(\mu)}, k_2^{(\mu)}$.
This is consistent with the derivation of (\ref{eq:omega_fit_form}), which relies on $\epsilon$ being small, however we emphasise that \emph{this was not imposed as a numerical constraint} but arises naturally from the data.
Recall that the magnitude of $\epsilon m_1^{(\eta)}$ measures the extent of the deviation of $A$ from being exactly low rank, so its small but non-zero values suggest that it is indeed important to allow for the true Hessian to have non-zero rank in the $N\rightarrow\infty$ limit.
Finally, we comment that the best exponent is generally not $\upsilon=1/2$.
Again, the results from the Resnet are the most reliable and they appear to show that the batch scaling, as characterised by $\upsilon$, is not constant throughout training, particularly comparing epoch 0 and epoch 300, say.
\begin{figure}[h!]
    \centering
    \begin{subfigure}{0.3\linewidth}
     \centering
     \includegraphics[width=\linewidth]{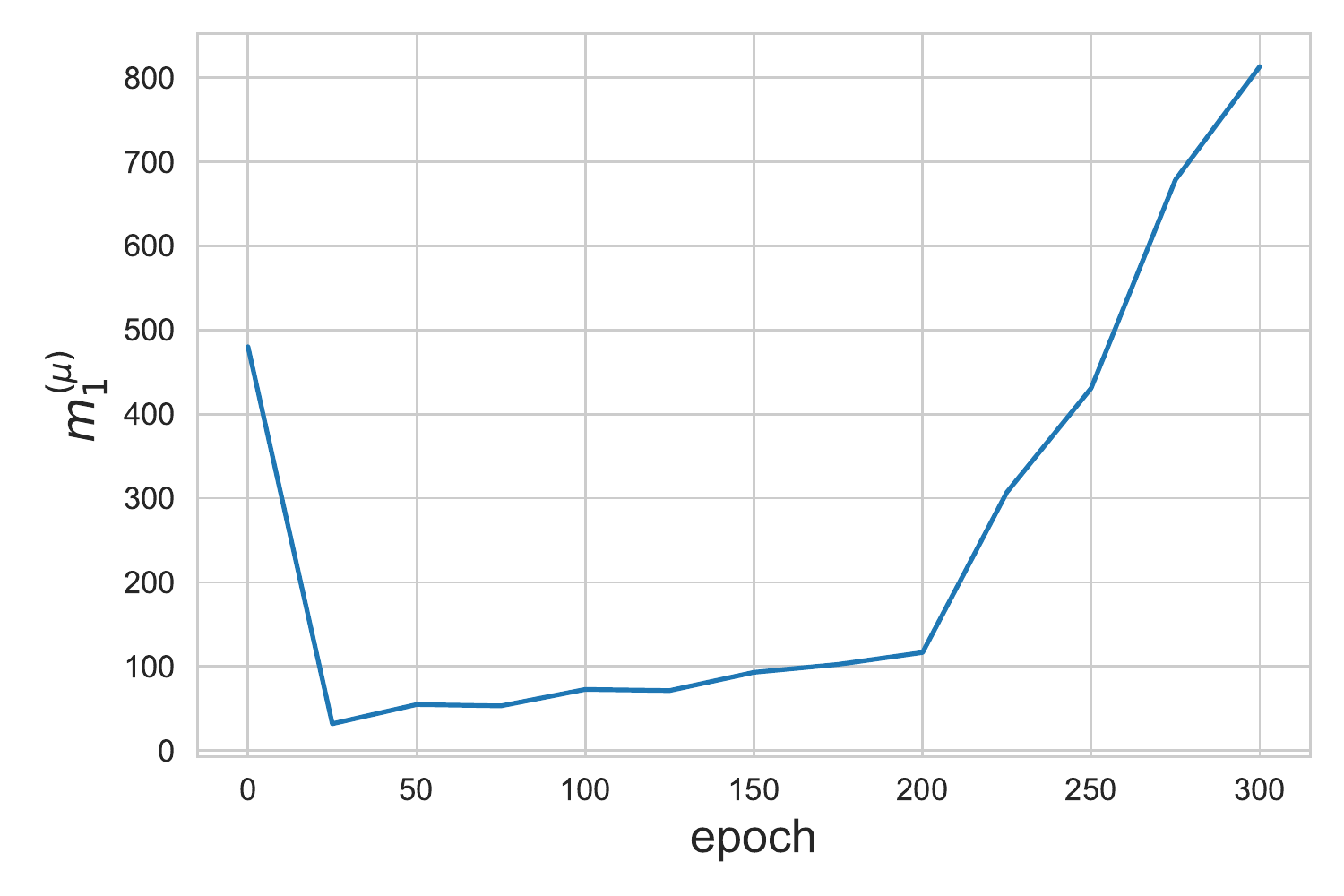}
     \subcaption{$m_1^{(\mu)}$, Resnet on CIFAR100}
    \end{subfigure}
    \begin{subfigure}{0.3\linewidth}
     \centering
     \includegraphics[width=\linewidth]{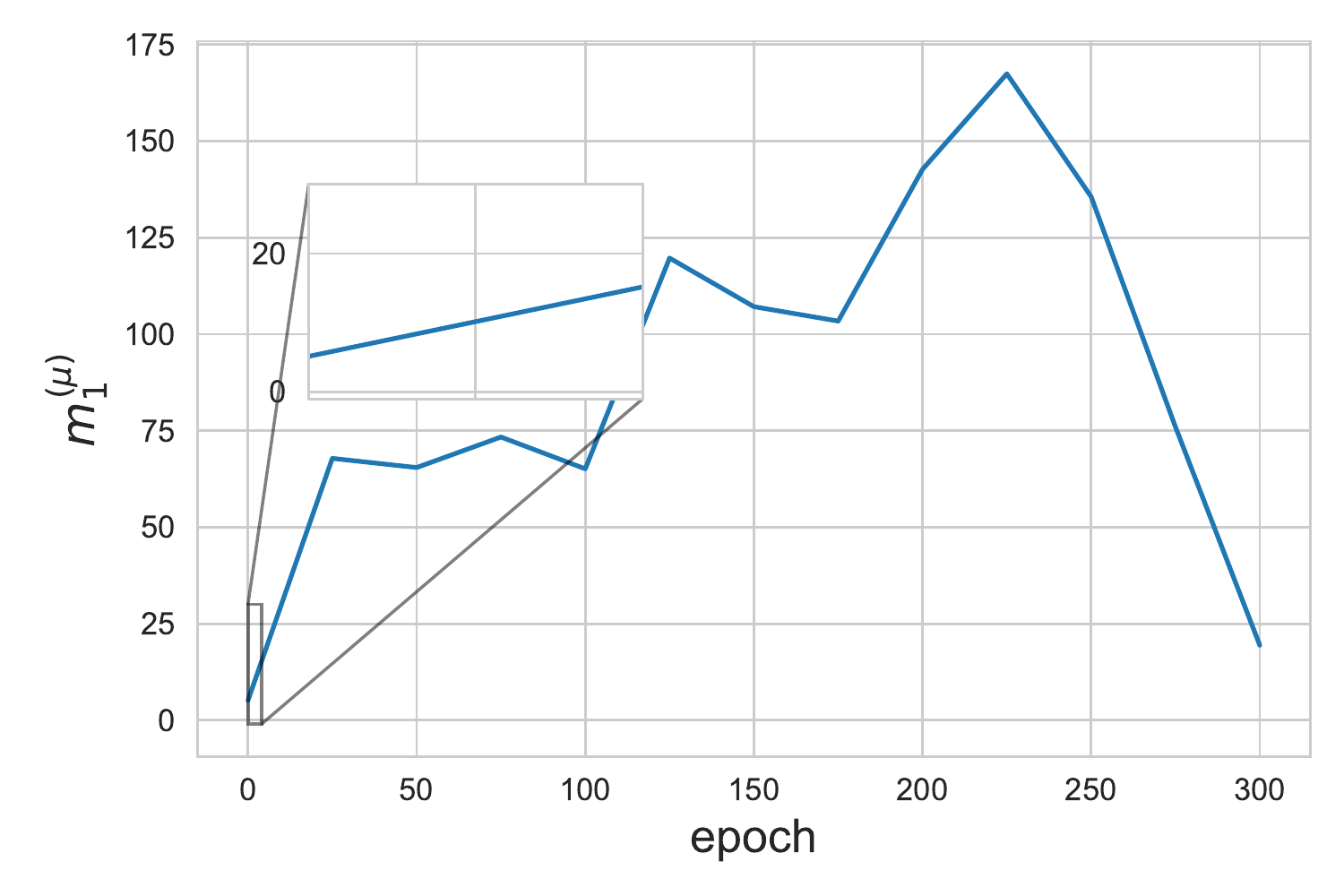}
     \subcaption{$m_1^{(\mu)}$, VGG16 on CIFAR100}
     \label{fig:vgg_m1}
    \end{subfigure}
    \begin{subfigure}{0.3\linewidth}
     \centering
     \includegraphics[width=\linewidth]{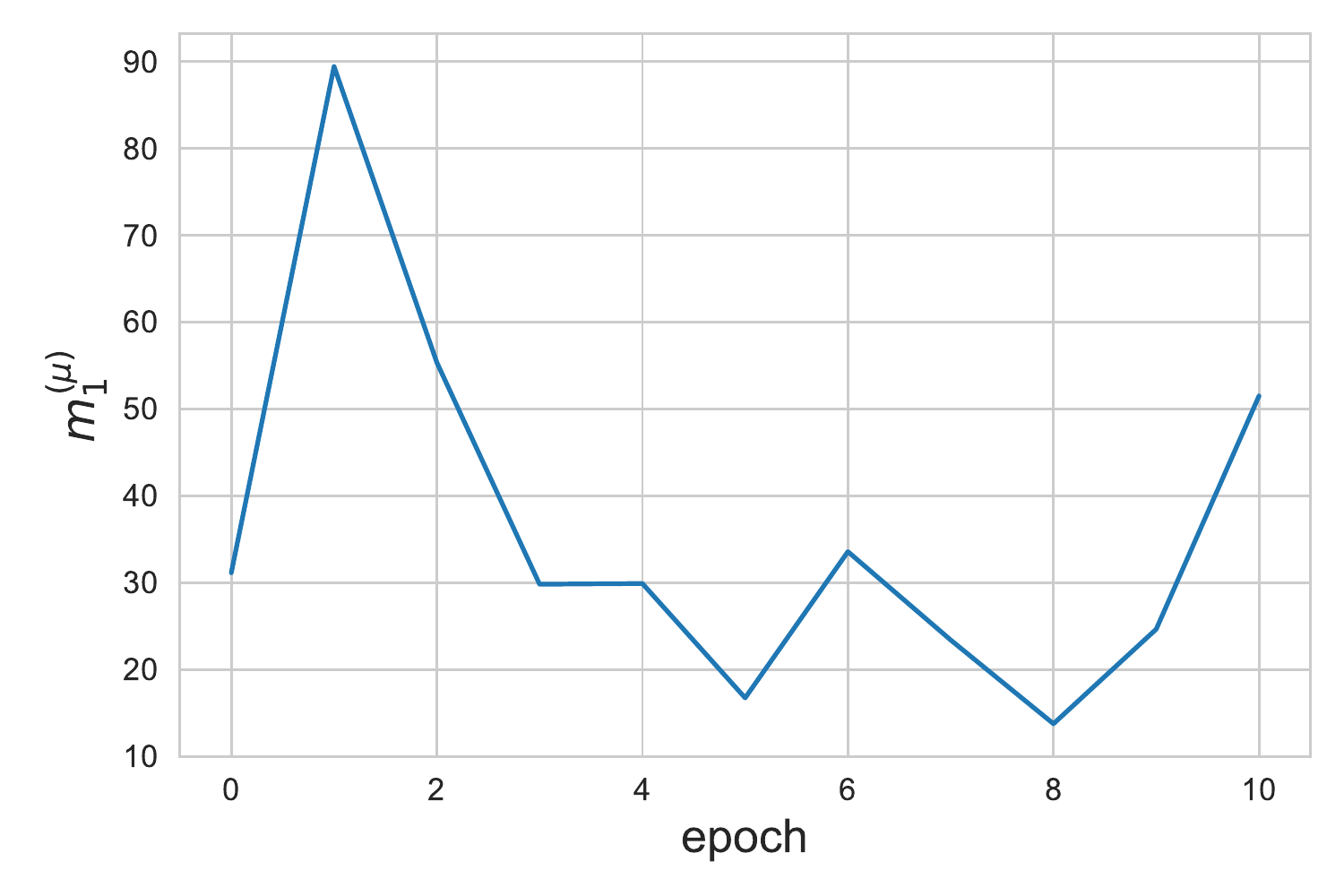}
     \subcaption{$m_1^{(\mu)}$, MLP on MNIST}
    \end{subfigure}
    
    \begin{subfigure}{0.3\linewidth}
     \centering
     \includegraphics[width=\linewidth]{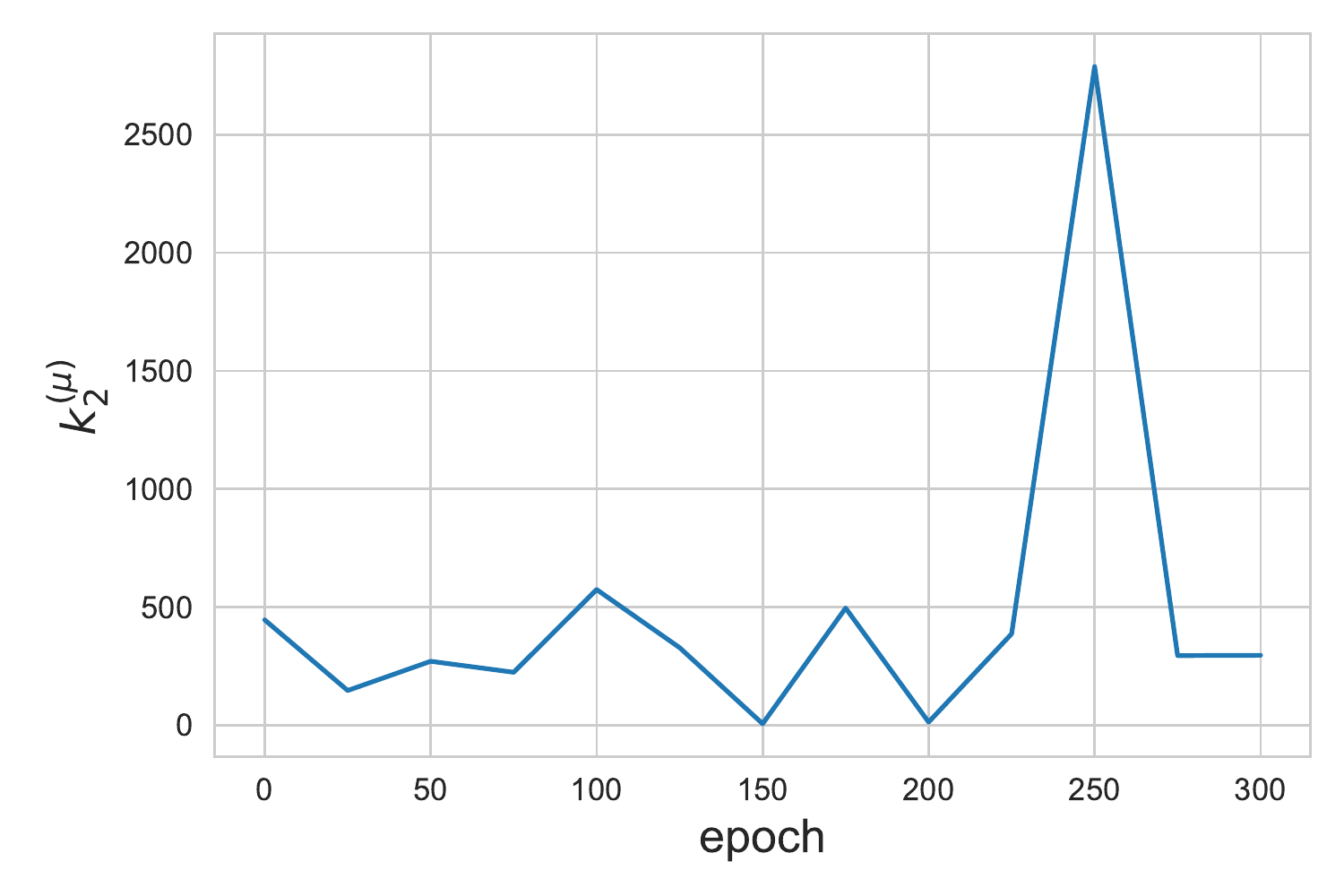}
     \subcaption{$m_2^{(\mu)}$, Resnet on CIFAR100}
    \end{subfigure}
    \begin{subfigure}{0.3\linewidth}
     \centering
     \includegraphics[width=\linewidth]{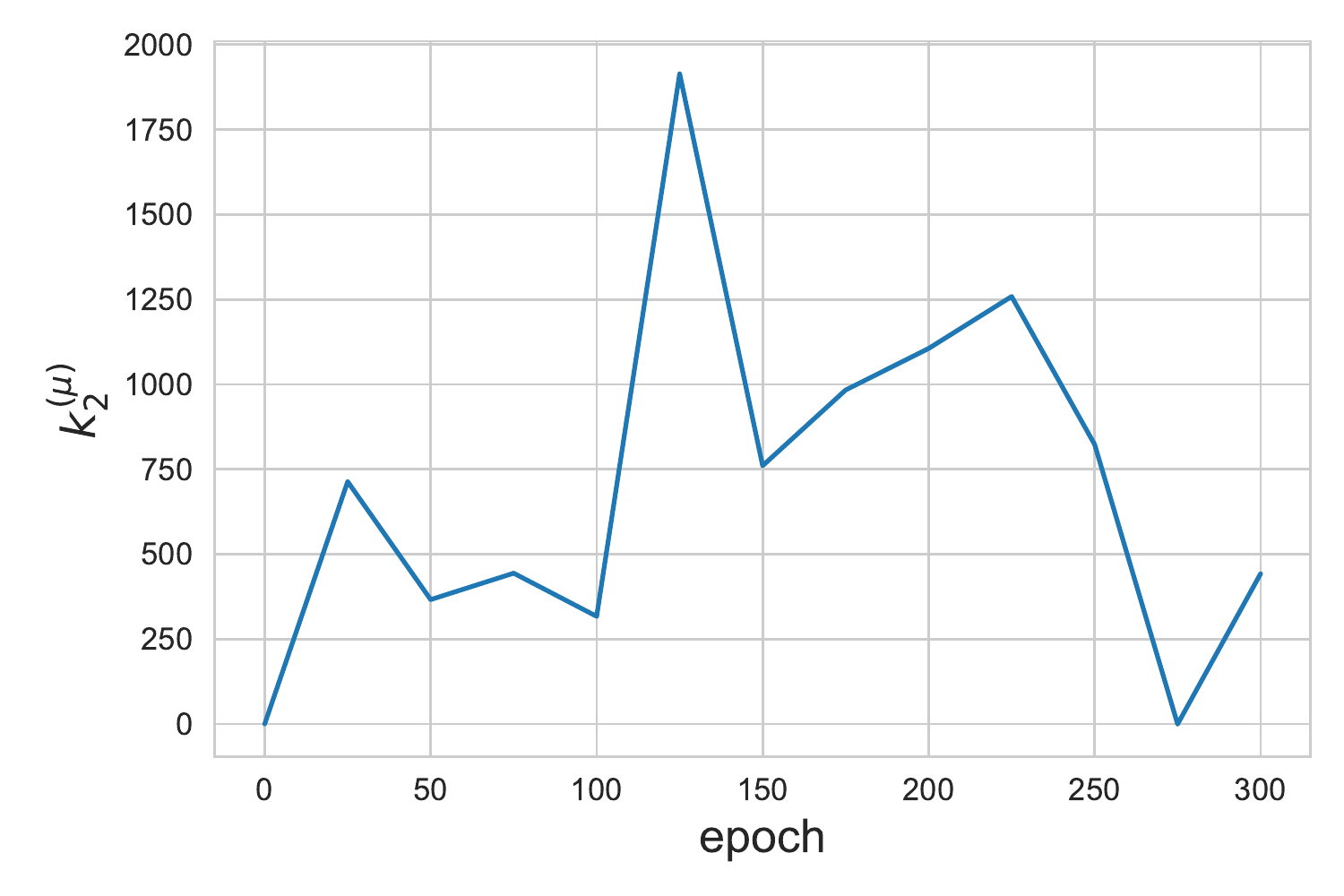}
     \subcaption{$m_2^{(\mu)}$, VGG16 on CIFAR100}
    \end{subfigure}
    \begin{subfigure}{0.3\linewidth}
     \centering
     \includegraphics[width=\linewidth]{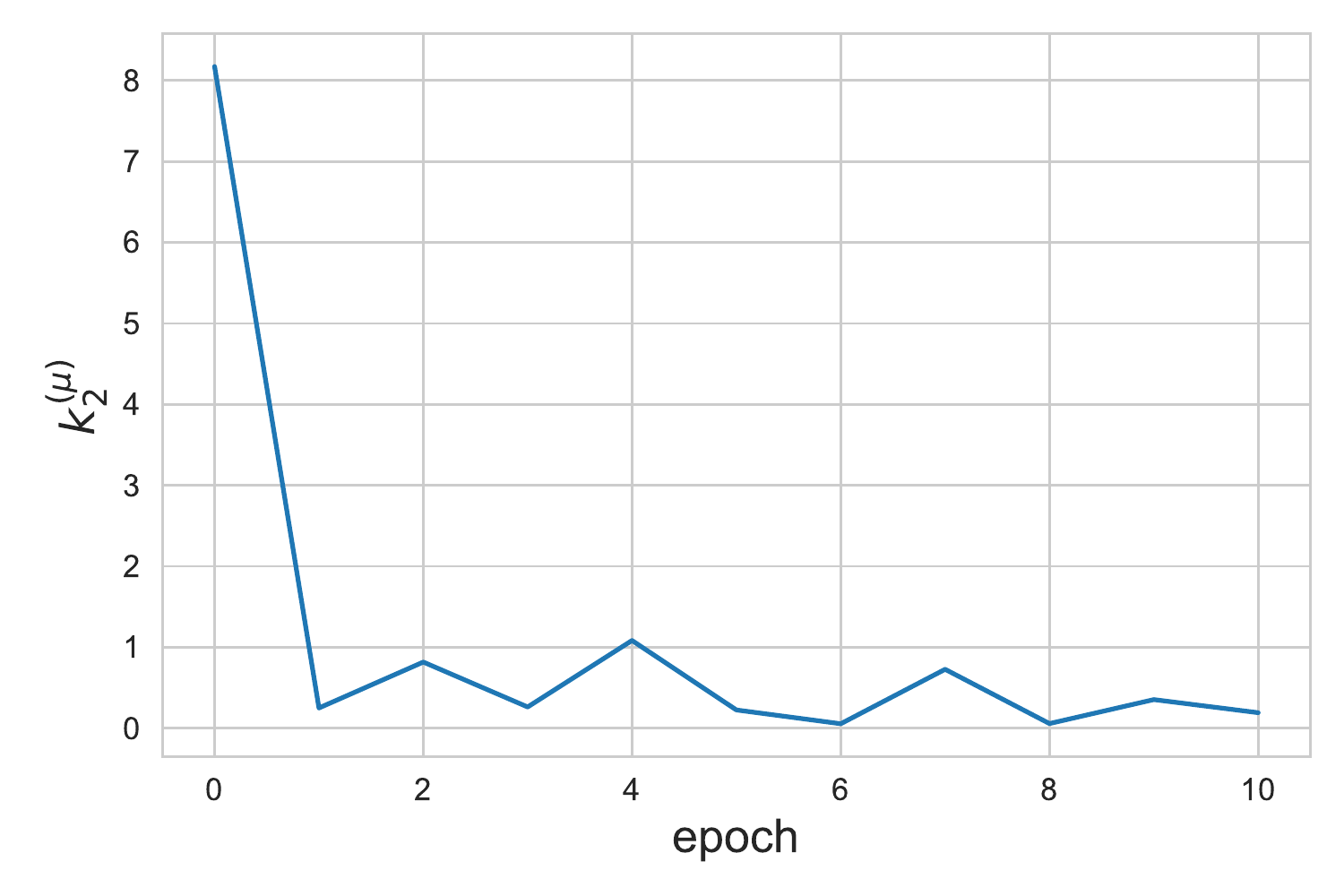}
     \subcaption{$m_2^{(\mu)}$, MLP on MNIST}
    \end{subfigure}
    
    \begin{subfigure}{0.3\linewidth}
     \centering
     \includegraphics[width=\linewidth]{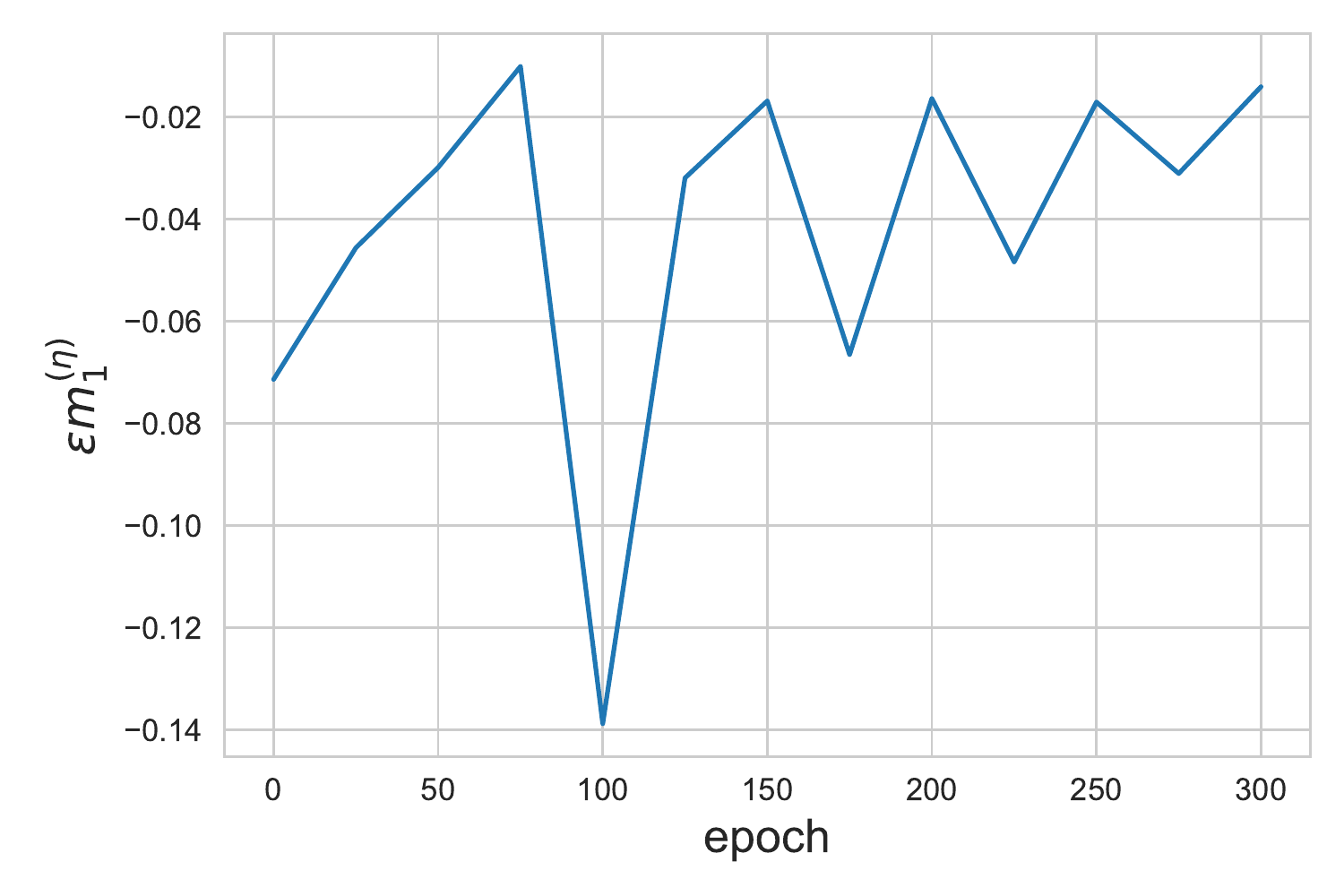}
     \subcaption{$\epsilon m_1^{(\eta)}$, Resnet on CIFAR100}
    \end{subfigure}
    \begin{subfigure}{0.3\linewidth}
     \centering
     \includegraphics[width=\linewidth]{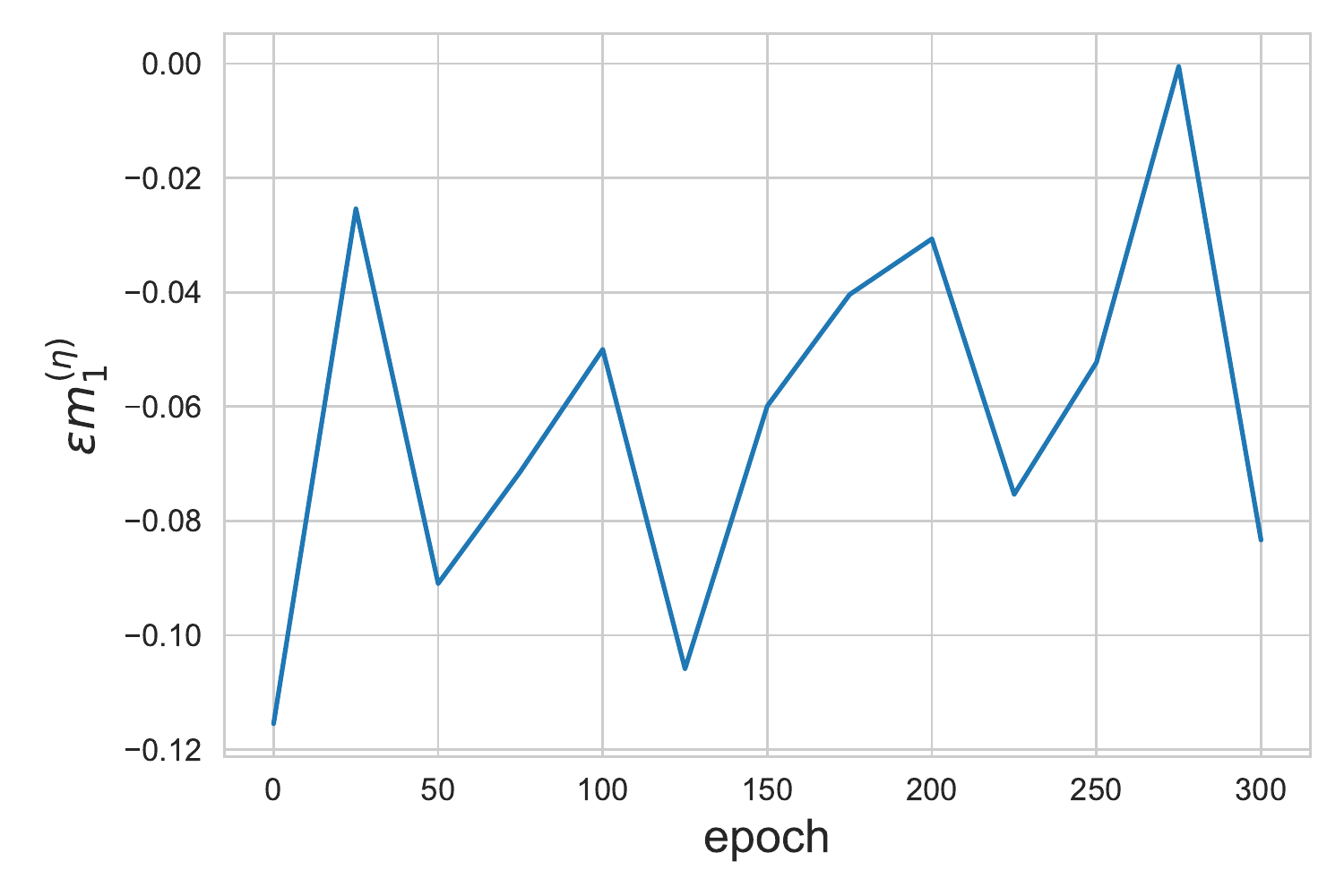}
     \subcaption{$\epsilon m_1^{(\eta)}$, VGG16 on CIFAR100}
    \end{subfigure}
    \begin{subfigure}{0.3\linewidth}
     \centering
     \includegraphics[width=\linewidth]{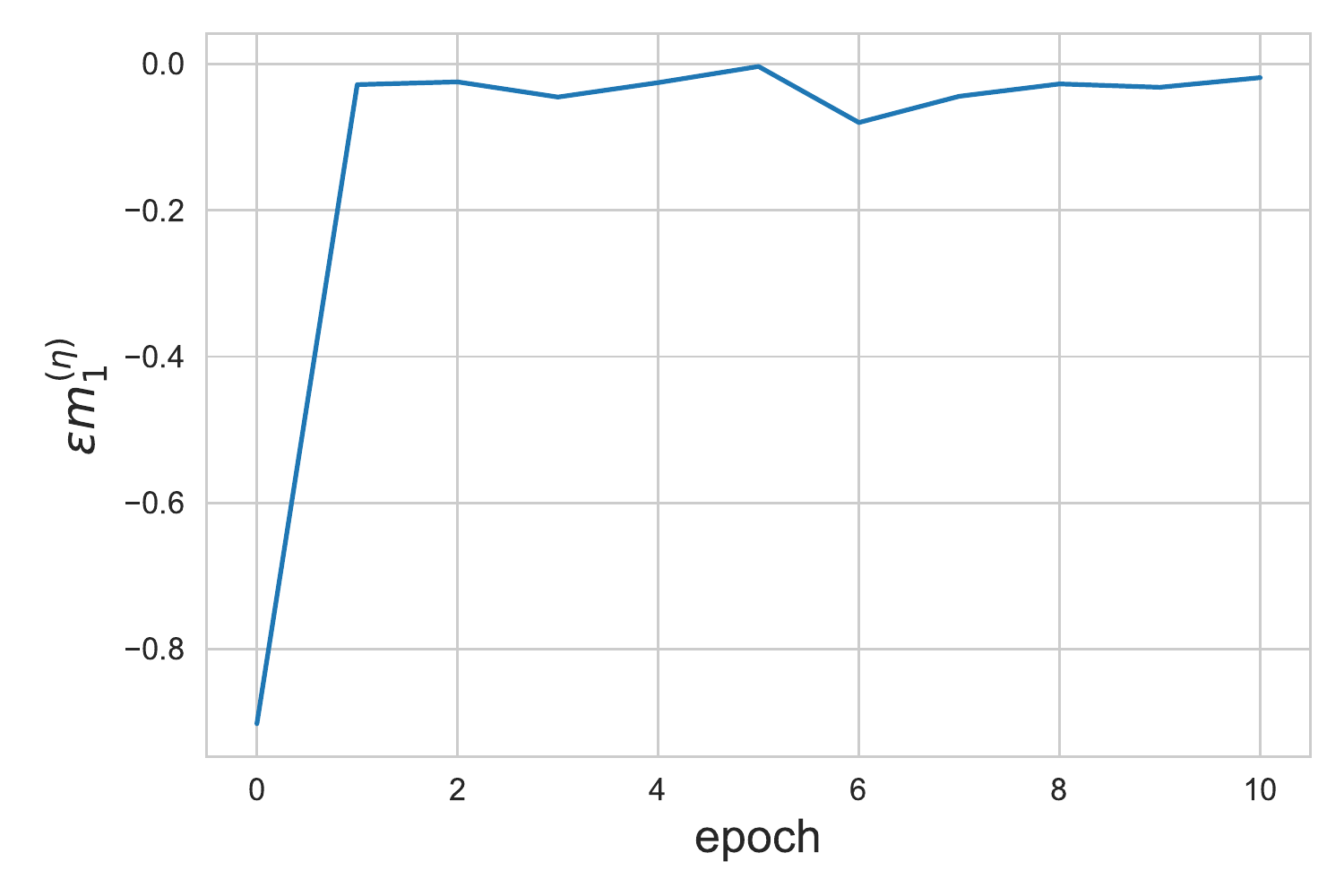}
     \subcaption{$\epsilon m_1^{(\eta)}$, MLP on MNIST}
    \end{subfigure}
    
    \begin{subfigure}{0.3\linewidth}
     \centering
     \includegraphics[width=\linewidth]{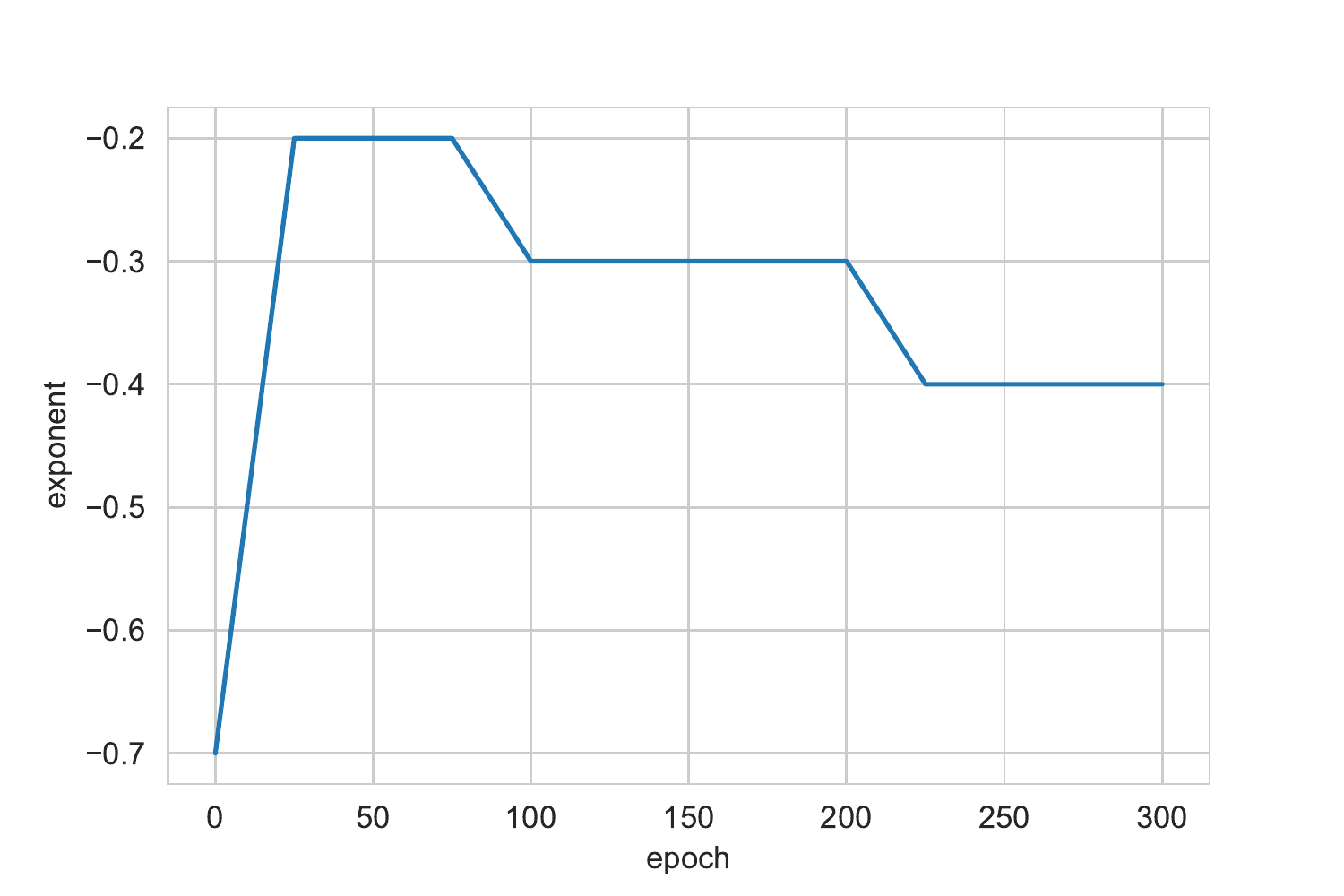}
     \subcaption{Exponent $\upsilon$, Resnet on CIFAR100}
    \end{subfigure}
    \begin{subfigure}{0.3\linewidth}
     \centering
     \includegraphics[width=\linewidth]{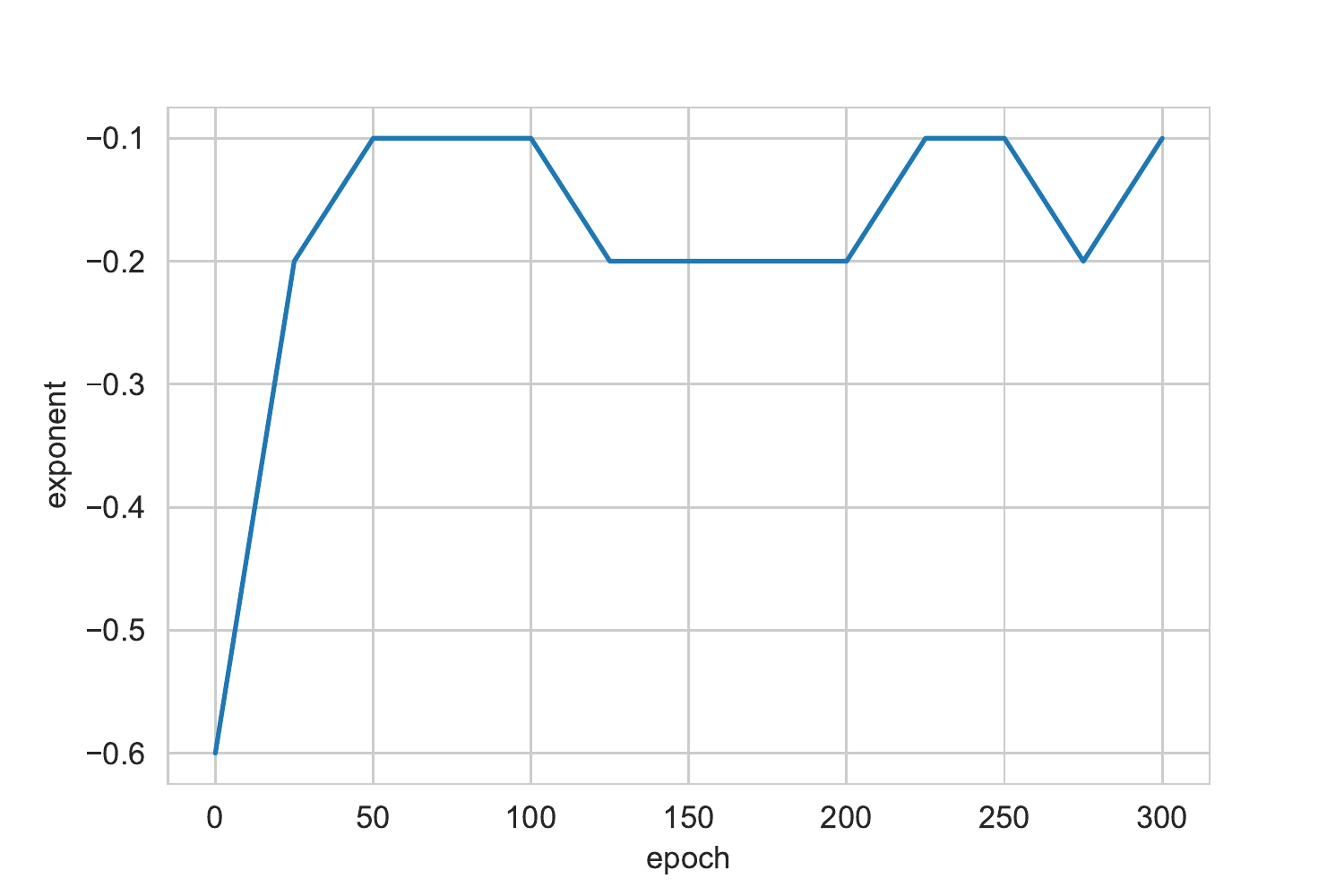}
     \subcaption{Exponent $\upsilon$, VGG16 on CIFAR100}
    \end{subfigure}
    \begin{subfigure}{0.3\linewidth}
     \centering
     \includegraphics[width=\linewidth]{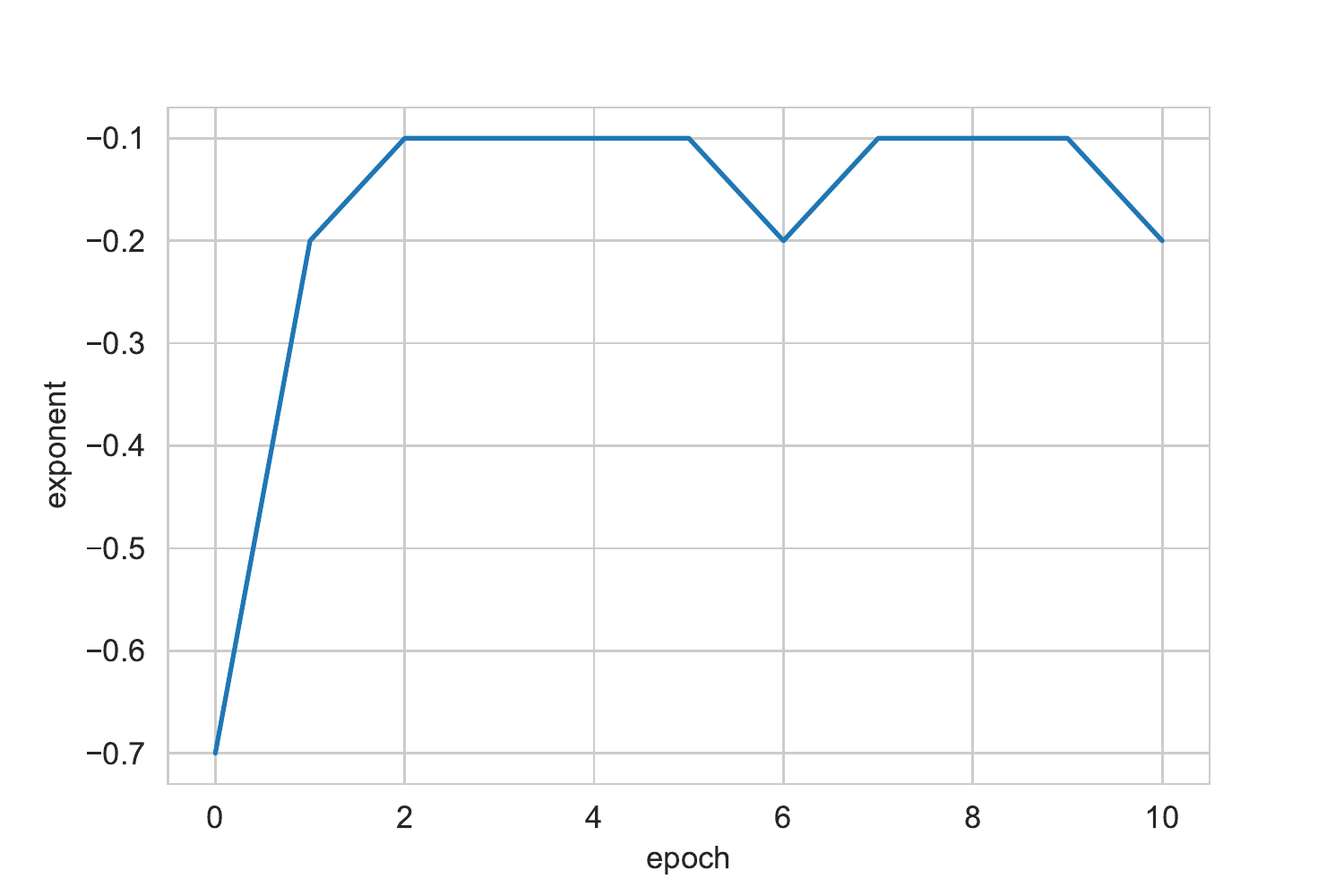}
     \subcaption{Exponent $\upsilon$, MLP on MNIST}
    \end{subfigure}    
    
    \begin{subfigure}{0.3\linewidth}
     \centering
     \includegraphics[width=\linewidth]{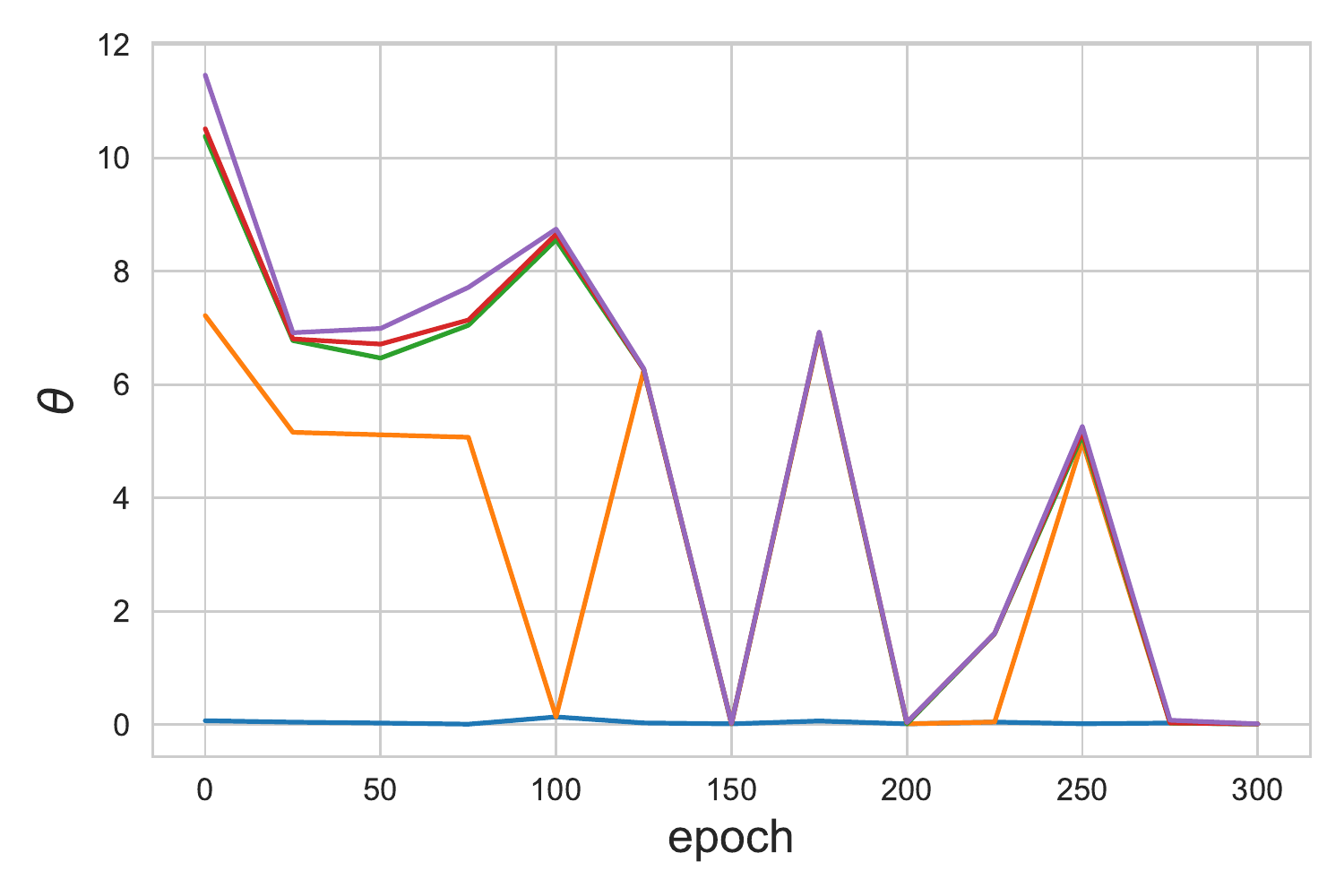}
     \subcaption{$\theta_1,\ldots, \theta_5$, Resnet on CIFAR100}
    \end{subfigure}
    \begin{subfigure}{0.3\linewidth}
     \centering
     \includegraphics[width=\linewidth]{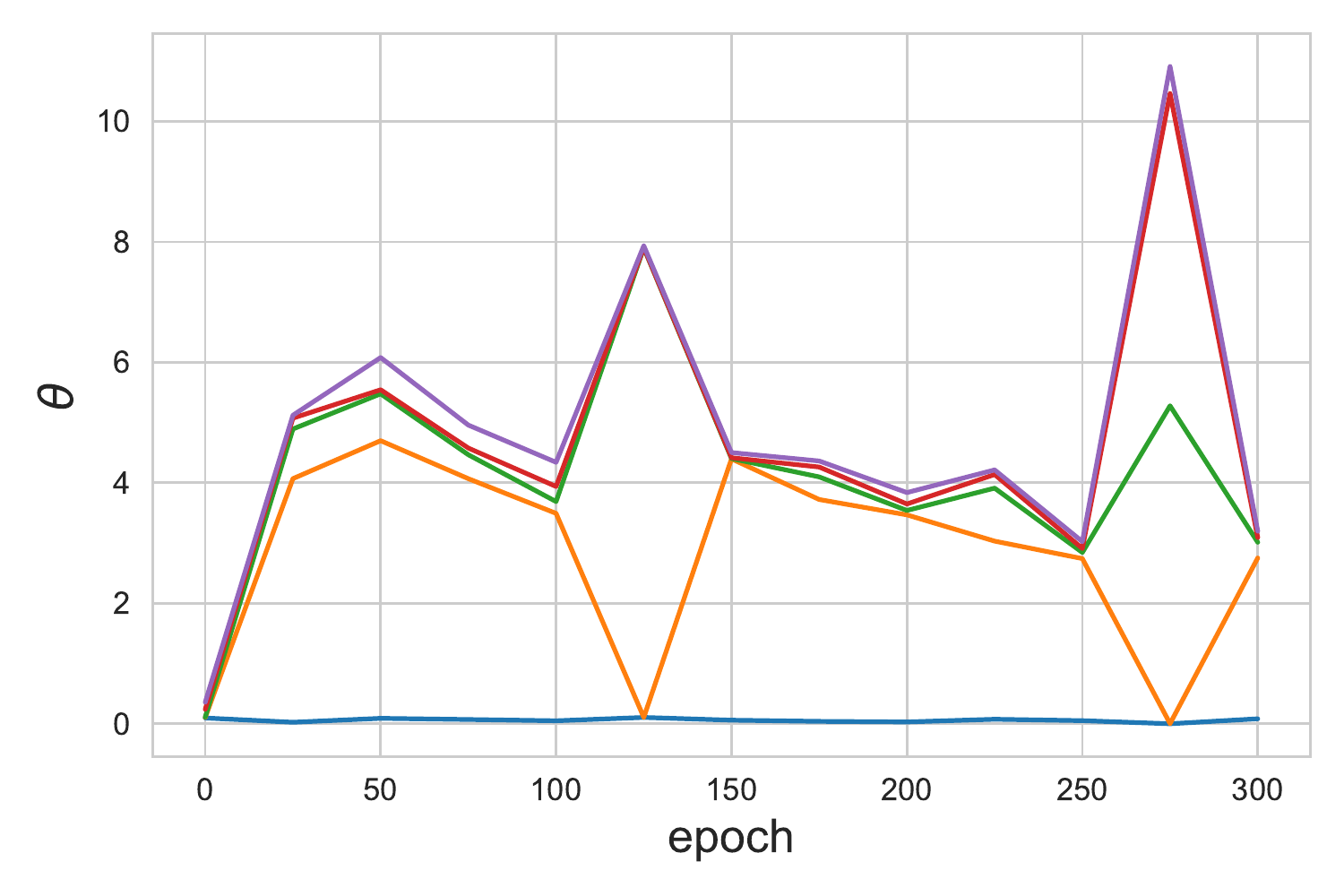}
     \subcaption{$\theta_1,\ldots, \theta_5$, VGG16 on CIFAR100}
    \end{subfigure}
    \begin{subfigure}{0.3\linewidth}
     \centering
     \includegraphics[width=\linewidth]{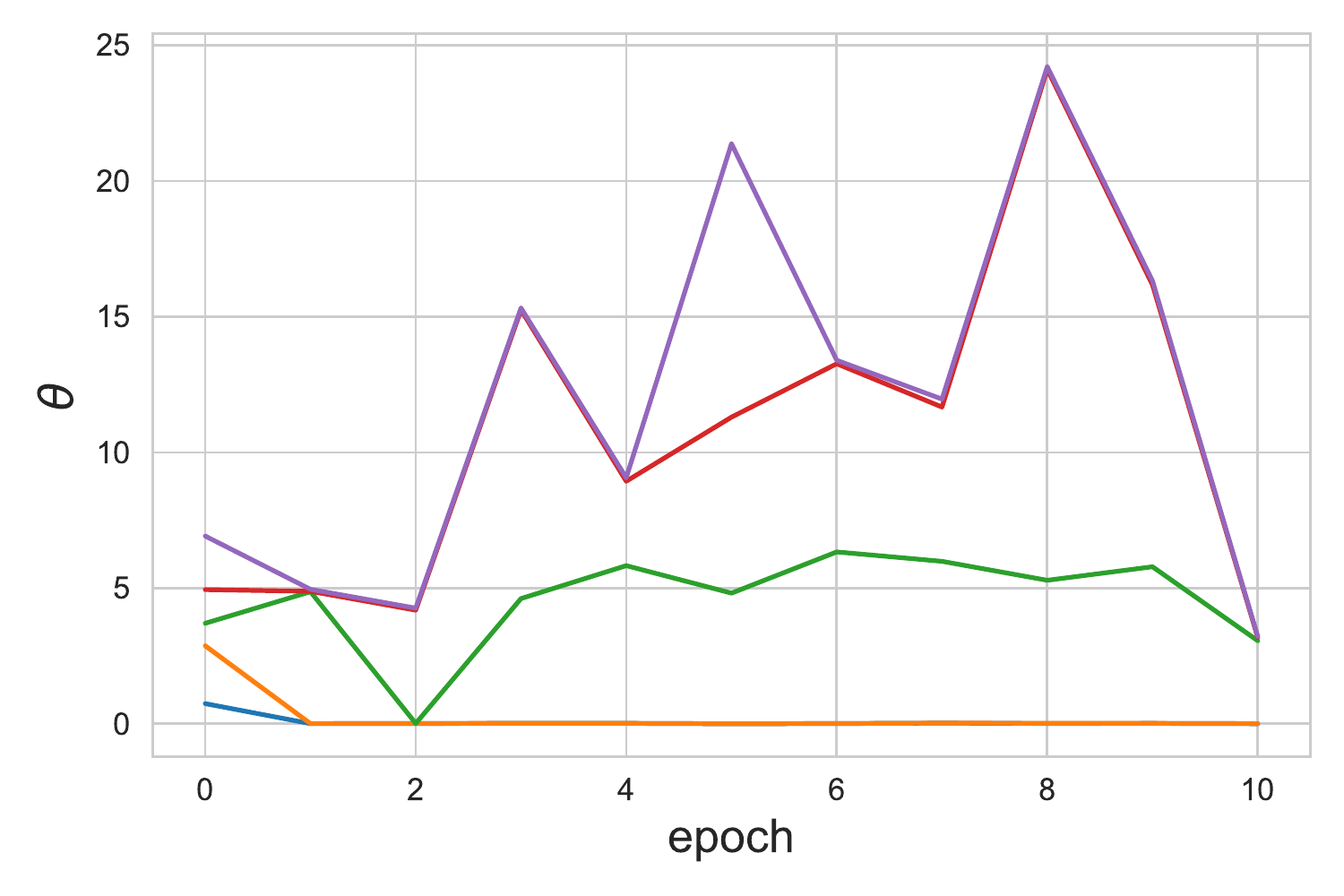}
     \subcaption{$\theta_1,\ldots, \theta_5$, MLP on MNIST}
    \end{subfigure}
    \centering
    \caption{The blue lines are parametric power law fits of the form (\ref{eq:omega_fit_form}). }
    \label{fig:outlier_fit_summary_params}
\end{figure}


\subsection{Justification and motivation of QUE}\label{subsec:que_just}
Work on random matrix universality has shown that a \emph{local law} is the key ingredient in the establishing universal local statistics of  eigenvalues \cite{erdHos2012bulk, erdos2017dynamical} and universal delocalisation of eigenvectors \cite{bourgade2017eigenvector}.
There are several forms of local law, but all make control, with high probability, the error between the (random) matrix Green's function $G(z) = (z - X)^{-1}$ and certain deterministic equivalents.
In all cases we use the set 
\begin{align}
    \vec{S} = \left\{E + i\eta \in \C \mid |E| \leq \omega^{-1}, ~ N^{-1 + \omega} \leq \eta \leq \omega^{-1}\right\}
\end{align}
for $\omega\in(0, 1)$ and the local law statements holds for all (large) $D>0$ and (small) $\xi > 0$ and for all large enough $N$.
The \emph{averaged local law} states:
\begin{align}
   \sup_{z\in\vec{S}} \P\left(\left|\frac{1}{N}\Tr G(z) - g_{\mu}(z)\right| > N^{\xi}\left(\frac{1}{N\eta} + \sqrtsign{\frac{\Im g_{\mu}(z)}{N\eta}}\right)\right) \leq N^{-D}.
\end{align}
The \emph{isotropic local law} states:
\begin{align}\label{eq:iso_local_law}
    \sup_{\|\vec{u}\|,\|\vec{v}\|  = 1, z\in\vec{S}}\P\left( |\vec{u}^TG(z)\vec{v} - g_{\mu}(z)| > N^{\xi}\left(\frac{1}{N\eta} + \sqrtsign{\frac{\Im g_{\mu}(z)}{N\eta}}\right)\right) \leq N^{-D}.
\end{align}
The \emph{anisotropic local law} states:
\begin{align}\label{eq:aniso_local_law}
    \sup_{\|\vec{u}\|,\|\vec{v}\|  = 1, z\in\vec{S}}\P\left( |\vec{u}^TG(z)\vec{v} - \vec{u}^T\Pi(z)\vec{v}| > N^{\xi}\left(\frac{1}{N\eta} + \sqrtsign{\frac{\Im g_{\mu}(z)}{N\eta}}\right)\right) \leq N^{-D}
\end{align}
where $\Pi(\cdot)$ is an $N\times N$ deterministic matrix function on $\mathbb{C}$.
The \emph{entrywise local law} states:
\begin{align}\label{eq:entry_local_law}
    \sup_{z\in\vec{S}, 1\leq i,j\leq N}\P\left( |G_{ij}(z) - \Pi_{ij}(z)| > N^{\xi}\left(\frac{1}{N\eta} + \sqrtsign{\frac{\Im g_{\mu}(z)}{N\eta}}\right)\right) \leq N^{-D}.
\end{align}

The anisotropic local law is a stronger version of the entrywise local law.
The anisotropic local law is a more general version of the isotropic local law, which can be recovered in the isotropic case by taking $\Pi = g_{\mu} I$.
The entrywise local law can also be applied in the isotropic case by taking $\Pi = g_{\mu} I$.
The averaged local law is weaker than all of the other laws.
General Wigner matrices are known to obey isotropic local semi-circle laws \cite{erdHos2017universality}.
Anisotropic local laws are known for general deformations of Wigner matrices and general covariance matrices \cite{knowles2017anisotropic} as well as quite general classes of correlated random matrices \cite{erdHos2019random}.

\medskip
As mentioned above, quantum unique ergodicity was proved for general Wigner matrices in \cite{bourgade2017eigenvector}.
It appears that the key ingredient in the proof of QUE (\ref{eq:que_def}) in \cite{bourgade2017eigenvector} is the isotropic local semicircle law (\ref{eq:iso_local_law}) for general Wigner matrices.
Indeed, all the intermediate results in Sections 4 of \cite{bourgade2017eigenvector} take only (\ref{eq:iso_local_law}) and general facts about the Dyson Brownian Motion eigenvector flow given by \begin{align}
    d\lambda_k &= \frac{dB_{kk}}{\sqrtsign{N}} + \left(\frac{1}{N}\sum_{\ell\neq k} \frac{1}{\lambda_k - \lambda_{\ell}}\right)dt,\\
    du_k &= \frac{1}{\sqrtsign{N}}\sum_{\ell\neq k}\frac{dB_{kl}}{\lambda_k - \lambda_{\ell}} u_{\ell} - \frac{1}{2N} \sum_{\ell\neq k}\frac{dt}{(\lambda_k - \lambda_{\ell})^2} u_{k}.
\end{align}
This can be generalised to 
\begin{align}
    d\lambda_k &= \frac{dB_{kk}}{\sqrtsign{N}} + \left(-V(\lambda_i) + \frac{1}{N}\sum_{\ell\neq k} \frac{1}{\lambda_k - \lambda_{\ell}}\right)dt,\\
    du_k &= \frac{1}{\sqrtsign{N}}\sum_{\ell\neq k}\frac{dB_{kl}}{\lambda_k - \lambda_{\ell}} u_{\ell} - \frac{1}{2N} \sum_{\ell\neq k}\frac{dt}{(\lambda_k - \lambda_{\ell})^2} u_{k}.
\end{align}
where $V$ is a potential function.
Note that the eigenvector dynamics are unaffected by the presence of the potential $V$, so we expect to be able to generalise the proof of \cite{knowles2017anisotropic} to any random matrix ensemble with an isotropic local law by defining the potential $V$ so that the invariant ensemble with distribution $Z^{-1}e^{-N \Tr V(X)}dX$ has equilibrium measure $\mu$ ($Z$ is a normalisation constant).
We show how to construct such a $V$ from $\mu$ in Section \ref{sec:inv}.



The arguments so far suffice to justify a generalisation of the ``dynamical step'' in the arguments of \cite{bourgade2017eigenvector}, so it remains to consider the ``comparison step''. The dynamical step establishes QUE for the matrix ensemble with a small Gaussian perturbation, but in the comparison step one must establish that the perturbation can be removed without breaking QUE.
To our knowledge no such argument has been articulated beyond generalized Wigner matrices, with the independence of entries and comparable scale of variances being critical to the arguments given by \cite{bourgade2017eigenvector}.
Our guiding intuition is that QUE of the form (\ref{eq:que_def}) is a general property of random matrices and can reasonably be expected to hold in most, if not all, cases in which there is a local law and universal local eigenvalue statistics are observed.
At present, we are not able to state a precise result establishing QUE in sufficient generality to be relevant for this work, so we shall take it as an assumption.
\begin{assump}\label{thm:que_gaussian}
    Let $X$ be an ensemble of $N\times N$ real symmetric random matrices. Assume that $X$ admits a limiting spectral measure is $\mu$ with Stieljtes transform $m$. Suppose that the isotropic local law (\ref{eq:iso_local_law}) holds for $X$ with $\mu$.
    Then there is some set $\mathbb{T}_N \subset [N]$ with $|\mathbb{T}_N^c| = o(N)$ such that with $|I|=n$, for any polynomial $P$ in $n$ indeterminates, there exists some $\epsilon(P) > 0$ such that for large enough $N$ we have
    \begin{align}
        \sup_{\substack{I \subset \mathbb{T}_N, |I|=n,\\ \|\vec{q}\| = 1}} \left| \E\left(P\left(\left(N(\vec{q}^T\vec{u}_k)^2\right)_{k\in I}\right)\right) - \E\left(P\left(\left(|\mathcal{N}_j|^2\right)_{k\in I}\right)\right)\right| \leq N^{-\epsilon}.
    \end{align}
\end{assump}
Note that the isotropic local law in Assumption \ref{thm:que_gaussian} can be obtained from the weaker entrywise law (\ref{eq:entry_local_law}) as in Theorem 2.14 of \cite{alex2014isotropic} provided there exists a $C>0$ such that $\E|X_{ij}|^2 \leq CN^{-1}$ for all $i,j$ and there exists $C_p > 0$ such that $\E |\sqrtsign{N}X_{ij}|^p \leq C_p$ for all $i,j$ and integer $p>0$.

\begin{rem}
In \cite{bourgade2017eigenvector} the restriction $I\subset\TN$ is given for the explicit set 
\begin{align}
\mathbb{T}_N = [N] \backslash \{(N^{1/4}, N^{1-\delta}) \cup (N - N^{1-\delta}, N - N^{1/4})\}
\end{align}
for some $0< \delta < 1$.
In the case of generalised Wigner matrices, this restriction on the indices has since been shown to be unnecessary \cite{benigni2022optimal, benigni2020eigenvectors,benigni2021fluctuations,benigni2021fluctuations}.
In our context, we could simply take as an assumption all results holds with $\mathbb{T}_N = [N]$, however our results can in fact be proved using only the above assumption that $|\mathbb{T}_N^c| = o(N)$, so we shall retain this weaker form of the assumptions.
\end{rem}

This section is not intended to prove QUE from explicit known properties of deep neural network Hessians, but rather to provide justification for it as a reasonable modeling assumption in the noise model for Hessians defined in section \ref{subsec:hess_model}.
We have shown how QUE can be obtained from an isotropic (or entrywise) local law beyond the Wigner case.
It is important to go beyond Wigner or any other standard random matrix ensemble, as we have observed above that the standard macroscopic spectral densities of random matrix theory such as the semicircle law are not observed in practice.
That said, we are not aware of any results establishing QUE in the more general case of anistropic local laws, and this appears to be a very significant technical challenge.
We must finally address why a local law assumption, isotropic or otherwise, may be reasonable for the noise matrix $X$ in our Hessian model.
Over the last decade or so, universal local statistics of random matrices in the form of $k$-point correlation function on the appropriate microscopic scale have been established for a litany of random matrix ensembles.
An immediate consequence of such results is that, on the scale of unit mean eigenvalue spacing, Wigner's surmise holds, depending only on the symmetry class (orthogonal, unitary or symplecitic).
As with the QUE proof discussed above, the key ingredient in these proofs, as part of the ``three step strategy'' \cite{erdos2017dynamical}, is establishing a local law.
The theoretical picture that has emerged is that when universal local eigenvalue statistics are observed in random matrices, it is due to the mechanism of short time scale relaxation of local statistics under Dyson Brownian Motion made possible by a local law.
It has been observed that universal local eigenvalue statistics do indeed appear to be present in the Hessian of real, albeit quite small, deep neural networks \cite{baskerville2022appearance}.
Given all of this context, we propose that a local law assumption of some kind is reasonable  for deep neural network Hessians and not particularly restrictive.
As we have shown, if we are willing to make the genuinely restrictive assumption of an isotropic local law for the Hessian noise model, then QUE follows.
However an anistropic local law is arguably more plausible as we expect deep neural networks Hessians to contain a good deal of dependence between entries, and such correlations are know to generically lead to anisotropic local laws \cite{erdHos2019random}.

\subsection{Motivation of true Hessian structure}
In this section we revisit and motivate the assumptions made about the Hessian in Section \ref{subsec:hess_model}.
Firstly note that one can always define $A = \E H_{\text{batch}}$ and it is natural then to associate $A$ with the true Hessian $H_{\text{true}}$.
In light of (\ref{eq:batch_hessian_def}), it is natural to expect some fixed form of the law for $H_{\text{batch}} - A$ for any batch size, but with an overall scaling $\sbf$, which must naturally be decreasing in $b$.
Next we address the assumptions made about the spectrum of $A$.
The first assumption one might think to make is that $A$ has fixed rank relative to $N$, with spectrum consisting only of the spikes $\theta_i, \theta_j'$.
Behind such an assumption is the intuition that the data distribution does not depend on $N$ and so, in the over-parametrised limit $N\rightarrow\infty$, the overwhelming majority of directions in weight space are unimportant.
The form we take for $A$ in the above is a strict generalisation of the fixed rank assumption; $A$ still has a fixed number of spiked directions, but the parameter $\epsilon$ controls the rank of $A$.
Since any experimental investigation is necessarily limited to $N<\infty$, the generalisation to $N>0$ is particularly important.
Compact support of the measures $\mu$ and $\eta$ is consistent with experimental observations of deep neural network Hessian spectra.

\subsection{The batch size scaling}
Our experimental results considered $\sbf = b^{-1/2}$, a choice which we now justify.
From (\ref{eq:batch_hessian_def}) we have \begin{align}
    H_{\text{batch}} = \frac{1}{b}\sum_{i=1}^b\left(H_{\text{true}} + X^{(i)}\right)
\end{align}
where $X^{(i)}$ are i.i.d. samples from the law of $X$.
Suppose that the entries $X_{ij}$ were Gaussian, with $\text{Cov}(X_{ij}, X_{kl}) = \Sigma_{ij,kl}$.
Then $Z = X^{(p)}_{ij} + X^{(q)}_{ij}$ has 
\begin{align}
    \text{Cov}(Z_{ij}, Z_{kl}) = \E X_{ij}^{(p)}X_{kl}^{(p)} + \E X_{ij}^{(q)}X_{kl}^{(q)} -  \E X_{ij}^{(p)}\E X_{kl}^{(p)} - \E X_{ij}^{(q)}\E X_{kl}^{(q)}  = 2\Sigma_{ij, kl}.
\end{align}
In the case of centred $X$, one then obtains\begin{align}
    \frac{1}{b}\sum_{i=1}^b X^{(i)} \overset{d}{=} b^{-1/2}X.
\end{align}
Note that this does not quite match the case described in Section \ref{subsec:hess_model}, since we do not there assume $\E X = 0$, however we take this a rough justification for $\sbf = b^{-1/2}$ as an ansatz.
Moreover, numerical experimentation with $\sbf = b^{-q}$ for values of $q>0$ shows that $q = 1/2$ gives roughly the best fit to the data.
\FloatBarrier
\section{Spectral free addition from QUE}\label{sec:que}

\subsection{Intermediate results on QUE}\label{sec:fourier}
This section establishes some intermediate results that follow from assuming QUE for the eigenvectors of a matrix. They will be crucial for our application in the following section. 

\begin{lem}\label{lem:que_preserved_under_rotation}
    Consider a real orthogonal $N\times N$ matrix $U$ with rows $\{\vec{u}_i^T\}_{i=1}^N$. Assume that $\{\vec{u}_i\}_{i=1}^N$ are the eigenvectors of a real random symmetric matrix with QUE. Let $P$ be a fixed $N\times N$ real orthogonal matrix. Let $V = UP$ and denote the rows of $V$ by $\{\vec{v}_i^T\}_{i=1}^N$. Then $\{\vec{v}_i\}_{i=1}^N$ also satisfy QUE.
\end{lem}
\begin{proof}
   Take any unit vector $\vec{q}$, then for any $k=1,\ldots, N$
   \begin{align*}
       \vec{q}^T \vec{v}_k = \sum_{j}q_jV_{kj} &= \sum_{j, l}q_j U_{kl}P_{lj} = (P\vec{q})^T \vec{u}_k.
   \end{align*}
   But $\|P\vec{q}\|_2=\|\vec{q}\|_2=1$ since $P$ is orthogonal, so the statement of QUE for $\{\vec{u}_i\}_{i=1}^N$   transfers directly to $\{\vec{v}_i\}_{i=1}^N$ thanks to the supremum of all unit $\vec{q}$.
\end{proof}

\begin{lem}\label{lem:weaker_col_que}
    Consider a real orthogonal $N\times N$ matrix $U$ with rows $\{\vec{u}_i^T\}_{i=1}^N$. Assume that $\{\vec{u}_i\}_{i=1}^N$ are the eigenvectors of a real random symmetric matrix with QUE. Let $\ell_0(\vec{q}) = \sum_i \1\{q_i \neq 0\}$ count the non-zero elements of a vector with respect to a fixed orthonormal basis $\{\vec{e}_i\}_{i=1}^N$. 
    For any fixed integer $s>0$, define the set \begin{align}
            \Vs = \left\{\vec{q}\in \R^N \mid \|\vec{q}\|=1, ~ \ell_0(\vec{q})=s, ~ q_i=0 ~\forall i\in \TN^c\right\}.
    \end{align}
    Then the columns $\{\vec{u}_i'\}_{i=1}^N$ of $U$ satisfy a weaker form of QUE (for any fixed $n, s>0$):
    \begin{align}\label{eq:que_weak_def}
    \sup_{\substack{\vec{q}\in\Vs \\ }}\sup_{\substack{I \subset \TN\\ |I| = n}} \left|\E P\left(\left(N|\vec{q}^T\vec{u}_k|^2\right)_{k\in I}\right) - \E P\left(\left(|\mathcal{N}_j|^2\right)_{j=1}^m
    \right)\right| \leq N^{-\epsilon}.
\end{align}
We will denote this form of QUE as $\hQUE$.
\end{lem}
\begin{proof}
   Take some $\vec{q}\in \Vs$. Then there exists some $J\subset\TN$ with $|J|=s$ and non-zero $\{q_k\}_{k\in J}$ such that\begin{align*}
       \vec{q}^T\vec{u}_k' = \sum_{j\in J} q_j \vec{e}_j^T \vec{u}_k'.
   \end{align*}
   Take $\{\vec{e}_i\}_{i=1}^N$ to be a standard basis with $(\vec{e}_i)_j = \delta_{ij}$, then $\vec{e}_j^T \vec{u}_k' = U_{jk} = \vec{e}_k^T\vec{u}_j$ so \begin{align*}
       \vec{q}^T\vec{u}_k' = \sum_{j\in J} q_j \vec{e}_k^T\vec{u}_j
   \end{align*}
   but then the coefficients $q_j$ can be absorbed into the definition of the general polynomial in the statement (\ref{eq:que_def}) of QUE for $\{\vec{u}_i\}_{i=1}^N$, which completes the proof, noting that the sum only includes indices contained in $\TN$ owing to the definition of $\Vs$.
\end{proof}

\begin{lem}\label{lem:explicit_error}
    Fix some real numbers $\{y_i\}_{i=1}^r$. Fix also a diagonal matrix $\Lambda$ and an orthonormal set of vectors $\{\vec{v}_i\}_{i=1}^N$ that satisfies $\hQUE$. Then there exists an $\epsilon>0$ and $\vec{\eta}_i\in\C^N$ with
      \begin{align}
       \eta_{ij}^2 &\in [-1, 1] ~ \forall j\in\TN,\\
       \eta_{ij}^2 &\in [-N^{\epsilon}, N^{\epsilon}] ~ \forall j\in\TN^c.
   \end{align}
 such that for any integer $l>0$\begin{align}\label{eq:explicit_error_lem1}
        \mathbb{E} \left(\sum_{i=1}^r y_i \vec{v}_i^T\Lambda \vec{v}_i\right)^l - \mathbb{E}\left(\sum_{i=1}^r y_i \frac{1}{N}\vec{g}_i^T\Lambda \vec{g}_i\right)^l = N^{-(1+\epsilon)l} \left( \sum_{i=1}^r y_i \vec{\eta}_i^T\Lambda \vec{\eta}_i\right)^l
    \end{align}
    where the $\vec{g}_i$ are i.i.d. Gaussians $N(0, I_N)$. 
    
\end{lem}
\begin{proof}
   Let $\{\vec{e}_i\}_{i=1}^N$ be the standard orthonormal basis from above. Then 
   \begin{align}
       \E \left(\sum_{i=1}^r y_i \vec{v}_i^T\Lambda \vec{v}_i\right)^l &= \E \sum_{i_1,\ldots, i_l=1}^r\prod_{k=1}^{l} y_{i_k} \vec{v}_{i_k}^T\Lambda\vec{v}_{i_k}\notag\\
       &= \E \sum_{i_1,\ldots, i_l=1}^r\sum_{j_1,\ldots, j_l=1}^N\prod_{k=1}^{l} y_{i_k} \lambda_{j_k} (\vec{e}_{j_k}^T\vec{v}_{i_k})^2\label{eq:que_error_bound1}\\
       \implies \E \left(\sum_{i=1}^r y_i \vec{v}_i^T\Lambda \vec{v}_i\right)^l-\mathbb{E}\left(\sum_{i=1}^r y_i \frac{1}{N}\vec{g}_i^T\Lambda \vec{g}_i\right)^l &= N^{-l}\sum_{i_1,\ldots, i_l=1}^r\sum_{j_1,\ldots, j_l=1}^N\prod_{k=1}^{l} y_{i_k} \lambda_{j_k} \left[N\E(\vec{e}_{j_k}^T\vec{v}_{i_k})^2 - \E (\vec{e}_{j_k}^T\vec{g}_{i_k})^2\right]\notag\\
       &=N^{-l}\sum_{i_1,\ldots, i_l=1}^r\sum_{j_1,\ldots, j_l\in\TN}\prod_{k=1}^{l} y_{i_k} \lambda_{j_k} \left[N\E(\vec{e}_{j_k}^T\vec{v}_{i_k})^2 - \E (\vec{e}_{j_k}^T\vec{g}_{i_k})^2\right]\notag\\
       &+ N^{-l}\sum_{i_1,\ldots, i_l=1}^r\sum_{\substack{j_1\in\TN^c,\\ j_2,\ldots, j_l\in\TN}}\prod_{k=1}^{l} y_{i_k} \lambda_{j_k} \left[N\E(\vec{e}_{j_k}^T\vec{v}_{i_k})^2 - \E (\vec{e}_{j_k}^T\vec{g}_{i_k})^2\right]\notag\\
       &+ \ldots
   \end{align}
   The ellipsis represents the similar terms where further of the $j_1, \ldots, j_r$ are in $\TN^c$.
   For $j\in\TN^c$ the terms \begin{align}
       \left[N\E(\vec{e}_{j_k}^T\vec{v}_{i_k})^2 -\E (\vec{e}_{j_k}^T\vec{g}_{i_k})^2\right]
   \end{align}
   are excluded from the statement of $\hQUE$, however we can still bound them crudely.
   Indeed \begin{align}
       \sum_{j\in\TN^c} N(\vec{e}_j^T\vec{v}_i)^2 = \sum_{j=1}^N N(\vec{e}_j^T\vec{v}_i)^2 - \sum_{j\in\TN}N(\vec{e}_j^T\vec{v}_i)^2 = N - \sum_{j\in\TN}N(\vec{e}_j^T\vec{v}_i)^2
   \end{align}
   but since the bound of $\hQUE$ applies for $j\in\TN$ \begin{align}
       N\E (\vec{e}_j^T\vec{v}_i)^2 = \E(\vec{e}_j^T\vec{g})^2 + o(1) = 1 + o(1) ~~ \forall j\in\TN,
   \end{align}
   then \begin{align}
       \sum_{j\in\TN^c} N(\vec{e}_j^T\vec{v}_i)^2 = N - N(1 + o(1)) = o(N) ~ \implies ~  \E(\vec{e}_j^T\vec{v}_i)^2 = o(1) ~ \forall j \in\TN^c.
   \end{align}
   Note that this is error term is surely far from optimal, but is sufficient here.
   Overall we can now say \begin{align}
      \left|\left[N\E(\vec{e}_{j}^T\vec{v}_{i})^2 -\E (\vec{e}_{j}^T\vec{g}_{i})^2\right]\right| \leq 1 + o(1) \leq 2 ~ \forall j \in\TN^c.
   \end{align}

   We can apply $\hQUE$ to the terms in square parentheses to give $\epsilon_{1}, \ldots, \epsilon_{r}>0$ such that 
   \begin{align}
       |N\E(\vec{e}_{j_k}^T\vec{v}_{i_k})^2 - \E (\vec{e}_{j_k}^T\vec{g}_{i_k})^2| \leq N^{-\epsilon_{i_k}} ~~~ \forall j_k\in\TN ~\forall i_k=1,\ldots, r.
   \end{align}
   We can obtain a single error bound by setting $\epsilon = \min_i \epsilon_i$, where clearly $\epsilon > 0$ and then write
   \begin{align}\label{eq:eta_defn}
    N\E(\vec{e}_{j_k}^T\vec{v}_{i_k})^2 - \E (\vec{e}_{j_k}^T\vec{g}_{i_k})^2 = \eta_{i_kj_k}^2 N^{-\epsilon}    
   \end{align}
   where $ \eta_{i_kj_k}^2 \in [-1, 1]$. To further include the indices $j\in\TN^c$, we extended the expression (\ref{eq:eta_defn}) to all $j_k$ by saying
   \begin{align}
       \eta_{i_kj_k}^2 &\in [-1, 1] ~ \forall j_k\in\TN,\\
       \eta_{i_kj_k}^2 &\in [-N^{\epsilon}, N^{\epsilon}] ~ \forall j_k\in\TN^c.
   \end{align}
Overall we have \begin{align}
     \E \left(\sum_{i=1}^r y_i \vec{v}_i^T\Lambda \vec{v}_i\right)^l-\mathbb{E}\left(\sum_{i=1}^r y_i \frac{1}{N}\vec{g}_i^T\Lambda \vec{g}_i\right)^l = N^{-l(1+\epsilon)}\sum_{i_1,\ldots, i_l=1}^r\sum_{j_1,\ldots, j_l=1}^N\prod_{k=1}^{l} y_{i_k} \lambda_{j_k} \eta_{i_kj_k}^2
\end{align}
but by comparing with (\ref{eq:que_error_bound1}) we can rewrite as
\begin{align}
    \E \left(\sum_{i=1}^r y_i \vec{v}_i^T\Lambda \vec{v}_i\right)^l-\mathbb{E}\left(\sum_{i=1}^r y_i \frac{1}{N}\vec{g}_i^T\Lambda \vec{g}_i\right)^l = \left( \sum_{i=1}^r N^{-(1+\epsilon)}y_i \vec{\eta}_i^T\Lambda \vec{\eta}_i\right)^l
\end{align}
where $\vec{\eta}_i^T = (\eta_{i1},\ldots, \eta_{iN})$.

\end{proof}

\subsection{Main result}
\begin{thm}\label{thm:nearly_free_addn}
    Let $X$ be an $N\times N$ real symmetric random matrix and let $D$ be an $N\times N$ symmetric matrix (deterministic or random). Let $\hat{\mu}_X, \hat{\mu}_{D}$ be the empirical spectral measures of the sequence of matrices $X, D$ and assume there exist deterministic limit measures $\mu_X, \mu_D$. Assume that $X$ has QUE, i.e. \ref{thm:que_gaussian}.
    Assume also the $\hat{\mu}_X$ concentrates in the sense that 
    \begin{align}\label{eq:concentration_condition}
        \P(W_1(\hat{\mu}_X, \mu_X) > \delta) \lesssim e^{-N^{\tau} f(\delta)}
    \end{align}
    where $\tau>0$ and $f$ is some positive increasing function.
    Then $H = X + D$ has a limiting spectral measure and it is given by the free convolution $\mu_X \boxplus \mu_D$.
\end{thm}
\begin{rem}
A condition like (\ref{eq:concentration_condition}) is required so that the Laplace method can be applied to the empirical measure $\hat{\mu}_X$. There are of course other ways to formulate such a condition.
Consider for example the conditions used in Theorems 1.2 and 4.1 of \cite{arous2021exponential}.
There it is assumed the existence of a sequence of deterministic measures $(\mu_N)_{N \geq 1}$ and a constant $\kappa>0$ such that for large enough $N$
\begin{align}\label{eq:concentration_condition_deterministic}
    W_1(\E \hat{\mu_X}, \mu_N) \leq N^{-\kappa}, ~~ W_1(\mu_N, \mu_X) \leq N^{-\kappa},
\end{align}
which is of course just a deterministic version of (\ref{eq:concentration_condition}).
\cite{arous2021exponential} introduce the extra condition around concentration of Lipschitz traces:
\begin{align}\label{eq:lipschitz_traces}
    \P\left( \left|\frac{1}{N} \Tr f(H_N) - \frac{1}{N}\E \Tr f(H_N)\right| > \delta \right) \leq \exp\left(-\frac{c_{\zeta}}{N^{\zeta}} \min\left\{\left(\frac{N\delta}{\|f\|_{Lip}}\right)^2, \left(\frac{N\delta}{\|f\|_{Lip}}\right)^{1+\epsilon_0} \right\}\right),
\end{align}
for all $\delta>0$, Lipschitz $f$ and $N$ large enough, where $\zeta, c_{\zeta}>0$ are some constants.
As shown in the proof of Theorem 1.2, this condition is sufficient to obtain \begin{align}\label{eq:concetration_integral}
    \P \left( \left|\int |\lambda| d\hat{\mu}_X(\lambda) - \int |\lambda| d\E\hat{\mu}_X(\lambda)\right| \leq t \right) \leq \exp\left(-\frac{c_{\zeta}}{N^{\zeta}} \min\left\{ (2Nt\eta)^2, (2Nt\eta)^{1+\epsilon_0}\right\}\right)
\end{align}
for any $t>0$ and for large enough $N$.
Note that \cite{arous2021exponential} prove this instead for integration against a regularised version of $\log|\lambda\|$, but the proof relies only the integrand's being Lipschitz, so it goes through just the same here.
(\ref{eq:concetration_integral}) and (\ref{eq:concentration_condition_deterministic}) clearly combine to give (\ref{eq:concentration_condition}).
The reader may ignore this remark if they are content to take (\ref{eq:concentration_condition}) as an assumption.
Alternatively, as we have shown, (\ref{eq:concentration_condition}) can be replaced by (\ref{eq:concentration_condition_deterministic}) and (\ref{eq:lipschitz_traces}), conditions which have already been used for quite general results in the random matrix theory literature.
\end{rem}
\begin{proof}
We shall denote use the notation \begin{align}
    G_H(z) = \frac{1}{N}\Tr (z - H)^{-1}.
\end{align}
Recall the supersymmetric approach to calculating the expected trace of the resolvent of a random matrix ensemble:
\begin{align}
    \E_H G_H(z) = \frac{1}{N}\frac{\partial}{\partial \j}\Bigg|_{\j=0} \mathbb{E}_H Z_H(\j)
\end{align}
where \begin{align}
    Z_H(\j) &= \frac{\det(z + \j - H)}{\det(z - H)} = \int d\Psi e^{-i\Tr AH} e^{i\Tr \Psi\Psi^{\dagger}J},\\
    A &= \phi\phi^{\dagger} + \chi\chi^{\dagger},\\
    J &= I_N \otimes \left(\begin{array}{cc} z & 0 \\ 0 & \j + z\end{array}\right),\\
    d\Psi &= \frac{d\phi d\phi^* d\chi d\chi^*}{-(2\pi)^N i},\\
    \Psi &= \left(\begin{array}{c} \phi \\ \chi \end{array}\right)
\end{align}
with $\phi\in\C^N$ and $\chi,\chi^*$ being $N$-long vectors of anti-commuting variables.
Independence of $X$ and $D$ gives
\begin{align}
    \E_H Z_H(\j) &= \int d\Psi e^{i\Tr \Psi\Psi^{\dagger}J} \E_{X,D} e^{-i\Tr A(X + D)} \notag\\
    &= \int d\Psi e^{i\Tr \Psi\Psi^{\dagger}J} \E_D e^{-i\Tr AD} \E_X e^{-i\Tr AX}.
\end{align}
$\E_D$ simply means integration against a delta-function density if $D$ is deterministic. 

\medskip
Let us introduce some notation: for $N\times N$ matrices $K$, $\Phi_X(K) = \E_X e^{-i\Tr XK}$, and similarly $\Phi_D$. We also define a new matrix ensemble $\bar{X} \overset{d}{=} O^T \Lambda O$, where $\Lambda=\text{diag}(\lambda_1, \ldots, \lambda_N)$ are equal in distribution to the eigenvalues of $X$ and $O$ is an entirely independent Haar-distributed orthogonal matrix.

Now
\begin{align}
   &\E_H Z_H(\j) = \int d\Psi e^{i\Tr \Psi\Psi^{\dagger}J} \Phi_{\bX}(K)\Phi_D(K) + \int d\Psi e^{i\Tr \Psi\Psi^{\dagger}J} (\Phi_X(K) - \Phi_{\bX}(K))\Phi_D(K)\notag\\
   \implies &\E G_{D+X}(z) = \E G_{D+\bX}(z) + \frac{1}{N}\frac{\partial}{\partial \j}\Bigg|_{\j = 0} \int d\Psi e^{i\Tr \Psi\Psi^{\dagger}J} (\Phi_X(K) - \Phi_{\bX}(K))\Phi_D(K) \equiv   \E G_{D+\bX}(z) + E(z)
\end{align}
and so we need to analyse the error term $E(z)$.\\

Now consider $X = U^T\Lambda U$ where the rows of $U$ are the eigenvectors $\{\vec{u}_i\}_i$ of $X$. Say also that $K = Q^TYQ$ for diagonal $Y = (y_1, \ldots, y_r, 0, \ldots, 0)$, where we note that $K$ has fixed rank, by construction. Then \begin{align*}
    \Tr XK = Y (UQ^T)^T\Lambda (UQ^T)
\end{align*}
but Lemma \ref{lem:que_preserved_under_rotation} establishes that the rows of $UQ^T$ obey QUE, since the rows of $U$ do. Further, Lemma \ref{lem:weaker_col_que} then establishes that the columns of $UQ^T$ obey $\hQUE$ as required by Lemma \ref{lem:explicit_error}. Let $\{\vec{v}_i\}$ be those columns, then we have \begin{align}\label{eq:trace_as_vsum}
    \Tr XK = \sum_{i=1}^r y_i \vec{v}_i^T \Lambda \vec{v}_i.
\end{align}

The expectation over $X$ can be split into eigenvalues and conditional eigenvectors 
\begin{align}\label{eq:phix_expansion}
    \Phi_X(K) = \E_{\Lambda}\E_{U\mid \Lambda}\sum_{l=0}^{\infty} \frac{1}{l!} (-i)^l\left(\Tr U^T \Lambda UK\right)^l.
\end{align}
We can simply bound \begin{align}
   \left| \sum_{l=0}^{n} \frac{1}{l!} (-i)^l\left(\Tr U^T \Lambda U\right)^l\right| \leq e^{|\Tr U^T\Lambda UK}
\end{align}
for any $n$, but clearly \begin{align}
    \E_{U\mid \Lambda}e^{|\Tr U^T\Lambda UK|} < \infty
\end{align}
since, whatever the distribution of $U\mid\Lambda$, the integral is over a compact group (the orthogonal group $O(N)$) and the integrand has no singularities. 
Therefore, by the dominated convergence theorem
\begin{align}
     \Phi_X(K) = \E_{\Lambda}\sum_{l=0}^{\infty} \frac{1}{l!} (-i)^l\E_{U\mid \Lambda}\left(\Tr U^T \Lambda UK\right)^l
\end{align}
and in precisely the same way
\begin{align}\label{eq:phibx_expansion}
      \Phi_X(K) = \E_{\Lambda}\sum_{l=0}^{\infty} \frac{1}{l!} (-i)^l\E_{O\sim \mu_{Haar}}\left(\Tr O^T \Lambda OK\right)^l.
\end{align}
Recalling (\ref{eq:trace_as_vsum}) we now have
\begin{align}\label{eq:phix_conditional_expansion}
    \Phi_X(K) = \E_{\Lambda}\sum_{l=0}^{\infty} \frac{1}{l!} (-i)^l \E_{U\mid \Lambda}\left(\sum_{i=1}^r y_i \vec{v}_i^T \Lambda \vec{v}_i\right)^l.
\end{align}
and similarly 
\begin{align}\label{eq:phixbar_conditional_expansion}
    \Phi_{\bX}(K) = \E_{\Lambda}\sum_{l=0}^{\infty} \frac{1}{l!} (-i)^l \E_{U\mid \Lambda}\left(\sum_{i=1}^r y_i \bar{\vec{v}}_i^T \Lambda \bar{\vec{v}}_i\right)^l.
\end{align}
where the $\bar{\vec{v}}_i$ are defined in the obvious way from $\bX$.
We would now apply $\hQUE$, but to do so we must insist that $\E_{\Lambda}$ is taken over the \emph{ordered} eigenvalues of $X$.
Having fixed that convention, lemma \ref{lem:explicit_error} can be applied to the terms 
\begin{align}
    \E_{U\mid \Lambda}\left(\sum_{i=1}^r y_i \vec{v}_i^T \Lambda \vec{v}_i\right)^l
\end{align}
in (\ref{eq:phix_conditional_expansion}).
The terms in $\Phi_{\bX}$ can be treated similarly.
This results in \begin{align}
    \Phi_X(K) - \Phi_{\bX}(K) &= \E_{\Lambda}\Bigg[\sum_{l=0}^{\infty} \frac{i^l}{l!} \left\{\E_{\{\vec{g}_i\}_{i=1}^r}\left(\sum_{i=1}^r y_i \frac{1}{N}\vec{g}_i^T\Lambda \vec{g}_i\right)^l + \left(\sum_{i=1}^r N^{-(1+\epsilon)} y_i \vec{\eta}_i^T\Lambda \vec{\eta}_i\right)^l \right\}\notag\\
    & ~~~~~~~~-\sum_{l=0}^{\infty} \frac{i^l}{l!} \left\{\E_{\{\vec{g}_i\}_{i=1}^r}\left(\sum_{i=1}^r y_i \frac{1}{N}\vec{g}_i^T\Lambda \vec{g}_i\right)^l + \left(\sum_{i=1}^r N^{-(1+\epsilon)} y_i \bar{\vec{\eta}}_i^T\Lambda \bar{\vec{\eta}}_i\right)^l \right\}\Bigg]
\end{align}
The exponential has infinite radius of convergence, so we may re-order the terms in the sums to give cancellation
\begin{align*}
    \Phi_X(K) - \Phi_{\bX}(K) = \E_{\Lambda}\sum_{l=1}^{\infty}\frac{1}{l!} N^{-(1+\epsilon)l}(-i)^l\left(\sum_{i=1}^r y_i \vec{\eta}_i^T\Lambda \vec{\eta}_i\right)^l -\E_{\Lambda}\sum_{l=1}^{\infty}\frac{1}{l!} N^{-(1+\epsilon)l}(-i)^l\left(\sum_{i=1}^r y_i \bar{\vec{\eta}}_i^T\Lambda \bar{\vec{\eta}}_i\right)^l.
\end{align*}
Here $\epsilon>0$ and $\vec{\eta}_i, \tilde{\vec{\eta}}_i\in \C^N$ with 
\begin{align}
   -1\leq [(\vec{\eta}_i)_j]^2 , [(\bar{\vec{\eta}}_i)_j]^2 \leq 1 ~&~ \forall i=1,\ldots, r, ~  \forall j\in\TN,\\
   -N^{\epsilon}\leq [(\vec{\eta}_i)_j]^2 , [(\bar{\vec{\eta}}_i)_j]^2 \leq N^{\epsilon} ~&~ \forall i=1,\ldots, r, ~  \forall j\in\TN.
\end{align}

Simplifying, we obtain \begin{align}
    \Phi_X(K) - \Phi_{\bX}(K) = \E_{\Lambda}\exp\left(-iN^{-(1+\epsilon)}\sum_{i=1}^r y_i \vec{\eta}_i^T\Lambda \vec{\eta}_i\right) - \E_{\Lambda}\exp\left(-iN^{-(1+\epsilon)}\sum_{i=1}^r y_i \tilde{\vec{\eta}}_i^T\Lambda \tilde{\vec{\eta}}_i\right).
\end{align}
Since $|\TN^c| \leq 2N^{1-\delta}$ we have \begin{align}
    \sum_{j\in\TN^c} |\lambda_j| \leq \mathcal{O}(N^{1-d} N^{-1}) \Tr |\Lambda|
\end{align}
and so \begin{align}
    |\vec{\eta}_i^T\Lambda \vec{\eta}_i| \leq \Tr|\Lambda|\left( 1 + \mathcal{O}(N^{\epsilon - \delta})\right).
\end{align}
For any fixed $\delta>0$, $\epsilon$ can be reduced if necessary so that $\epsilon < \delta$ and then for sufficiently large $N$ we obtain, say, 
\begin{align}
    |\vec{\eta}_i^T\Lambda \vec{\eta}_i| \leq 2 \Tr|\Lambda|.
\end{align}
Thence we can write $\vec{\eta}_i^T\Lambda\vec{\eta}_i = \Tr|\Lambda| \xi_i$ for $\xi_i\in[-2, 2]$, and similarly $\tilde{\vec{\eta}}_i^T\Lambda\tilde{\vec{\eta}}_i = \Tr|\Lambda| \tilde{\xi}_i$. Now
\begin{align*}
    \E_{\Lambda}\exp\left(-iN^{-(1+\epsilon)} \sum_{i=1}^r \xi_iy_i\int \Tr|\Lambda|\right)=\E_{\Lambda}\exp\left(-iN^{-\epsilon} \sum_{i=1}^r \xi_iy_i\int d\hat{\mu}_X(\lambda)|\lambda|\right)
\end{align*}
so we can apply Laplace's method to the empirical spectral measure $\hat{\mu}_X$ to obtain \begin{align}
    \E_{\Lambda}\exp\left(-iN^{-(1+\epsilon)} \sum_{i=1}^r \xi_iy_i\int \Tr|\Lambda|\right)=\exp\left(-iN^{-\epsilon} (q+o(1))\sum_{i=1}^r \xi_iy_i\right) + o(1)
\end{align}
where the $o(1)$ terms do not depend on the $y_i$ and where we have defined \begin{align}
    q = \int d\mu_X(\lambda)|\lambda|.
\end{align}
Further, we can write $\sum_{i=1}^r \xi_i y_i = \zeta\Tr K$, where $\zeta \in [\min_i\{\xi_i\}, \max_i\{\xi_i\}]\subset [-1, 1]$, and similarly $\sum_{i=1}^r \tilde{\xi}_i y_i = \tilde{\zeta}\Tr K$. Then \begin{align}
    \Phi_X(K) - \Phi_{\bX}(K) = e^{-iN^{-\epsilon}\zeta (q + o(1)) \Tr K} - e^{-iN^{-\epsilon}\tilde{\zeta} (q + o(1)) \Tr K} + o(1)
\end{align}
but
\begin{align}
    &\frac{1}{N}\frac{\partial}{\partial \j}\Bigg|_{\j=0}\int d\Psi e^{i\Tr \Psi\Psi^{\dagger}J}e^{-iN^{-\epsilon}\zeta (q + o(1)) \Tr K}\Phi_D(K) = \E G_{D + N^{-\epsilon}\zeta(q + o(1))I}(z) = \E G_D(z + \mathcal{O}(N^{-\epsilon}))\notag\\
    \implies & E(z) = \E G_D(z + \mathcal{O}(N^{-\epsilon})) + o(1) - \E G_D(z + \mathcal{O}(N^{-\epsilon})) - o(1) = o(1).
\end{align}
We have thus established that \begin{align}
    \E G_{D+X}(z) = \E G_{D+\bX}(z) + o(1)
\end{align}
from which one deduces that $\mu_{D+X} = \mu_{D + \bX} = \mu_D \boxplus \mu_{\bX} = \mu_D\boxplus\mu_X$.
\end{proof}

\begin{rem}
We have also constructed a non-rigorous argument for Theorem \ref{thm:nearly_free_addn} where the supersymmetric approach is replaced by the replica method. This approach simplifies some of the analysis but at the expense of being less convincing.
\end{rem}

\subsection{Experimental validation}
We consider the following matrix ensembles:
\begin{align*}
    M\sim GOE^n ~ &: ~ \text{Var}(M_{ij}) = \frac{1 + \delta_{ij}}{2n},\\
    M\sim UWig^n &: ~ \sqrtsign{n}M_{ij} \overset{i.i.d}{\sim} U(0, \sqrtsign{6}) ~~\text{up to symmetry},\\
    M\sim \Gamma Wig^n &: ~ 2\sqrtsign{n}M_{ij} \overset{i.i.d}{\sim} \Gamma(2) ~~\text{up to symmetry},\\
    M\sim UWish^n &: ~ M\overset{d}{=}\frac{1}{m}XX^T, ~ X_{ij}\overset{i.i.d}{\sim}U(0, \sqrtsign{12}) ~~\text{for }X\text{ of size }n\times m, ~~ \frac{n}{m} \overset{n,m\rightarrow\infty}{\rightarrow} \alpha,\\
    M\sim Wish^n &: M\overset{d}{=}\frac{1}{m}XX^T, ~X_{ij}\overset{i.i.d}{\sim}\mathcal{N}(0,1)~~\text{for }X\text{ of size }n\times m, ~~ \frac{n}{m} \overset{n,m\rightarrow\infty}{\rightarrow} \alpha.
\end{align*}

All of the $GOE^n, UWig^n, \Gamma Wig^n$ have the same limiting spectral measure, namely $\mu_{SC}$, the semi-circle of radius $\sqrtsign{2}$. $UWish^n, Wish^n$ have a Marcenko-Pastur limiting spectral measure $\mu_{MP}$, and the constant $\sqrtsign{12}$ is chosen so that the parameters of the MP measure match those of a Gaussian Wishart matrix $Wish^n$. $GOE^n, Wish^n$ are the only ensembles whose eigenvectors are Haar distributed, but all ensembles obey a local law in the sense above. It is known that the sum of $GOE^n$ and any of the other ensembles will have limiting spectral measure given by the free additive convolution of $\mu_{SC}$ and the other ensemble's measure (so either $\mu_{SC}\boxplus\mu_{MP}$ or $\mu_{SC}\boxplus\mu_{SC}$). Our result implies that the same holds for addition of the non-invariant ensembles. Sampling from the above ensembles is simple, so we can easily generate spectral histograms from multiple independent matrix samples for large $n$. $\mu_{SC}\boxplus\mu_{SC}$ is just another semi-circle measure but with radius $2$. $\mu_{SC}\boxplus\mu_{MP}$ can be computed in the usual manner with $R$-transforms and is given by the solution to the polynomial \begin{align*}
    \frac{\alpha}{2}t^3 - \left(\frac{1}{2} + \alpha z\right)t^2 + (z + \alpha - 1)t - 1 = 0.
\end{align*}
i.e. Say the cubic has roots $\{r_1, r_2 + is_2, r_2 - is_2\}$ for $s_2\geq 0$, then the density of $\mu_{SC}\boxplus\mu_{MP}$ at $z$ is $s_2/\pi$. This can all be solved numerically. The resulting plots are in Figure \ref{fig:free_sum_comparison} and clearly show agreement between the free convolutions and sampled spectral histograms.

\begin{figure}[h]
    \centering
    \begin{subfigure}{0.3\linewidth}
        \includegraphics[width=\textwidth]{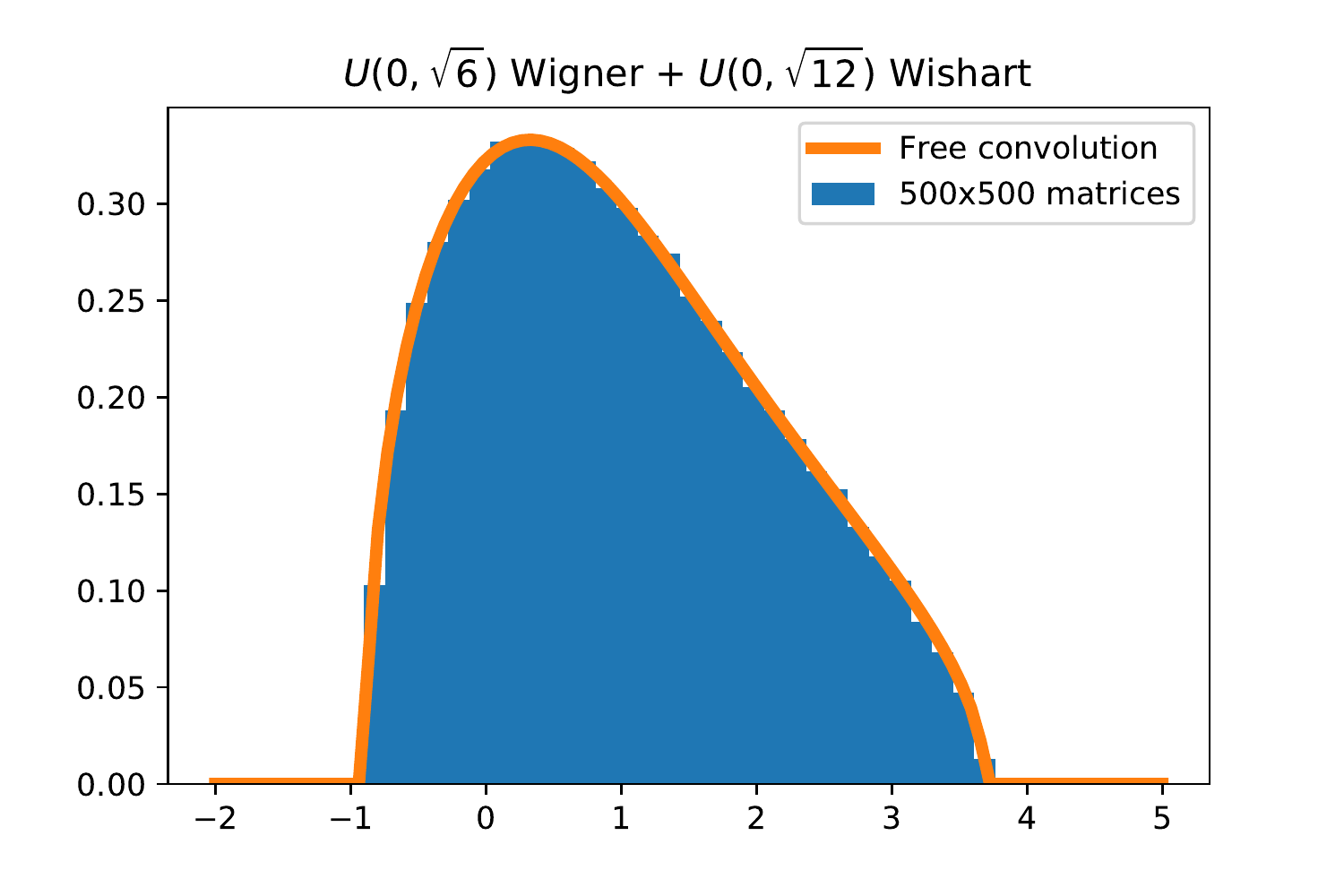}
        \subcaption{ $UWig^n + UWish^n$}
    \end{subfigure}
    \begin{subfigure}{0.3\linewidth}
        \includegraphics[width=\textwidth]{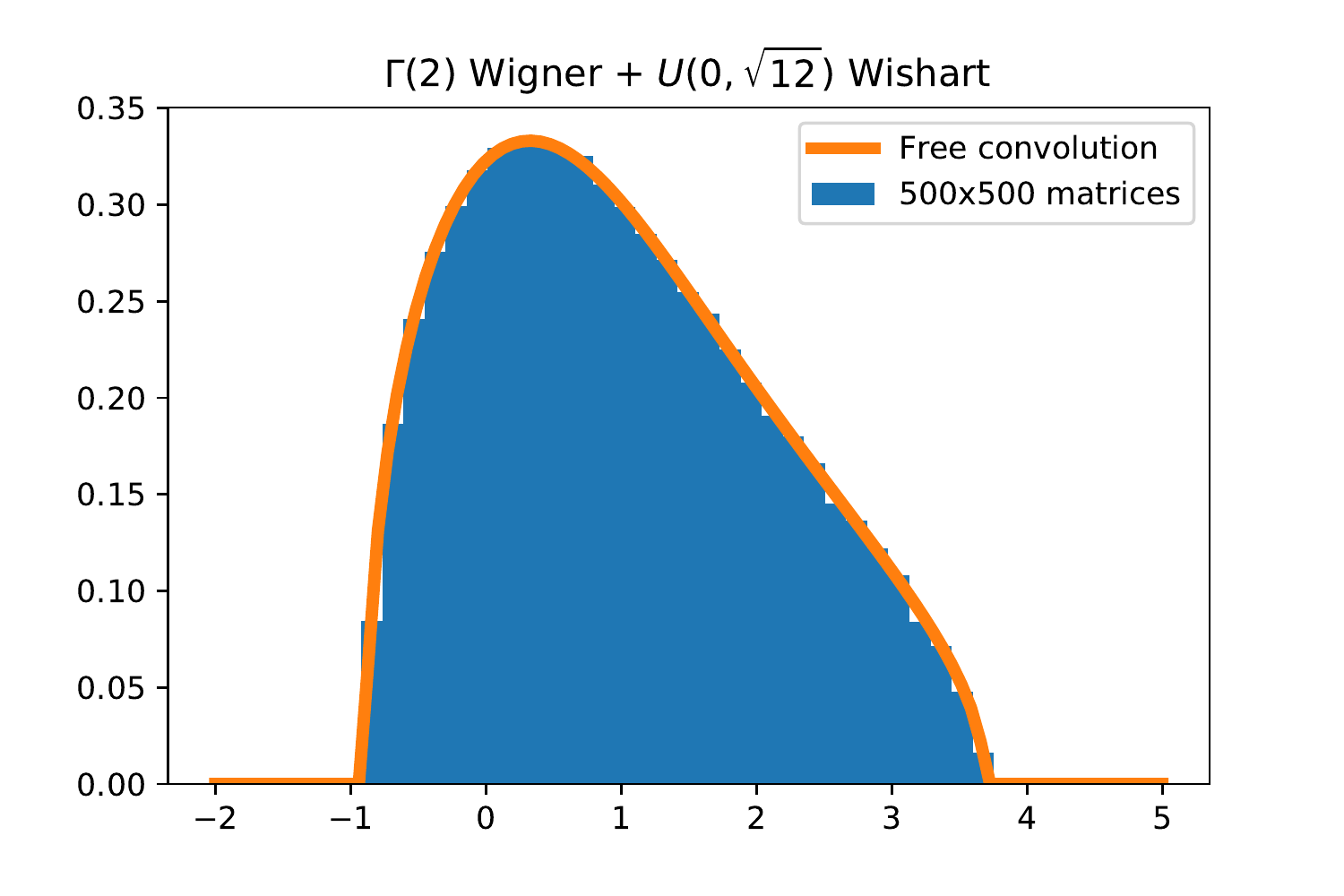}
        \subcaption{ $\Gamma Wig^n + UWish^n$}
    \end{subfigure}
    \begin{subfigure}{0.3\linewidth}
        \includegraphics[width=\textwidth]{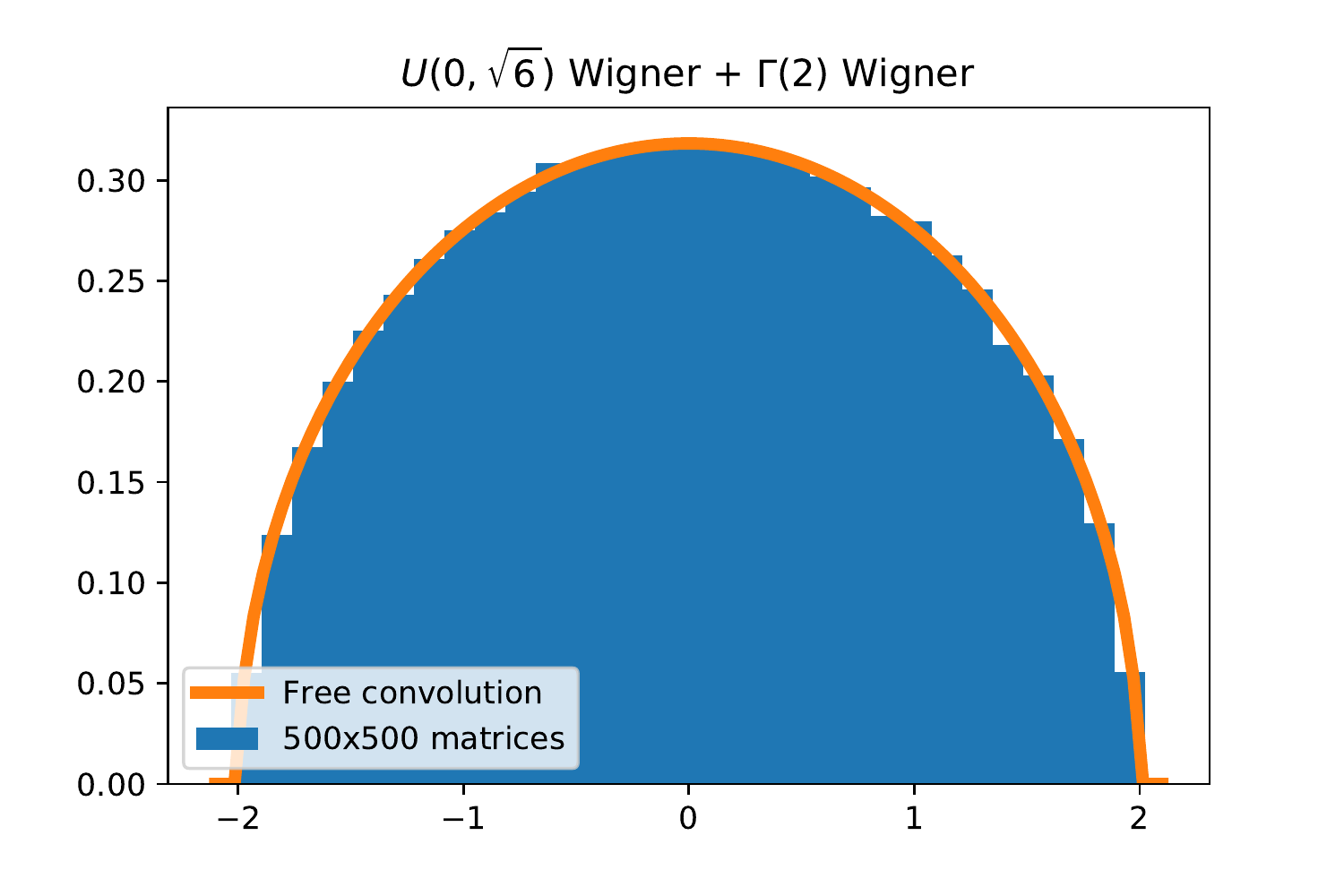}
        \subcaption{ $UWig^n + \Gamma Wig^n$}
    \end{subfigure}
    \caption{Comparison of theoretical spectral density and empirical from sampled matrices all of size $500\times 500$. We combine $50$ independent matrix samples per plot.}
    \label{fig:free_sum_comparison}
\end{figure}

We can also test the result in another more complicated case. Consider the case of random $d$-regular graphs on $N$ vertices. Say $M\sim Reg^{N,d}$ is the distribution of the adjacency matrix of such random graphs. The limiting spectral density of $M\sim Reg^{N,d}$ is known in closed form, as is its Stieljtes transform \cite{bauerschmidt2019local} and \cite{bauerschmidt2019local} established a local law of the kind required for our results. Moreover, there are known efficient algorithms for sampling random $d$-regular graphs \cite{kim2003generating, steger1999generating} along with implementations \cite{SciPyProceedings_11}. 
Let $\mu_{KM}^{(d)}$ be the Kesten-McKay law, the limiting spectral measure of $d$-regular graphs. We could find an explicit degree-6 polynomial for the Stieljtes transform of $\mu_{KM}^{(d)}\boxplus \mu_{SC}$ and compare to spectral histograms as above. Alternatively we can investigate agreement with $\mu_{KM}^{(d)}\boxplus \mu_{SC}$ indirectly by sampling and comparing spectra from say $Reg^{N,d} + UWig^N$ and also from  $Reg^{N,d} + GOE^N$.
The latter case will certainly yield the distribution $\mu_{KM}^{(d)}\boxplus \mu_{SC}$ since the GOE matrices are freely independent from the adjacency matrices.
Figure shows a q-q plot for samples from these two matrix distributions and demonstrates near-perfect agreement, thus showing that indeed the spectrum of $Reg^{N,d} + UWig^N$  is indeed described by $\mu_{KM}^{(d)}\boxplus \mu_{SC}$.
We reached the same conclusion when repeating the above experiment with $UWish^N + Reg^{N,d}$ and $Wish^n + Reg^{N,d}$.
\begin{figure}
    \centering
    \includegraphics[width=0.4\textwidth]{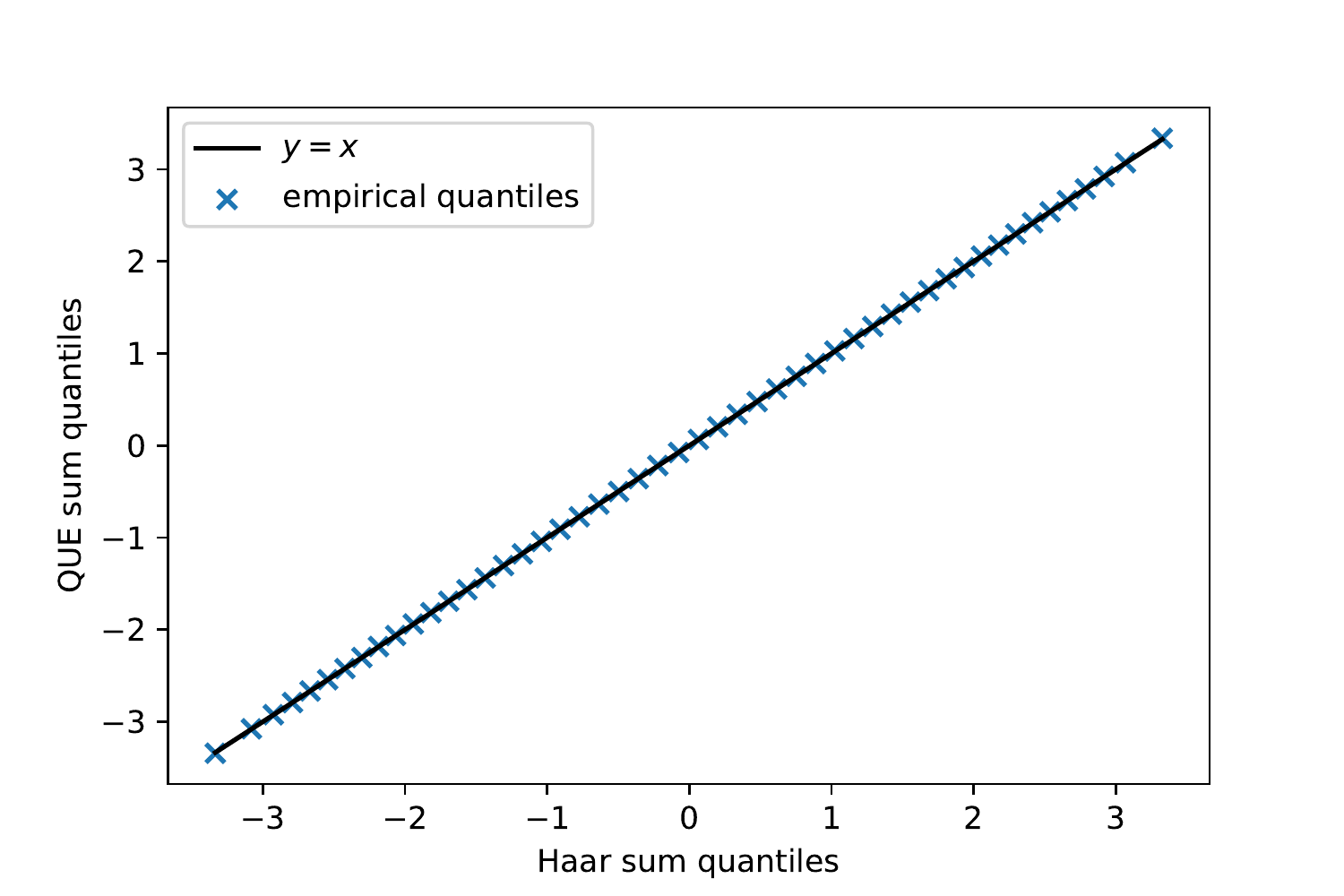}
    \caption{q-q plot comparing the spectrum of samples from $Reg^{N,d} + UWig^N$ ($y$-axis) to samples from $Reg^{N,d} + GOE^N$ ($x$-axis).}
    \label{fig:qq_plot}
\end{figure}
\section{Universal complexity of loss surfaces}\label{sec:determinants}

\subsection{Extension of a key result and prevalence of minima}
Let's recall Theorem 4.5 from \cite{arous2021exponential}. $H_N(u)$ is our random matrix ensemble with some parametrisation $u\in\R^m$ and its limiting spectral measure is $\mu_{\infty}(u)$.

Define 
\begin{align}
    \mathcal{G}_{-\epsilon} = \{u\in\R^m \mid \mu_{\infty}(u) \left( (-\infty, 0) \right) \leq \epsilon\}.
\end{align}
So $\mathcal{G}_{-\epsilon}$ is the event that $\mu_{\infty}(u)$ is close to being supported only on $(0, \infty)$. Let $l(u), r(u)$ be te left and right edges respectively of the support of $\mu_{\infty}(u)$.

\begin{thm}[\cite{arous2021exponential} Theorem 4.5]
    Fix some $\mathcal{D}\subset \R^m$ and suppose that $\mathcal{D}$ and the matrices $H_N(u)$ satisfy the following.
    \begin{itemize}
        \item For every $R>0$ and every $\epsilon>0$, we have 
        \begin{align}\label{eq:bl_condition}
            \lim_{N\rightarrow\infty} \frac{1}{N\log N}\log\left[\sup_{u\in B_R}\P\left(d_{BL}(\hat{\mu}_{H_N(u)}, \mu_{\infty}(u) \right) > \epsilon\right] = -\infty.
        \end{align}
    \item Several other assumptions detailed in \cite{arous2021exponential}.
    \end{itemize}
    
    Then for any $\alpha>0$ and any fixed $p\in\N$, we have
    \begin{align}
        \lim_{N\rightarrow\infty} \frac{1}{N}\log\int_{\mathcal{D}} e^{-(N+p)\alpha u^2}\E\left[|\det(H_N(u))|\1\{i(H_N(u)) = 0\}\right]du = \sup_{u\in\mathcal{D}\cap\mathcal{G}}\left\{\int_{\mathbb{R}} \log|\lambda| d\mu_{\infty}(u) (\lambda) - \alpha u^2\right\}.
    \end{align}
\end{thm}

We claim the following extension 
\begin{cor}\label{cor:general_k}
    Under the same assumptions as the above theorem and for any integer sequence $k(N) > 0$ such that $k/N \rightarrow 0$ as $N\rightarrow\infty$, we have
    \begin{align}
        \lim_{N\rightarrow\infty} \frac{1}{N}\log\int_{\mathcal{D}} e^{-(N+p)\alpha u^2}\E\left[|\det(H_N(u))|\1\{i(H_N(u)) \leq k\}\right]du = \sup_{u\in\mathcal{D}\cap\mathcal{G}}\left\{\int_{\mathbb{R}} \log|\lambda| d\mu_{\infty}(u) (\lambda) - \alpha u^2\right\}.
    \end{align}
\end{cor}
\begin{proof}
    Firstly note that 
    \begin{align}
       &\frac{1}{N}\log\int_{\mathcal{D}} e^{-(N+p)\alpha u^2}\E\left[|\det(H_N(u))|\1\{i(H_N(u)) \leq k\}\right] du \notag\\
       \geq &\frac{1}{N}\log\int_{\mathcal{D}} e^{-(N+p)\alpha u^2}\E\left[|\det(H_N(u))|\1\{i(H_N(u)) = 0\}\right] du,
    \end{align}
    so it suffices to establish a complementary upper bound. The proof in of Theorem 4.5 in \cite{arous2021exponential} establishes an upper bound using 
    \begin{align}
        \lim_{N\rightarrow\infty} \frac{1}{N} \log \int_{\left(\mathcal{G}_{-\epsilon}\right)^c} e^{-N\alpha u^2}  \E\left[|\det(H_N(u)|\1\{i(H_N(u)) = 0\}\right] du = -\infty
    \end{align}
    which holds for all $\epsilon > 0$, so our proof is complete if we can prove the analogous result
    \begin{align}\label{eq:low_prob_analogue}
        \lim_{N\rightarrow\infty} \frac{1}{N} \log \int_{\left(\mathcal{G}_{-\epsilon}\right)^c} e^{-N\alpha u^2}  \E\left[|\det(H_N(u)|\1\{i(H_N(u)) \leq k\}\right] du = -\infty.
    \end{align}
    
    As in \cite{arous2021exponential}, let $f_{\epsilon}$ be some $\frac{1}{2}$-Lipschitz function satisfying $\frac{\epsilon}{2}\1_{x \leq -\epsilon} \leq f_{\epsilon}(x)\leq \frac{\epsilon}{2} \1_{x\leq 0}$. Suppose $u\in \left(\mathcal{G}_{-\epsilon}\right)^c$ and also $i(H_N(u)) \leq k$. Then we have \begin{align}
        0 \leq \int d\hat{\mu}_{H_N(u)}(x) ~ f_{\epsilon}(x) \leq \frac{k\epsilon}{2N}
    \end{align}
    and also
    \begin{align}
        \frac{\epsilon^2}{2} \leq \int d\mu_{\infty}(u)(x) ~ f_{\epsilon}(x) \leq \frac{\epsilon}{2}.
    \end{align}
    We have 
    \begin{align}
        d_{BL}(\hat{\mu}_{H_N(u)}, \mu_{\infty}(u) ) &\geq \left|\int d\hat{\mu}_{H_N(u)}(x) ~ f_{\epsilon}(x)  - 
    \int d\mu_{\infty}(u)(x) ~ f_{\epsilon}(x) \right|\notag\\
    &\geq \left|\left|\int d\hat{\mu}_{H_N(u)}(x) ~ f_{\epsilon}(x)\right|  - 
    \left|\int d\mu_{\infty}(u)(x) ~ f_{\epsilon}(x) \right|\right|,
    \end{align}
    so if we can choose 
    \begin{align}\label{eq:epsilon_eta}
        \frac{k\epsilon}{2N} \leq \frac{\epsilon^2}{2} - \eta
    \end{align}
    for some $\eta > 0$, then we obtain $d_{BL}(\hat{\mu}_{H_N(u)}, \mu_{\infty}(u) ) \geq \eta$. Then applying (\ref{eq:bl_condition}) yields the result (\ref{eq:low_prob_analogue}). (\ref{eq:epsilon_eta}) can be satisfied if
    \begin{align}\label{eq:epsilon_eta_explicit}
        \epsilon \geq \frac{k}{2N} + \frac{1}{2}\sqrtsign{\frac{k^2}{N^2} + 8\eta}.
    \end{align}
    So, given $\epsilon>0$, we can take $N$ large enough such that, say, $\frac{k(N)}{N} < \frac{\epsilon}{4}$. By taking $\eta < \frac{\epsilon^2}{128}$ we obtain
    \begin{align}
        \frac{k}{2N} + \frac{1}{2}\sqrtsign{\frac{k^2}{N^2} + 8\eta} < \frac{\epsilon}{4} + \frac{1}{\sqrtsign{2}}\max\left(\sqrtsign{8\eta}, \frac{\epsilon}{4}\right) < \frac{1 + 2^{-1/2}}{4}\epsilon
    \end{align}
    and so (\ref{eq:epsilon_eta_explicit}) is satisfied. Now finally (\ref{eq:bl_condition}) can be applied (with $\eta$ in place of $\epsilon$) and so we conclude (\ref{eq:low_prob_analogue}).
    
    \medskip
    Overall we see that the superexponential BL condition (\ref{eq:bl_condition}) is actually strong enough to deal with any $o(N)$ index not just index-0. This matches the GOE (or generally invariant ensemble) case, in which the terms with $\1\{i(H_N(u))=k\}$ are suppressed compared to the exact minima terms $\1\{i(H_N(u))=0\}$.
\end{proof}
\begin{rem}
Note that \ref{cor:general_k} establishes that, on the exponential scale, the number of critical points of any index $k(N) = o(N)$ is no more than the number of exact local minima.
\end{rem}

\subsection{The dichotomy of rough and smooth regions}
Recall the batch loss from Section \ref{subsec:hess_model}:
\begin{align}
   \frac{1}{b} \sum_{i=1}^{b} \mathcal{L}(f_{\vec{w}}(\vec{x}_i), y_i), ~~ (\vec{x}_i, y_i)\overset{\text{i.i.d.}}{\sim} \mathbb{P}_{data}.
\end{align}
As with the Hessian in Section \ref{subsec:hess_model}, we use the model $L\equiv L_{\text{batch}}(\vec{w}) = L_{\text{true}}(\vec{w}) + \sbf V(\vec{w})$, where $V$ is a random function $\mathbb{R}^N\rightarrow\mathbb{R}$.



Now let us define the complexity for sets $\B\subset \R^N$ \begin{align}
    C_N(\B) = |\{ \vec{w}\in \B \mid \nabla L(\vec{w}) = 0\}|.
\end{align}
This is simply the number of stationary points of the training loss in the region $\B$ of weight space. A Kac-Rice formula applied to $\nabla L$ gives 
\begin{align}\label{eq:w_b_dom_integral_kacrice}
    \E C_N = \int_{\B} d\vec{w} ~ \phi_{\vec{w}}(-\sbf^{-1}\nabla L_{\text{true}}) \E |\det (A + \sbf X)|
\end{align}
where $\phi_{\vec{w}}$ is the density of $\nabla V$ at $\vec{w}$.
A rigorous justification of this integral formula would, for example, have to satisfy the conditions of the results of \cite{adler2007random}.
This is likely to be extremely difficult in any generality, though is much simplified in the case of Gaussian $V$ (and $X$) - see \cite{adler2007random} Theorem 12.1.1 or \cite{baskerville2021loss} Theorem 4.4.
Hereafter, we shall take (\ref{eq:w_b_dom_integral_kacrice}) as assumed.
The next step is to make use of strong self-averaging of the random matrix determinants. 
Again, we are unable to establish this rigorously at present, but note that this property has been proved in some generality by \cite{arous2021exponential}, although we are unable to satisfy all the conditions of those results in any generality here.
Self-averaging and using the addition results above gives
\begin{align*}
    \frac{1}{N}\log\E |\det (A + \sbf X)| = \int d(\mu_b \boxplus \nu)(\lambda) \log |\lambda| + o(1)
\end{align*}
where $\mu_b, \nu$ depend in principle on $\vec{w}$.
We are concerned with $N^{-1}\log \E C_N$, and in particular its sign, which determines the complexity of the loss surface in $\B$: positive $\leftrightarrow$ exponentially many (in $N$) critical points, negative $\leftrightarrow$ exponentially few (i.e. none).
The natural next step is to apply the Laplace method with large parameter $N$ to determine the leading order term in $\E C_N$, however the integral is clearly not of the right form. Extra assumptions on $\phi_{\vec{w}}$ and $\nabla L_{\text{true}}$ could be introduced, e.g. that they can be expressed as functions of only a finite number of combinations of coordinates of $\vec{w}$.

\medskip
Suppose that $\phi_{\vec{w}}$ has its mode at $0$, for any $\vec{w}$, which is arguably a natural property, reflecting in a sense that the gradient noise has no preferred direction in $\R^N$.
The sharp spike at the origin in the spectral density of deep neural network Hessians suggests that generically
\begin{align}\label{eq:log_int_neg}
    \int d(\mu_b \boxplus \nu)(\lambda) \log |\lambda| < 0.
\end{align}
We claim it is reasonable to expect the gradient (and Hessian) variance to be increasing in $\|\vec{w}\|_2$.
Indeed, consider the general form of a multi-layer perceptron neural network:
\begin{align}
    f_{\vec{w}}(\vec{x}) = \sigma(\vec{b}^{(L)} + W^{(L)}\sigma(\vec{b}^{(L-1)} + W^{(L-1)}\ldots \sigma(\vec{b}^{(1)} + W^{(1)}\vec{x} )\ldots ))
\end{align}
where all of the weight matrices $W^{(l)}$ and bias vectors $\vec{b}^{(l)}$ combine to give the weight vector $\vec{w}$.
Viewing $\vec{x}$ as a random variable, making $f$ a random function of $\vec{w}$, we expect from the above that the variance in $f_{\vec{w}}$ is generally increasing in $\|\vec{w}\|_2$, and so therefore similarly with $L_{\text{batch}}$.

\medskip
Overall it follows that $\phi_{\vec{w}}(-\sbf^{-1}\nabla L_{\text{true}})$ is generally decreasing in $\|\nabla L_{\text{true}}\|$, but the maximum value at $\phi_{\vec{w}}(0)$ is decreasing in $\|\vec{w}\|_2$.
The picture is therefore that the loss surface is simple and without critical points in regions for which $\nabla L_{\text{true}}$ is far from $0$.
In neighbourhoods of $\nabla L_{\text{true}} = 0$, the loss surface may become complex, with exponentially many critical points, however if $\|\vec{w}\|_2$ is too large then the loss surface may still be without critical points.
In addition, the effect of larger batch size (and hence larger $\sbf^{-1}$) is to simplify the surface.
These considerations indicate that deep neural network loss surfaces are simplified by over-parametrisation, leading to the spike in the Hessian spectrum and thus (\ref{eq:log_int_neg}).
The simple fact that neural networks' construction leads gradient noise variance to increase with $\|\vec{w}\|_2$ has the effect of simplifying the loss landscape far from the origin of weight space, and even precluding the existence of any critical points of the batch loss.

\section{Implications for curvature from local laws}\label{sec:precond}
Consider a general stochastic gradient update rule with curvature-adjusted preconditioning:
\begin{align}\label{eq:precond_udpate}
    \vec{w}_{t+1} = \vec{w}_t - \alpha B_t^{-1} \nabla L(\vec{w}_t)
\end{align}
where recall that $L(\vec{w})$ is the batch loss, viewed as a random function on weight space.
$B_t$ is some preconditioning matrix which in practice would be chosen to somehow approximate the curvature of $L$.
Such methods are discussed at length in \cite{martens2016second} and also describe some of the most successful optimisation algorithms used in practice, such as Adam \cite{kingma2014adam}.
The most natural choice for $B_t$ is $B_t = \nabla^2 L(\vec{w}_t)$, namely the Hessian of the loss surface.
In practice, it is standard to include a damping parameter $\delta>0$ in $B_t$, avoid divergences when inverting.
Moreover, typically $B_t$ will be constructed to be some positive semi-definite approximation to the curvature such as the generalised Gauss Newton matrix \cite{martens2016second}, or the diagonal gradient variance form used in Adam \cite{kingma2014adam}.
Let us now suppose that $B_t = B_t(\delta) = \hat{H}_t + \delta$, where $\hat{H}_t$ is some chosen positive semi-definite curvature approximation and $\delta>0$.
We can now identify $B_t(\delta)^{-1}$ as in fact the Green's function of $\hat{H}_t$, i.e. \begin{align}
    B_t(\delta)^{-1} = -(-\delta - \hat{H}_t)^{-1} = -G_t(-\delta).
\end{align}
But $G_t$ is precisely the object used in the statement of a local law on for $\hat{H}_t$.
Note that $\nabla L(\vec{w}_t)$ is a random vector and however $\hat{H}_t$ is constructed, it will generally be a random matrix and dependent on $\nabla L(\vec{w}_t)$ in some manner that is far too complicated to handle analytically.
As we have discussed at length hitherto, we conjecture that a local law is reasonable assumption to make on random matrices arising in deep neural networks.
In particular \cite{baskerville2022appearance} demonstrated universal local random matrix theory statistics not just for Hessians of deep networks but also for Generalised Gauss-Newton matrices.
Our aim here is to demonstrate how a local law on $\hat{H}_t$ dramatically simplifies the statistics of (\ref{eq:precond_udpate}).
Note that some recent work \cite{wei2022more} has also made use of random matrix local laws to simplify the calculation of test loss for neural networks.

\medskip
A local law on $\hat{H}_t$ takes the precise form (for any $\xi, D>0$
\begin{align}\label{eq:precond_local_law}
    \sup_{\|\vec{u}\|,\|\vec{v}\|  = 1, z\in\vec{S}}\P\left( |\vec{u}^TG(z)\vec{v} - \vec{u}^T\Pi(z)\vec{v}| > N^{\xi}\left(\frac{1}{N\eta} + \sqrtsign{\frac{\Im g_{\mu}(z)}{N\eta}}\right)\right) \leq N^{-D}
\end{align}
where \begin{align}
    \vec{S} = \left\{E + i\eta \in \C \mid |E| \leq \omega^{-1}, ~ N^{-1 + \omega} \leq \eta \leq \omega^{-1}\right\}
\end{align}
$\mu$ is the limiting spectral measure of $\hat{H}_t$ and, crucially, $\Pi$ is a \emph{deterministic} matrix.
We will use the following standard notation to re-express (\ref{eq:precond_local_law})
\begin{align}\label{eq:local_law_inside}
    |\vec{u}^TG(z)\vec{v} - \vec{u}^T\Pi(z)\vec{v}| \prec \Psi_N(z), ~~~ \|\vec{u}\|,\|\vec{v}\|  = 1, z\in\vec{S},
\end{align}
where $\Psi_N(z) = \frac{1}{N\eta} + \sqrtsign{\frac{\Im g_{\mu}(z)}{N\eta}}$ and the probabilistic statement, valid for all $\xi, D>0$ is implicit in the symbol $\prec$.
In fact, we will need the local law outside the spectral support, i.e. at $z = x + i\eta$ where $x\in\mathbb{R}\backslash\text{supp}(\mu)$.
In that case $\Psi_N(z)$ is replaced by $\frac{1}{N(\eta + \kappa)}$ where $\kappa$ is the distance of $x$ from $\text{supp}(\mu)$ on the real axis, i.e. 
\begin{align}\label{eq:local_law_outside}
    |\vec{u}^TG(z)\vec{v} - \vec{u}^T\Pi(z)\vec{v}| \prec \frac{1}{N(\eta + \kappa)}, ~~~ \|\vec{u}\|,\|\vec{v}\|  = 1, ~ x\in\mathbb{R}\backslash\text{supp}(\mu).
\end{align}
For $\delta>0$ this becomes\begin{align}
      |\vec{u}^TG(-\delta)\vec{v} - \vec{u}^T\Pi(-\delta)\vec{v}| \prec \frac{1}{N\delta} \|\vec{u}\|_2 \|\vec{v}\|_2
\end{align}
for $\delta>0$ and now any $\vec{u}, \vec{v}$.
Applying this to (\ref{eq:precond_udpate}) gives \begin{align}
     |\vec{u}^TB_t^{-1}\nabla L(\vec{w}_t) - \vec{u}^T\Pi_t(-\delta)\nabla L(\vec{w}_t) | \prec \frac{1}{N\delta} \|\vec{u}\|_2 \|\nabla L(\vec{w}_t)\|_2.
\end{align}
Consider any $\vec{u}$ with $\|\vec{u}\|_2 = \alpha$, then we obtain  \begin{align}
     |\vec{u}^TB_t^{-1}\nabla L(\vec{w}_t) - \vec{u}^T\Pi_t(-\delta)\nabla L(\vec{w}_t) | \prec \frac{\alpha  \|\nabla L(\vec{w}_t)\|_2}{N\delta}.
\end{align}
Thus with high probability, for large $N$, we can replace (\ref{eq:precond_udpate}) by \begin{align}\label{eq:precond_update_det_equiv}
    \vec{w}_{t+1} = \vec{w}_t - \alpha \Pi_t(-\delta) \nabla L(\vec{w}_t)
\end{align}
incurring only a small error, provided that \begin{align}
    \delta >> \frac{\|\nabla L(\vec{w}_t)\|_2}{N} \alpha.
\end{align}
Note that the only random variable in (\ref{eq:precond_update_det_equiv}) is $\nabla L (\vec{w}_t)$.
If we now consider the case $\nabla L (\vec{w}_t) = \nabla \bar{L}(\vec{w}_t) + \vec{g}(\vec{w}_t)$ for deterministic $\bar{L}$, then \begin{align}\label{eq:precond_update_final}
    \vec{w}_{t+1} = \vec{w}_t - \alpha \Pi_t(-\delta) \nabla \bar{L}(\vec{w}_t) - \alpha \Pi_t(-\delta)\vec{g}(\vec{w}_t)
\end{align}
and so the noise in the parameter update is entirely determined by the gradient noise.
Moreover note the \emph{linear} dependence on $\vec{g}$ in (\ref{eq:precond_update_final}).
For example, a Gaussian model for $\vec{g}$ immediately yields a Gaussian form in (\ref{eq:precond_update_final}), and e.g. if $\E \vec{g} = 0$, then \begin{align}
    \E(\vec{w}_{t+1} - \vec{w}_t) = -\alpha \Pi_t(-\delta) \E \nabla L(\vec{w}_t).
\end{align}

\medskip
A common choice in practice for $\hat{H}$ is a diagonal matrix, e.g. the diagonal positive definite curvature approximation employed by Adam \cite{kingma2014adam}.
In such cases, $\hat{H}$ is best viewed as an approximation to the eigenvalues of some positive definite curvature approximation.
The next result establishes that a local law assumption on a general curvature approximation matrix can be expected to transfer to an analogous result on a diagonal matrix of ts eigenvalues.
\begin{prop}
    Suppose that $\hat{H}$ obeys a local law of the form (\ref{eq:local_law_outside}).
    Define the diagonal matrix $D$ such that $D_i \overset{d}{=} \lambda_i$ where $\{\lambda_i\}_i$ are the sorted eigenvalues of $\hat{H}$.
    Let $G_D(z) = (z - D)^{-1}$ be the resolvent of $D$.
    Let $\mathfrak{q}_j[\mu]$ be the $j$-th quantile of $\mu$, the limiting spectral density of $\hat{H}$, i.e.\begin{align}
        \int_{-\infty}^{\mathfrak{q}_j[\mu]} d\mu(\lambda) = \frac{j}{N}.
    \end{align}
    Then $D$ obeys the local law \begin{align}
        |(G_D)_{ij} - \delta_{ij}(z - \mathfrak{q}_j[\mu])^{-1}| \prec \frac{1}{N^{2/3} (\kappa + \eta)^2}, ~~ z = x + i\eta, ~x\in\mathbb{R}\backslash\text{supp}(\mu),
    \end{align}
    where $\kappa$ is the distance of $x$ from $\text{supp}(\mu)$.
    Naturally, we can redefine $D_i = \lambda_{\sigma{i}}$ for any permutation $\sigma\in S_N$ and the analogous statement replacing $\mathfrak{q}_j[\mu]$ with $\mathfrak{q}_{\sigma(j)}$ will hold.
\end{prop}
\begin{proof}
    As in \cite{erdos2017dynamical}, the local law (\ref{eq:local_law_inside}), (\ref{eq:local_law_outside}) is sufficient to obtain rigidity of the eigenvalues in the bulk, i.e.
    for any $\epsilon, D > 0$\begin{align}\label{eq:rigid}
        \P\left(\exists j ~\mid~ |\lambda_j - \mathfrak{q}_j[\mu]| \geq N^{\epsilon}\left[\min(j, N-j+1)\right]^{-1/3}N^{-2/3}\right) \leq N^{-D}.
    \end{align}
    Then we have \begin{align}
       \left|\frac{1}{z - \lambda_j} - \frac{1}{z - \mathfrak{q}_j[\mu]}\right|=\left| \frac{\lambda_j - \mathfrak{q}_j[\mu]}{(z - \lambda_j)(z - \mathfrak{q}_j[\mu])}\right|.
    \end{align}
    For $z=x+i\eta$ and $x$ at a distance $\kappa>0$ from $\text{supp}(\mu)$ \begin{align}
        |z-\mathfrak{q}_j[\mu]|^2 \geq \eta^2 + \kappa^2 \geq \frac{1}{2}(\eta + \kappa)^2,
    \end{align}
    and the same can be said for $|z - \mathfrak{q}_j[\mu]|^2$ with high probability, by applying the rigidity (\ref{eq:rigid}). A second application of rigidity to $|\lambda_j - \mathfrak{q}_j[\mu]|$ gives\begin{align}
         \left|\frac{1}{z - \lambda_j} - \frac{1}{z - \mathfrak{q}_j[\mu]}\right|  \prec \frac{1}{N^{2/3}\min(j, N-j+1)^{1/3} (\kappa + \eta)^2}
    \end{align}
    which yields the result.
\end{proof}

With this result in hand, we get the generic update rule akin to (\ref{eq:precond_update_final}), with high probability \begin{align}
    \vec{w}_{t+1} = \vec{w}_t -\alpha ~\text{diag}\left(\frac{1}{\pi_j+ \delta}\right) \nabla \bar{L}(\vec{w}_t) - \alpha~\text{diag}\left(\frac{1}{\pi_j + \delta}\right) \vec{g}(\vec{w}_t)
\end{align}
where we emphasise again that the $\pi_j$ are \emph{deterministic} and the only stochastic term is the gradient noise $\vec{g}(\vec{w}_t)$.

\medskip 
\paragraph{Implications for preconditioned stochastic gradient descent.} The key insight from this section is that generic random matrix theory effects present in preconditioning matrices of large neural networks can be expected to drastically simplify the optimisation dynamics due to high-probability concentration of the pre-conditioning matrices around deterministic equivalents, nullifying the statistical interaction between the pre-conditioning matrices and gradient noise.
Moreover, with this interpretation, the damping constant typically added to curvature estimate matrices is more than a simple numerical convenience: it is essential to yield the aforementioned concentration results.

\medskip
As an example of the kind of analysis that the above makes possible, consider the results of \cite{ia}.
The authors consider a Gaussian process model for the noise in the loss surface, resulting in tractable analysis for convergence of stochastic gradient descent in the presence of statistical dependence between gradient noise in different iterations.
Such a model implies a specific form of the loss surface Hessian and its statistical dependence on the gradient noise.
This situation is a generalisation of the spin glass model exploited in various works \cite{choromanska2015loss,baskerville2021loss,baskerville2022spin}, except that in those cases the Hessian can be shown to be independent of the gradients.
Absent the very special conditions that lead to independence, one expects the analysis to be intractable, hence why the authors in \cite{ia} restrict to SGD without preconditioning, or simply assume a high probability concentration on a deterministic equivalent.
To make this discussion more concrete, consider a model $L = L_{\text{true}} + V$ where $V$ is a Gaussian process with mean $0$ and covariance function
\begin{align}
    K(\vec{x}, \vec{x}') = k\left(\frac{1}{2}\|\vec{x} - \vec{x}'\|_2^2\right) q\left( \frac{1}{2}(\|\vec{x}\|_2^2 + \|\vec{x}'\|_2^2) \right),
\end{align}
where $k$ is some decreasing function and $q$ some increasing function.
The discussion at the end of the previous section suggests that the covariance function for loss noise should not be modelled as stationary, hence the inclusion of the $q$ term.
For convenience define $\Delta = \frac{1}{2}(\|\vec{x} - \vec{x}'\|_2^2)$ and $S = \frac{1}{2}(\|\vec{x}\|_2^2 + \|\vec{x}'\|_2^2)$.
Then it is a short exercise in differentiation to obtain \begin{align}
    \text{Cov}\left(\partial_i V(\vec{w}), \partial_j V(\vec{w})\right) &= \text{Cov}\left(\partial_i V(\vec{w}), \partial_j V(\vec{w}')\right)\Bigg|_{\vec{w}=\vec{w}'} \notag\\
    &= \frac{\partial^2}{\partial w_i\partial w_j'}K(\vec{w}, \vec{w}')\Bigg|_{\vec{w}=\vec{w}'} \notag\\
    &= -k'(0)q(\|\vec{w}\|_2)\delta_{ij} + k(0)q''(\|\vec{w}\|_2^2) w_iw_j.
\end{align}
and moreover \begin{align}
    \text{Cov}\left(\partial_{il} V(\vec{w}), \partial_j V(\vec{w})\right) &= \text{Cov}\left(\partial_{il} V(\vec{w}), \partial_j V(\vec{w}')\right)\Bigg|_{\vec{w}=\vec{w}'} \notag\\
    &= \frac{\partial^3}{\partial w_i\partial w_l\partial w_j'}K(\vec{w}, \vec{w}')\Bigg|_{\vec{w}=\vec{w}'} \notag\\
    &= -k'(0)q'(\|\vec{w}\|_2^2)w_l\delta_{ij} + q'''(\|\vec{w}\|_2^2)k(0) w_iw_lw_j' - k'(0)q'(\|\vec{x}\|_2)w_i\delta_{jl}.
\end{align}
Hence we see that the gradients of $L$ and its Hessian are statistically dependent by virtue of the non-stationary structure of $V$.
Putting aside issues of positive definite pre-conditioning matrices, and taking $\delta$ such that $(\nabla^2 L + \delta)^{-1}$ exists (almost surely) for large $N$, it is clear that the distribution of $(\nabla^2 L + \delta)^{-1}\partial V$ will be complicated and non-Gaussian.
This example concretely illustrates our point: even in almost the simplest case, where the gradient noise is Gaussian, the pre-conditioned gradients are generically considerably more complicated and non-Gaussian.
Moreover, centred Gaussian noise on gradient is transformed into generically non-centred noise by pre-conditioning.
Continuing the differentiation above, it is elementary to obtain the covariance structure of the Hessian $\nabla^2 V$, though the expressions are not instructive.
Crucially, however, the Hessian is Gaussian and the covariance of any of its entries is $\mathcal{O}(1)$ (in large $N$), so the conditions in Example 2.12 of \cite{erdHos2019random} apply to yield an optimal local law on the Hessian, which in turn yields the above high-probability concentration of $(\nabla^2 L + \delta)^{-1}$ provided that $\delta$ is large enough.


\section{Conclusion}\label{sec:concl}
In this paper we have considered several aspects of so-called universal random matrix theory behaviour in deep neural networks.
Motivated by prior experimental results, we have introduced a model for the Hessians of DNNs that is more general than any previously considered and, we argue, actually flexible enough to capture the Hessians observed in real-world DNNs.
Our model is built using random matrix theory assumptions that are more general than those previously considered and may be expected to hold in quite some generality.
By proving a new result for the addition of random matrices, using a novel combination of quantum unique ergodicity and the supersymmetric method, we have derived expressions for the spectral outliers of our model.
Using Lanczos approximation to the outliers of large, practical DNNs, we have compared our expressions for spectral outliers to data and demonstrated strong agreement for some DNNs.
As well as corroborating our model, this analysis presents indirect evidence of the presence of universal local random matrix statistics in DNNs, extending earlier experimental results.
Our analysis also highlights a possibly interesting distinction between some DNN architectures, as Resnet architectures appear to better agree with our theory than other architectures and Resnets have been previously observed to have better-behaved loss surfaces than many other architectures.

\medskip
We also presented quite general arguments regarding the number of local optima of DNN loss surfaces and how `rough' or `smooth' such surfaces are.
Our arguments build on a rich history of complexity calculations in the statistical physics and mathematics literature but, rather than performing detailed calculations in some specific, highly simplified toy model, we instead present general insights based on minimal assumptions.
Finally we highlight an important area where random matrix local laws, an essential aspect of universality, may very directly influence the performance of certain popular optimisation algorithms for DNNs.
Indeed, we explain how numerical damping, combined with random matrix local laws, can act to drastically simplify the training dynamics of large DNNs.

\medskip
Overall it is our hope that this paper demonstrates the relevance of random matrix theory to deep neural networks beyond highly simplified toy models.
Moreover, we have shown how quite general and universal properties of random matrices can be fruitfully employed to derive practical, observable properties of DNN spectra.
This work leaves several challenges for future research.
All of our work relies on either local laws for e.g. DNN Hessians, or on matrix determinant self-averaging results.
Despite the considerable progress towards establishing local laws for random matrices over the last decade or-so, it appears that establishing any such laws for, say, the Hessians of any DNNs is quite out of reach.
We expect that the first progress in this direction will come from considering DNNs with random i.i.d. weights and perhaps simple activation functions.
Based on the success of recent works on random DNNs \cite{pastur2020randomiid}, we conjecture that the Gram matrices of random DNN Jacobians may be the simplest place to establish a local law, adding to the nascent strand of \emph{nonlinear} random matrix theory \cite{pennington2017nonlinear,benigni2019eigenvalue,pastur2020randomiid}.
We also believe that there is more to be gained in further studies of forms of random matrix universality in DNNs.
For example, our ideas may lead to tractable analysis of popular optimisation algorithms such as Adam \cite{kingma2014adam} as the problem is essentially reduced to deriving a local law for the gradient pre-conditioning matrix and dealing with the gradient noise.

\section*{Acknowledgements}
The authors wish to thank Patrick Lopatto for noticing some technical errors and providing valuable insights into proof strategies and existing results for QUE.

\appendix

\section{Invariant equivalent ensembles}\label{sec:inv}

\cite{unterberger2019global} gives the following integro-differential equation relating the equilibrium measure $\mu$ to the potential $V$ of invariant ensembles: \begin{align}\label{eq:v_stieljtes}
    \frac{\beta}{2}\dashint\frac{1}{x-y}d\mu(y) = V'(x).
\end{align} 
So in the case of real symmetric matrices we have \begin{align}\label{eq:muinf_V_relation}
    \frac{1}{2} \bar{g_{\mu}}(x) = V'(x)
\end{align}
where $g_{\mu}$ is the Stieljtes transform of $\mu$ and the bar over $\bar{g_{\mu}}$ indicates that the principal value has been taken.


Given a sufficiently nice $\mu$ (\ref{eq:v_stieljtes}) defines $V$ up-to a constant of integration on $\supp(\mu)$, but $V$ is not determined on $\R\backslash\supp(\mu)$.

\begin{lem}\label{lemma:V_unique_mu}
    For compactly supported probability measure $\mu$ on $\R$ and real potential $V$, define \begin{align}
        S_V[\mu](y) = V(y) - \int d\mu(x) \log|y-x|.
    \end{align}Suppose $S_V[\mu](y)=c$, a constant, for all $y\in\supp(\mu)$ and $S_V[\mu](y) \geq c$ for all $y\in\mathbb{R}$. Then $\mu$ is a minimiser amongst all probability measures on $\R$ of the energy \begin{align}
        \mathcal{E}_V[\mu] = \int d\mu(x) V(x) - \iint_{x< y} d\mu(x)d\mu(y)\log|x-y|.
    \end{align}
\end{lem}
 \begin{proof}
    Consider a probability measure that is close to $\mu$ in the sense of $W_1$ distance, say.
    For any such measure, one can find an arbitrarily close probability measure $\mu'$ of the form
    \begin{align}
        \mu' = \mu + \sum_{i=1}^r a_i\1_{[y_i - \delta_i, y_i + \delta_i]} - \sum_{i=1}^s b_i\1_{[z_i - \eta_i, z_i + \eta_i]}
    \end{align}
    where all $a_i, b_i>0$ and $\delta_i, \eta_i, a_i, b_i \leq \epsilon$ for some small $\epsilon>0$.
    To ensure that $\mu'$ is again a probability measure we must impose $\sum_ia_i = \sum_jb_j$.
    The strategy now is to expand $ \mathcal{E}_V[\mu'] $ about $\mu$ to first order in $\epsilon$, but first note the symmetrisation
    \begin{align}
         \iint_{x< y} d\mu(x)d\mu(y)\log|x-y| = \frac{1}{2} \iint_{x\neq y} d\mu(x)d\mu(y)\log|x-y|.
    \end{align}
    Then \begin{align}
        \mathcal{E}_V[\mu'] -  \mathcal{E}_V[\mu] &= \sum_{i=1}^r a_i V(y_i) - \sum_{i=1}^s b_i V(z_i) - \sum_{i=1}^ra_i \int d\mu(x)\log|x-y_i| + \sum_{i=1}^rb_i \int d\mu(x)\log|x-z_i| + \mathcal{O}(\epsilon^2) \notag\\
        &= \sum_{i=1}^r a_i S_V[\mu](y_i) - \sum_{i=1}^r b_i S_V[\mu](z_i) + \mathcal{O}(\epsilon^2).
    \end{align}
    Observe that if all $y_i, z_i\in \text{supp}(\mu)$ then $S_V[\mu](y_i) = S_V[\mu](y_i) = c $ and so $ \mathcal{E}_V[\mu'] = \mathcal{E}_V[\mu]$.
    Without loss of generality therefore, we take $y_i\notin \text{supp}(\mu)$ and $z_i\in \text{supp}(\mu)$, whence 
    \begin{align}
         \mathcal{E}_V[\mu'] -  \mathcal{E}_V[\mu] \geq c\sum_{i=1}^r a_i - c\sum_{i=1}^s b_i = 0.
    \end{align}
    
 \end{proof}
 
\begin{lem}
    Consider a probability measure $\mu$ on $\R$ with compact support, absolutely continuous with respect to the Lebesgue measure. Then there exists a potential $V:\R\rightarrow\R$ which yields a well-defined invariant distribution on real symmetric matrices for which the equilibrium measure is $\mu$.
\end{lem}
\begin{proof}
    (\ref{eq:muinf_V_relation}) can be integrated to obtain $V$ and the condition $S_V[\mu]=c$ (a constant) on $\supp(\mu)$ determines $V$ uniquely on $\supp(\mu)$. Next observe that, for $y\in\R\backslash\supp(\mu)$ there exists some constant $R>0$ such that $|x-y| \leq R + |y|$ and so $\log|x-y| \leq |y| + R$. Therefore \begin{equation}
        S_V[\mu](y) \geq V(y) - |y| - R.
    \end{equation}
    So the condition $S_V[\mu](y) \geq c$ for $y\in\R\backslash\supp(\muinf)$ can be satisfied by taking $V(y) = (y-y_0)^2 + b$ for some constants $y_0, b$, which we need not record. There is sufficient freedom in $b, y_0$ left to make $V$ continuous on $\R$. Since $V$ is constructed with Gaussian tails, it can certainly be used to define a legitimate invariant ensemble of real symmetric random matrices. By Lemma \ref{lemma:V_unique_mu}, the $V$ we have constructed has equilibrium measure $\mu$.
\end{proof}

\section{Experimental details}\label{sec:exp_details}
This section gives full details of the experimental set-up and analysis for the outlier experiments.

\subsection{Architectures and training of models.}
We use the GPU powered Lanczos quadrature algorithm \citep{gardner2018gpytorch, meurant2006lanczos}, with the Pearlmutter trick \citep{pearlmutter1994fast} for Hessian vector products, using the PyTorch \citep{paszke2017automatic} implementation of both Stochastic Lanczos Quadrature and the Pearlmutter. We then train a 16 Layer VGG CNN \citep{simonyan2014very} with $P=15291300$ parameters 
and the 28 Layer Wide Residual Network \citep{zagoruyko2016wide,he2016deep} architectures 
on the CIFAR-$100$ dataset (45,000 training samples and 5,000 validation samples) using SGD. We use the following  learning rate schedule:
\begin{equation}
	\label{eq:schedule}
	\alpha_t = 
	\begin{cases}
		\alpha_0, & \text{if}\ \frac{t}{T} \leq 0.5 \\
		\alpha_0[1 - \frac{(1 - r)(\frac{t}{T} - 0.5)}{0.4}] & \text{if } 0.5 < \frac{t}{T} \leq 0.9 \\
		\alpha_0r, & \text{otherwise.}
	\end{cases}
\end{equation}
We use a learning rate ratio $r=0.01$ and a total number of epochs budgeted $T=300$. We further use momentum set to $\rho=0.9$, a weight decay coefficient of $0.0005$ and data-augmentation on PyTorch \citep{paszke2017automatic}. 

\subsection{Implementation of constraints}\label{sec:impl_constraints}
As mentioned in the main text, one of the three weights of the linear model fit in the outlier analysis, $\beta$, is constrained to be positive, as it corresponds to a second cumulant, i.e. a variance, of a probability measure.
Recall that the linear model's parameters are solved exactly as functions of the unknown $\theta^{(i)}$, and these parameters are in turn optimised using gradient descent.
$\beta$ is unconstrained during the linear solve, but its value is determined by the $\theta^{(i)}$, so to impose the constraint $\beta>0$ we add to the mean squared error loss the term \begin{align}
    \beta = 1000 \max(0, -\beta)
\end{align}
which penalises negative $\beta$ values and is minimised at any non-negative value.
The factor $1000$ was roughly tuned by hand to give consistently positive values for $\beta$.

\medskip
There is also the constraint that $\theta^{(i)} > \theta^{(i+1)}>0$ for all $i$. This is imposed simply using a re-parametrisation.
We introduce unconstrained raw value $t^{(i)}$ taking values in $\R$ and define 
\begin{align*}
    \theta^{(i)} = \sum_{j=1}^{i} \log ( 1 + \exp(t^{(i)}) ),
\end{align*}
then the gradient descent optimisation is simply performed over the $t^{(i)}$.

\subsection{Fitting of outlier model}\label{sec:impl_fitting}
We optimise the mean squared error with respect to the raw parameters $t^{(i)}$ using 200 iterations of Adam \cite{kingma2014adam} with a learning rate of 0.2.
The learning rate was chosen heuristically by increasing in steps until training became unstable.
The number of iterations was chosen heuristically as being comfortably sufficient to obtain convergence of Adam.
The raw parameters $t^{(i)}$ were initialised by drawing independently from a standard Gaussian.
The $t^{(i)}$ were initialised and trained using the above method 20 times and the values with the lowest mean squared error were chosen.

\addcontentsline{toc}{section}{References}
\bibliographystyle{alpha}
\bibliography{references}

\newcommand{\etalchar}[1]{$^{#1}$}
\begin{thebibliography}{GPAM{\etalchar{+}}14}

\bibitem[AA{\v{C}}13]{auffinger2013random}
Antonio Auffinger, G{\'e}rard~Ben Arous, and Ji{\v{r}}{\'\i} {\v{C}}ern{\`y}.
\newblock Random matrices and complexity of spin glasses.
\newblock {\em Communications on Pure and Applied Mathematics}, 66(2):165--201,
  2013.

\bibitem[ABM21]{arous2021exponential}
G{\'e}rard~Ben Arous, Paul Bourgade, and Benjamin McKenna.
\newblock Exponential growth of random determinants beyond invariance.
\newblock {\em arXiv preprint arXiv:2105.05000}, 2021.

\bibitem[AEK{\etalchar{+}}14]{alex2014isotropic}
Bloemendal Alex, L{\'a}szl{\'o} Erd{\H{o}}s, Antti Knowles, Horng-Tzer Yau, and
  Jun Yin.
\newblock Isotropic local laws for sample covariance and generalized wigner
  matrices.
\newblock {\em Electronic Journal of Probability}, 19:1--53, 2014.

\bibitem[AGZ10]{anderson2010introduction}
Greg~W Anderson, Alice Guionnet, and Ofer Zeitouni.
\newblock {\em An introduction to random matrices}.
\newblock Number 118 in Cambridge stidies in advanced mathematics. Cambridge
  university press, 2010.

\bibitem[AT{\etalchar{+}}07]{adler2007random}
Robert~J Adler, Jonathan~E Taylor, et~al.
\newblock {\em Random fields and geometry}, volume~80.
\newblock Springer, 2007.

\bibitem[Ben20]{benigni2020eigenvectors}
Lucas Benigni.
\newblock Eigenvectors distribution and quantum unique ergodicity for deformed
  wigner matrices.
\newblock In {\em Annales de l'Institut Henri Poincar{\'e}, Probabilit{\'e}s et
  Statistiques}, volume~56, pages 2822--2867. Institut Henri Poincar{\'e},
  2020.

\bibitem[BES20]{bao2020support}
Zhigang Bao, L{\'a}szl{\'o} Erd{\H{o}}s, and Kevin Schnelli.
\newblock On the support of the free additive convolution.
\newblock {\em Journal d'Analyse Math{\'e}matique}, 142(1):323--348, 2020.

\bibitem[BGK22]{baskerville2022appearance}
Nicholas~P Baskerville, Diego Granziol, and Jonathan~P Keating.
\newblock Appearance of random matrix theory in deep learning.
\newblock {\em Physica A: Statistical Mechanics and its Applications},
  590:126742, 2022.

\bibitem[BGN11]{benaych2011eigenvalues}
Florent Benaych-Georges and Raj~Rao Nadakuditi.
\newblock The eigenvalues and eigenvectors of finite, low rank perturbations of
  large random matrices.
\newblock {\em Advances in Mathematics}, 227(1):494--521, 2011.

\bibitem[BHY19]{bauerschmidt2019local}
Roland Bauerschmidt, Jiaoyang Huang, and Horng-Tzer Yau.
\newblock Local kesten--mckay law for random regular graphs.
\newblock {\em Communications in Mathematical Physics}, 369(2):523--636, 2019.

\bibitem[BJSG{\etalchar{+}}19]{baity2019comparing}
Marco Baity-Jesi, Levent Sagun, Mario Geiger, Stefano Spigler, G{\'e}rard~Ben
  Arous, Chiara Cammarota, Yann LeCun, Matthieu Wyart, and Giulio Biroli.
\newblock Comparing dynamics: Deep neural networks versus glassy systems.
\newblock {\em Journal of Statistical Mechanics: Theory and Experiment},
  2019(12):124013, 2019.

\bibitem[BKMN21]{baskerville2021loss}
Nicholas~P Baskerville, Jonathan~P Keating, Francesco Mezzadri, and Joseph
  Najnudel.
\newblock The loss surfaces of neural networks with general activation
  functions.
\newblock {\em Journal of Statistical Mechanics: Theory and Experiment},
  2021(6):064001, 2021.

\bibitem[BKMN22]{baskerville2022spin}
Nicholas~P Baskerville, Jonathan~P Keating, Francesco Mezzadri, and Joseph
  Najnudel.
\newblock A spin glass model for the loss surfaces of generative adversarial
  networks.
\newblock {\em Journal of Statistical Physics}, 186(2):1--45, 2022.

\bibitem[BL21]{benigni2021fluctuations}
Lucas Benigni and Patrick Lopatto.
\newblock Fluctuations in local quantum unique ergodicity for generalized
  wigner matrices.
\newblock {\em arXiv preprint arXiv:2103.12013}, 2021.

\bibitem[BL22]{benigni2022optimal}
Lucas Benigni and Patrick Lopatto.
\newblock Optimal delocalization for generalized wigner matrices.
\newblock {\em Advances in Mathematics}, 396:108109, 2022.

\bibitem[BP19]{benigni2019eigenvalue}
Lucas Benigni and Sandrine P{\'e}ch{\'e}.
\newblock Eigenvalue distribution of nonlinear models of random matrices.
\newblock {\em arXiv preprint arXiv:1904.03090}, 2019.

\bibitem[BY17]{bourgade2017eigenvector}
Paul Bourgade and H-T Yau.
\newblock The eigenvector moment flow and local quantum unique ergodicity.
\newblock {\em Communications in Mathematical Physics}, 350(1):231--278, 2017.

\bibitem[CDM16]{capitaine2016spectrum}
Mireille Capitaine and Catherine Donati-Martin.
\newblock Spectrum of deformed random matrices and free probability.
\newblock {\em arXiv preprint arXiv:1607.05560}, 2016.

\bibitem[CHM{\etalchar{+}}15]{choromanska2015loss}
Anna Choromanska, Mikael Henaff, Michael Mathieu, G{\'e}rard~Ben Arous, and
  Yann LeCun.
\newblock The loss surfaces of multilayer networks.
\newblock In {\em Artificial intelligence and statistics}, pages 192--204.
  PMLR, 2015.

\bibitem[EKS19]{erdHos2019random}
L{\'a}szl{\'o} Erd{\H{o}}s, Torben Kr{\"u}ger, and Dominik Schr{\"o}der.
\newblock Random matrices with slow correlation decay.
\newblock In {\em Forum of Mathematics, Sigma}, volume~7. Cambridge University
  Press, 2019.

\bibitem[ES17]{erdHos2017universality}
L{\'a}szl{\'o} Erd{\H{o}}s and Kevin Schnelli.
\newblock Universality for random matrix flows with time-dependent density.
\newblock In {\em Annales de l'Institut Henri Poincar{\'e}, Probabilit{\'e}s et
  Statistiques}, volume~53, pages 1606--1656. Institut Henri Poincar{\'e},
  2017.

\bibitem[EY12]{erdHos2012universality}
L{\'a}szl{\'o} Erd{\H{o}}s and Horng-Tzer Yau.
\newblock Universality of local spectral statistics of random matrices.
\newblock {\em Bulletin of the American Mathematical Society}, 49(3):377--414,
  2012.

\bibitem[EY17]{erdos2017dynamical}
L{\'a}szl{\'o} Erdos and Horng-Tzer Yau.
\newblock A dynamical approach to random matrix theory.
\newblock {\em Courant Lecture Notes in Mathematics}, 28, 2017.

\bibitem[EYY12]{erdHos2012bulk}
L{\'a}szl{\'o} Erd{\H{o}}s, Horng-Tzer Yau, and Jun Yin.
\newblock Bulk universality for generalized wigner matrices.
\newblock {\em Probability Theory and Related Fields}, 154(1-2):341--407, 2012.

\bibitem[Fyo04]{fyodorov2004complexity}
Yan~V Fyodorov.
\newblock Complexity of random energy landscapes, glass transition, and
  absolute value of the spectral determinant of random matrices.
\newblock {\em Physical review letters}, 92(24):240601, 2004.

\bibitem[Fyo05]{fyodorov2005counting}
Yan~V Fyodorov.
\newblock Counting stationary points of random landscapes as a random matrix
  problem.
\newblock {\em Acta Physica Polonica B}, 36:2699--2707, 2005.

\bibitem[Gar88]{gardner1988space}
Elizabeth Gardner.
\newblock The space of interactions in neural network models.
\newblock {\em Journal of physics A: Mathematical and general}, 21(1):257,
  1988.

\bibitem[GB22]{https://doi.org/10.48550/arxiv.2011.08181}
Diego Granziol and Nicholas Baskerville.
\newblock A random matrix theory approach to damping in deep learning.
\newblock {\em arXiv preprint arXiv:2011.08181}, 2022.

\bibitem[GBW{\etalchar{+}}21]{ia}
Diego Granziol, Nicholas Baskerville, Xingchen Wan, Samuel Albanie, and Stephen
  Roberts.
\newblock Iterative averaging in the quest for best test error, 2021.

\bibitem[GPAM{\etalchar{+}}14]{goodfellow2014generative}
Ian Goodfellow, Jean Pouget-Abadie, Mehdi Mirza, Bing Xu, David Warde-Farley,
  Sherjil Ozair, Aaron Courville, and Yoshua Bengio.
\newblock Generative adversarial nets.
\newblock {\em Advances in neural information processing systems}, 27, 2014.

\bibitem[GPW{\etalchar{+}}18]{gardner2018gpytorch}
Jacob Gardner, Geoff Pleiss, Kilian~Q Weinberger, David Bindel, and Andrew~G
  Wilson.
\newblock Gpytorch: Blackbox matrix-matrix {G}aussian process inference with
  {GPU} acceleration.
\newblock In {\em Advances in Neural Information Processing Systems}, pages
  7576--7586, 2018.

\bibitem[Gra20]{granziol2020beyond}
Diego Granziol.
\newblock Beyond random matrix theory for deep networks.
\newblock {\em arXiv preprint arXiv:2006.07721}, 2020.

\bibitem[GWG19]{granziol2019deep}
Diego Granziol, Xingchen Wan, and Timur Garipov.
\newblock Deep curvature suite.
\newblock {\em arXiv preprint arXiv:1912.09656}, 2019.

\bibitem[GZR20]{granziol2020learning}
Diego Granziol, Stefan Zohren, and Stephen Roberts.
\newblock Learning rates as a function of batch size: A random matrix theory
  approach to neural network training.
\newblock {\em arXiv preprint arXiv:2006.09092}, 2020.

\bibitem[HSS08]{SciPyProceedings_11}
Aric~A. Hagberg, Daniel~A. Schult, and Pieter~J. Swart.
\newblock Exploring network structure, dynamics, and function using networkx.
\newblock In Ga\"el Varoquaux, Travis Vaught, and Jarrod Millman, editors, {\em
  Proceedings of the 7th Python in Science Conference}, pages 11 -- 15,
  Pasadena, CA USA, 2008.

\bibitem[HZRS16]{he2016deep}
Kaiming He, Xiangyu Zhang, Shaoqing Ren, and Jian Sun.
\newblock Deep residual learning for image recognition.
\newblock In {\em Proceedings of the IEEE conference on computer vision and
  pattern recognition}, pages 770--778, 2016.

\bibitem[KB14]{kingma2014adam}
Diederik~P Kingma and Jimmy Ba.
\newblock Adam: A method for stochastic optimization.
\newblock {\em arXiv preprint arXiv:1412.6980}, 2014.

\bibitem[KV03]{kim2003generating}
Jeong~Han Kim and Van~H Vu.
\newblock Generating random regular graphs.
\newblock In {\em Proceedings of the thirty-fifth annual ACM symposium on
  Theory of computing}, pages 213--222, 2003.

\bibitem[KY17]{knowles2017anisotropic}
Antti Knowles and Jun Yin.
\newblock Anisotropic local laws for random matrices.
\newblock {\em Probability Theory and Related Fields}, 169(1):257--352, 2017.

\bibitem[LXT{\etalchar{+}}18]{li2018visualizing}
Hao Li, Zheng Xu, Gavin Taylor, Christoph Studer, and Tom Goldstein.
\newblock Visualizing the loss landscape of neural nets.
\newblock {\em Advances in neural information processing systems}, 31, 2018.

\bibitem[Mar16]{martens2016second}
James Martens.
\newblock {\em Second-order optimization for neural networks}.
\newblock University of Toronto (Canada), 2016.

\bibitem[MBAB20]{maillard2019landscape}
Antoine Maillard, G\'erard Ben~Arous, and Giulio Biroli.
\newblock Landscape complexity for the empirical risk of generalized linear
  models.
\newblock In Jianfeng Lu and Rachel Ward, editors, {\em Proceedings of The
  First Mathematical and Scientific Machine Learning Conference}, volume 107 of
  {\em Proceedings of Machine Learning Research}, pages 287--327, Princeton
  University, Princeton, NJ, USA, 7 2020. PMLR.

\bibitem[MBC{\etalchar{+}}19]{mannelli2019afraid}
Stefano~Sarao Mannelli, Giulio Biroli, Chiara Cammarota, Florent Krzakala, and
  Lenka Zdeborov{\'a}.
\newblock Who is afraid of big bad minima? analysis of gradient-flow in spiked
  matrix-tensor models.
\newblock In {\em Advances in Neural Information Processing Systems}, pages
  8676--8686, 2019.

\bibitem[Meh04]{mehta2004random}
Madan~Lal Mehta.
\newblock {\em Random matrices}.
\newblock Elsevier, 2004.

\bibitem[MPV87]{mezard1987spin}
Marc M{\'e}zard, Giorgio Parisi, and Miguel Virasoro.
\newblock {\em Spin glass theory and beyond: An Introduction to the Replica
  Method and Its Applications}, volume~9.
\newblock World Scientific Publishing Company, 1987.

\bibitem[MS06]{meurant2006lanczos}
G{\'e}rard Meurant and Zden{\v{e}}k Strako{\v{s}}.
\newblock The {L}anczos and conjugate gradient algorithms in finite precision
  arithmetic.
\newblock {\em Acta Numerica}, 15:471--542, 2006.

\bibitem[Pap18]{papyan2018full}
Vardan Papyan.
\newblock The full spectrum of deepnet hessians at scale: Dynamics with sgd
  training and sample size.
\newblock {\em arXiv preprint arXiv:1811.07062}, 2018.

\bibitem[Pas20]{pastur2020random}
Leonid Pastur.
\newblock On random matrices arising in deep neural networks. gaussian case.
\newblock {\em arXiv preprint arXiv:2001.06188}, 2020.

\bibitem[Pea94]{pearlmutter1994fast}
Barak~A Pearlmutter.
\newblock Fast exact multiplication by the {H}essian.
\newblock {\em Neural computation}, 6(1):147--160, 1994.

\bibitem[PGC{\etalchar{+}}17]{paszke2017automatic}
Adam Paszke, Sam Gross, Soumith Chintala, Gregory Chanan, Edward Yang, Zachary
  DeVito, Zeming Lin, Alban Desmaison, Luca Antiga, and Adam Lerer.
\newblock Automatic differentiation in pytorch.
\newblock 2017.

\bibitem[PS20]{pastur2020randomiid}
L~Pastur and V~Slavin.
\newblock On random matrices arising in deep neural networks: General iid case.
\newblock {\em arXiv preprint arXiv:2011.11439}, 2020.

\bibitem[PSG18]{pennington2018emergence}
Jeffrey Pennington, Samuel Schoenholz, and Surya Ganguli.
\newblock The emergence of spectral universality in deep networks.
\newblock In {\em International Conference on Artificial Intelligence and
  Statistics}, pages 1924--1932. PMLR, 2018.

\bibitem[PW17]{pennington2017nonlinear}
Jeffrey Pennington and Pratik Worah.
\newblock Nonlinear random matrix theory for deep learning.
\newblock {\em Advances in neural information processing systems}, 30, 2017.

\bibitem[RABC19]{ros2019complex}
Valentina Ros, Gerard~Ben Arous, Giulio Biroli, and Chiara Cammarota.
\newblock Complex energy landscapes in spiked-tensor and simple glassy models:
  Ruggedness, arrangements of local minima, and phase transitions.
\newblock {\em Physical Review X}, 9(1):011003, 2019.

\bibitem[SGAL14]{sagun2014explorations}
Levent Sagun, V~Ugur Guney, Gerard~Ben Arous, and Yann LeCun.
\newblock Explorations on high dimensional landscapes.
\newblock {\em arXiv preprint arXiv:1412.6615}, 2014.

\bibitem[SW99]{steger1999generating}
Angelika Steger and Nicholas~C Wormald.
\newblock Generating random regular graphs quickly.
\newblock {\em Combinatorics, Probability and Computing}, 8(4):377--396, 1999.

\bibitem[SZ14]{simonyan2014very}
Karen Simonyan and Andrew Zisserman.
\newblock Very deep convolutional networks for large-scale image recognition.
\newblock {\em arXiv preprint arXiv:1409.1556}, 2014.

\bibitem[Tao12]{tao2012topics}
Terence Tao.
\newblock {\em Topics in random matrix theory}, volume 132.
\newblock American Mathematical Soc., 2012.

\bibitem[TSR22]{thamm2022random}
Matthias Thamm, Max Staats, and Bernd Rosenow.
\newblock Random matrix analysis of deep neural network weight matrices.
\newblock {\em arXiv preprint arXiv:2203.14661}, 2022.

\bibitem[Unt19]{unterberger2019global}
Jeremie Unterberger.
\newblock Global fluctuations for 1d log-gas dynamics. covariance kernel and
  support.
\newblock {\em Electronic Journal of Probability}, 24, 2019.

\bibitem[WHS22]{wei2022more}
Alexander Wei, Wei Hu, and Jacob Steinhardt.
\newblock More than a toy: Random matrix models predict how real-world neural
  representations generalize.
\newblock {\em arXiv preprint arXiv:2203.06176}, 2022.

\bibitem[ZK16]{zagoruyko2016wide}
Sergey Zagoruyko and Nikos Komodakis.
\newblock Wide residual networks.
\newblock {\em arXiv preprint arXiv:1605.07146}, 2016.

\end{thebibliography}

\end{document}